\documentclass[a4paper, UKenglish, cleveref, autoref, thm-restate, pdfa]{lipics-v2021}

\pdfoutput=1 
\hideLIPIcs 
\nolinenumbers 

\usepackage{csquotes}
\usepackage{mathtools}
\usepackage{stmaryrd}
\usepackage{tikz}
\usepackage{pifont}
\usepackage{colortbl}
\usepackage{hyperref}

\usetikzlibrary{arrows, arrows.meta}
\usetikzlibrary{backgrounds}
\usetikzlibrary{calc}
\usetikzlibrary{decorations.markings}
\usetikzlibrary{decorations.pathmorphing}
\usetikzlibrary{decorations.pathreplacing}
\usetikzlibrary{fadings}
\usetikzlibrary{fit}
\usetikzlibrary{matrix}
\usetikzlibrary{patterns}
\usetikzlibrary{positioning}
\usetikzlibrary{shapes}
\usetikzlibrary{tikzmark}

\theoremstyle{Definition}
\newtheorem{Definition}[theorem]{Definition}

\numberwithin{theorem}{section}
\numberwithin{corollary}{section}
\numberwithin{proposition}{section}
\numberwithin{lemma}{section}
\numberwithin{Definition}{section}
\numberwithin{claim}{section}

\tikzset{
	graph edge/.style = {draw, thick},
}

\newcommand{\cutout}[1]{} 

\newcommand{\E}{\exists}
\newcommand{\A}{\forall}

\renewcommand{\AA}{{\mathfrak A}}

\newcommand{\Trop}{\mathbb{T}}
\newcommand{\Vit}{\mathbb{V}}
\newcommand{\Nat}{\mathbb{N}}
\newcommand{\Ninf}{\mathbb{N}^\infty}
\newcommand{\Bool}{\mathbb{B}}
\newcommand{\Fuzzy}{\mathbb{F}}
\newcommand{\Lukas}{\mathbb{L}}
\newcommand{\Doubt}{\mathbb{D}}

\newcommand{\Sorb}{\mathbb{S}}

\newcommand{\Tt}{\ensuremath{\mathcal{T}}}
\newcommand{\calS}{\mathcal{S}}
\newcommand{\Semiring}{\ensuremath{\calS}}
\newcommand{\Semi}{\Semiring}
\newcommand{\calL}{\mathcal{L}}
\newcommand{\calT}{\mathcal{T}}

\renewcommand{\phi}{\varphi}
\renewcommand{\theta}{\vartheta}
\renewcommand{\epsilon}{\varepsilon}

\newcommand*{\N}{\ensuremath{\mathbb{N}}}

\newcommand*{\R}{\ensuremath{\mathbb{R}}}

\newcommand*{\llb}{\ensuremath{\llbracket}}
\newcommand*{\rrb}{\ensuremath{\rrbracket}}
\renewcommand{\AA}{{\mathfrak A}}
\newcommand{\BB}{{\mathfrak B}}
\newcommand{\exneq}{\ensuremath \exists^{\neq}}
\newcommand{\foraneq}{\ensuremath \forall^{\neq}}

\newcommand{\ra}{\rightarrow}

\DeclareMathOperator{\supp}{\mathsf{supp}}

\DeclareMathOperator{\FO}{FO}
\DeclareMathOperator{\qr}{qr}
\DeclareMathOperator{\Lit}{Lit}
\DeclareMathOperator{\img}{\mathsf{img}}

\newcommand{\Inf}{\bigsqcap}
\newcommand{\Sup}{\bigsqcup}
\newcommand{\meet}{\sqcap}
\newcommand{\join}{\sqcup}
\newcommand{\bbN}{\mathbb{N}}

\newcommand*{\ext}[1]{\llbracket #1 \rrbracket}

\colorlet{yes}{green!30}
\colorlet{no}{red!30}

\title{Preservation Theorems in Semiring Semantics}

\author{Sophie Brinke}{RWTH Aachen University, Germany}{brinke@logic.rwth-aachen.de}{%
https://orcid.org/0009-0005-2865-0069}{}

\author{Anuj Dawar}{University of Cambridge, UK}{Anuj.Dawar@cl.cam.ac.uk}{https://orcid.org/0000-0003-4764-8894}{}

\author{Erich Grädel}{RWTH Aachen University, Germany}{graedel@logic.rwth-aachen.de}{%
https://orcid.org/0000-0002-8950-9991}{}
\author{Benedikt Pago}{University of Cambridge, UK}{benedikt.pago@cl.cam.ac.uk}{https://orcid.org/0000-0001-6377-1230}{}

\authorrunning{S. Brinke, A. Dawar, E. Grädel, B. Pago}
\Copyright{Sophie Brinke, Anuj Dawar, Erich Grädel, and Benedikt Pago}


\category{}
\relatedversion{}

\begin{document}

\maketitle

\begin{abstract}
We study the status of classical model-theoretic preservation theorems such as 
the {\L}o\'s-Tarski theorem and the homomorphism preservation theorem in the context of
semiring semantics. 

Semiring semantics has its origins in the provenance analysis of database queries but has been extended to a systematic way of evaluating logical statements to values in a commutative semiring. Depending on the underlying semiring, this allows us to track descriptions of the atomic facts that are responsible for the truth of a statement or practical information about the evaluation such as costs or confidence. The systematic development of semiring semantics for first-order logic and other logical systems raises the question to what extent classical model-theoretic results can be generalised to 
this setting and how such results depend on the underlying semiring. 

The definitions of semantic properties such as preservation under extensions, substructures, or homomorphisms naturally generalise to the setting of semiring semantics. However, the
status of the corresponding preservation theorem strongly depends on the 
algebraic properties of the particular semirings. We prove that these preservation theorems
do indeed hold for all lattice semirings (a quite large class, encompassing practically relevant 
semirings and in particular all min-max semirings).
The proofs combine adaptations of the classical compactness and amalgamation methods 
with specific reduction methods for logical entailment that have been developed in semiring semantics.
On the other side, variants of the existential preservation theorem fail 
for many other semirings, including the tropical semiring, the Viterbi semiring, the {\L}ukasiewicz semiring, and the natural semirings $\N$ and $\N^\infty$.
Surprisingly, the existential preservation theorem does hold for \emph{finite} interpretations in a number of semirings, including the three-element min-max semiring, which extends the Boolean by just a single additional truth value. Thus, the situation for these semirings is in sharp contrast to the Boolean case, where the {\L}o\'s-Tarski theorem holds in general, but not in the finite.

\end{abstract}

\newpage

\section{Introduction}

\noindent{\bf Preservation theorems. } In model theory, preservation theorems 
typically state that the first-order sentences satisfying a specific semantic property
are precisely those that are equivalent to sentences of a particular syntactic
form. Well-known examples include the {\L}o\'s-Tarski theorem saying that 
a first-order sentence is preserved under extensions of its models if, and only if,
it is equivalent to an existential sentence, and the homomorphism preservation theorem which states that the first-order sentences preserved under homomorphisms are precisely those that are equivalent to existential-positive sentences (unions of conjunctive queries).
Preservation theorems matter because they tell us which syntactic forms capture which structural invariance properties, and this drives applications in databases, 
descriptive complexity, the simplification of axiomatisations, quantifier elimination, and the study of expressive power of fragments of first-order logic.

The prominence of preservation theorems in classical model theory has motivated persistent investigations of their status in other model-theoretic contexts, in particular in finite model theory and, more recently, also in many-valued logics. In finite model theory, it has been known for a long time that most of the classical preservation theorems fail when restricted to finite structures. In particular this concerns the {\L}o\'s-Tarski theorem for which Tait~\cite{Tait59} and Gurevich~\cite{Gurevich84} provide examples of  first-order sentences whose finite models are closed under extensions but which are not equivalent, over finite structures, to any existential sentence. This failure has been strengthened by Dawar and Sankaran~\cite{DawarSan21} who construct, for every~$n$, sentences whose finite models are closed under extensions, but which are not equivalent in the finite to any $\Sigma_n$-sentence.
Among the few positive results on preservation theorems in finite model theory, the most significant one, due to Rossman \cite{Rossman08},
establishes that the homomorphism preservation theorem does indeed hold on finite structures. There are many further results on preservation theorems in finite model theory, see~\cite{Rosen02}
for a survey and~\cite{Lopez-thesis} for a more recent treatment.

In many-valued logics, Badia et al.~\cite{BadiaCosDelNog19} generalised the {\L}o\'s-Tarski theorem and also the
Chang-{\L}o\'s-Suszko preservation theorem for $\Pi_2$-sentences under unions of chains to the
specific setting of fixed finite MTL-chains (linearly ordered algebraic models of monoidal t-norm–based logic).
Dellunde and Vidal \cite{DellundeVid19} established a variant of the homomorphism preservation theorem
in the same setting and Carr \cite{Carr26} provided a variant of this for finite many-valued structures, based on an extension of Rossman's proof.

\medskip\noindent{\bf Semiring semantics. } In this paper, we extend the study of preservation theorems to the setting of semiring semantics. This approach, 
which has its origins in the provenance analysis of database queries \cite{GreenKarTan07}, is based on the idea of evaluating logical statements not just to true or false, but more generally to values in some commutative semiring $\Semi = (S, +, \cdot, 0, 1)$. Atomic facts are annotated by values from $S$ which are then propagated through a 
database query or a logical formula, keeping track whether information is used alternatively or jointly. 
Depending on the chosen semiring, provenance valuations give practical information about a statement, for instance concerning the \emph{confidence} that we may 
have in its truth, the \emph{cost} of its evaluation, the required \emph{clearance level} for the access to not freely available data,  
the \emph{number of successful evaluation strategies}, and so on. Beyond such provenance evaluations in specific \emph{application semirings}, 
more precise information is obtained by evaluations in \emph{provenance semirings}
of polynomials or formal power series, which permit us to \emph{track} which
atomic facts are used (and how often) to compute the answer to the query. See  \cite{GreenTan17,Glavic21} for surveys on semiring provenance in databases.

The detailed information on a logical statement, as provided by semiring provenance, is of interest 
not just for databases, but also in many other areas of logic in computer science. 
There is thus ample motivation to extend the approach of semiring provenance beyond database queries to a general semiring semantics 
for logical systems, in particular for full first-order logic.
Further, semiring valuations have also been successfully applied to the strategy analysis for various forms of finite and infinite games \cite{GraedelTan20, GraedelLucNaa21b, Mrkonjic25}.
This poses a number of mathematical challenges that have been addressed in recent research, and it
raises the question to what extent classical logical methods and results also hold in this more general context.
Research on such questions has included the  study of  elementary equivalence and the axiomatisability of finite semiring interpretations \cite{GraedelMrk21},
of the equivalence of the relational calculus with relational algebra (Codd's Theorem) \cite{BadiaKolNog25}, 
0-1 laws \cite{GraedelHelNaaWil22},  Ehrenfeucht--Fra\"{\i}ss\'e games \cite{BrinkeGraMrk24}, the locality theorems by Hanf and Gaifman \cite{BiziereGraNaa23}, compactness \cite{BrinkeDawGraMrkNaa26},
and the interplay between local consistency and global consistency for relations over semirings \cite{AtseriasKol23, AtseriasKol24}.
The study of preservation theorems in semiring semantics extends this research to a further important model-theoretic topic.

From the study of model-theoretic methods and results in the setting of semiring semantics, the following general picture has emerged.
\begin{itemize}
    \item For most of the central model-theoretic notions (such as elementary equivalence, homomorphisms, evaluation and comparison games, locality, compactness, random structures etc.), the \emph{definitions}
    naturally generalise to the semiring setting. In some cases, the generalisations require that semirings with specific algebraic properties are considered (for instance absorptive semirings, or semirings with infinitary operations).
    \item Sometimes, variants of such definitions which are clearly equivalent in the classical Boolean case can become different, non-equivalent notions in semiring semantics. An important example is compactness, which can be formulated either in terms of satisfiability or in terms of entailment, and in the classical setting either one of the formulations of the compactness theorem obviously implies the other. In semiring semantics this is not the case. There, compactness in terms of entailment is a stronger notion than compactness in terms of satisfiability, and there are important semirings, such as the tropical semiring where the two notions can be separated~\cite{BrinkeDawGraMrkNaa26}.
    \item While the \emph{questions} about the status of classical logical theorems in  semiring semantics arise very naturally, their \emph{answers} (and the proofs of these) become much more complicated and often strongly depend on algebraic properties of the underlying semirings. Moreover, they often require the development of new mathematical methods. 
\end{itemize}

\noindent{\bf Our contributions. }
For preservation properties in semiring semantics,
we make similar observations.
The notions of preservation under extensions, substructures, 
or homomorphisms generalise very naturally to semiring semantics,
and so do the straightforward directions (from syntactic forms to semantic properties)
of the preservation theorems that we consider: Existential sentences are preserved under extensions, universal sentences are preserved downwards and existential-positive sentences are preserved under homomorphisms. We study the question for which semirings also the `difficult' direction (from semantic properties to syntactic forms) of the preservation theorems holds, both in the case of semiring interpretations over arbitrary domains and those just over finite domains.
More precisely, our questions concerning, say,  the {\L}o\'s-Tarski theorem can be stated as follows:
Fix a semiring $\Semi$. Is it the case that for every first-order sentence $\psi$ which is preserved under extensions of (finite) $\Semi$-interpretations there exists a purely existential sentence
which is equivalent to $\psi$ under all (finite) $\Semi$-interpretations? Can we identify algebraic notions that imply positive or negative answers to this question?

The questions for universal preservation and preservation under
homomorphism are analogous.
Notice that the classical duality between existential and universal preservation breaks down. While the classical {\L}o\'s-Tarski theorem can be equivalently formulated in terms of
existential or universal preservation, these are separate
issues in semiring semantics.
We also note that preservation theorems for semirings 
neither entail the classical Boolean variants, nor are they entailed by them. This is because on the semantic side, preservation properties for a semiring $\Semi$ are stronger than their Boolean analogues,
but on the syntactic side also $\Semi$-equivalence is a stronger notion than classical logical equivalence.

\medskip
Our main results can be summarised as follows:
\begin{itemize}
    \item The preservation theorems for extensions, subinterpretations and homomorphisms are indeed true for
    all \emph{lattice semirings} $(S, \sqcup, \sqcap, 0, 1)$. These are semirings induced by
a completely distributive complete lattice $(S,\leq)$ (i.e. a partial order closed under suprema and infima in which infima distributive over suprema), with suprema and infima as semiring operations;
they encompass in particular all min-max semirings, such as the fuzzy semiring $\Fuzzy=([0,1],\max,\min,0,1)$, and the \emph{security semiring} for reasoning about access control \cite{FosterGreTan08}.
The proofs combine adaptations of the classical compactness and amalgamation methods to the specific three-element min-max semiring $\Semi_3$ with a reduction method from \cite{BrinkeDawGraMrkNaa26},
    reducing entailments for arbitrary lattice semirings to $\Semi_3$-entailments.

    \item Variants of the existential preservation theorem fail 
    for many other semirings, including the tropical semiring, the Viterbi semiring, the {\L}ukasiewicz semiring and the natural semirings $\N$ and $\N^\infty$.

    \item Surprisingly, the existential preservation theorem does, however, hold for \emph{finite} interpretations in the Viterbi semiring, the tropical semiring, the {\L}ukasiewicz semiring, and every lattice semiring with at least three elements. 
    Thus the situation for these semirings is in sharp contrast to the
    classical Boolean case, where the {\L}o\'s-Tarski theorem holds in general, but not in the finite. To prove this result, we develop a novel technique that allows us to translate evaluation strategies for sentences between models of different finite cardinalities (\cref{lem-relationship-between-polynomials}).
\end{itemize}

In particular, the latter result raises the intriguing question whether the failure of extension preservation in the finite is perhaps merely a sporadic phenomenon that is unique to the Boolean semiring.
Thus, our study provides the foundation for a new line of future research: \emph{For which other (classes of) semirings does the extension preservation theorem survive in the finite?}

\section{Semiring Semantics for First-order Logic}

We briefly describe the foundations of semiring semantics for first-order logic and refer to~\cite{GraedelTan24} for further details.

\begin{Definition}
A \emph{commutative monoid} is an algebraic structure $(M,+,0)$ where $+$ is a binary commutative and associative operation with $0 \in M$ as its neutral element.
A \emph{commutative semiring} is a structure $\Semi = (S, +, \cdot, 0, 1)$ with $0 \neq 1$ such that $(S, +, 0)$ and $(S, \cdot, 1)$ are commutative monoids, $\cdot$ distributes over $+$, and $0 \cdot s = s \cdot 0 = 0$.
\end{Definition}

In the following, we assume that all semirings are commutative and \emph{naturally ordered}, that is, $s\leq t \Leftrightarrow \E r (s+r=t)$ defines a partial order. Addition and multiplication are monotone with respect to the natural order.
A semiring is \emph{absorptive} if $s+st=s$ for all $s,t\in S$, which is equivalent
to the property that multiplication is decreasing with respect to the natural order. Absorptive semirings are particularly important in semiring semantics, because they preserve, to some extent, the common logical dualities.
Addition in absorptive semirings is \emph{idempotent} (i.e., $s+s=s$ for each $s \in S$), so sums coincide with suprema.
Beyond the Boolean semiring $\Bool = (\{0,1\}, \vee, \wedge, 0, 1)$, there are 
many other absorptive semirings that provide useful information about the evaluation of a formula.

\begin{itemize}
\item A totally ordered set $(S, \leq)$ with least element $s$ and greatest element $t$ induces the \emph{min-max semiring} $(S, \max, \min, s, t)$. 
\item A more general class  is the class of  \emph{lattice semirings} $(S,\join,\meet,s,t)$ induced by
a bounded distributive lattice $(S,\leq)$. In fact, every absorptive semiring with idempotent multiplication is a lattice semiring.
	\item The \emph{tropical semiring} $\Trop= (\R^\infty_+ , \min, +, \infty, 0)$ is used for cost analysis.
	\item To reason about confidence, we may use the \emph{Viterbi semiring} $\Vit = ([0,1]_\R, \max, \cdot, 0, 1)$ or the \emph{\L ukasiewicz semiring} $\Lukas = ([0,1]_\R, \max, \odot, 0, 1)$ where $s \odot t \coloneqq\max(s + t - 1, 0)$.
	\item The semiring $\Doubt = ([0, 1]_\R, \min, \oplus, 1, 0)$ with $s \oplus t \coloneqq \min(s+t,1)$ models levels of doubt.
    \item The most important non-absorptive semirings are the natural semiring $(\N,+,\cdot,0,1)$ and its extension to
    $\Ninf=\N\cup\{\infty\}$.
	\item \emph{Provenance semirings} allow us to track the atomic facts that are responsible for the truth of a formula. In the non-absorptive case, the most general 
	provenance semirings are $\N[X]$, the semirings of polynomials over a finite set $X$ of indeterminates and coefficients from $\N$.  In the absorptive setting 
	 the important examples are the semirings $\Sorb(X)$ and $\Sorb^\infty (X)$ of (generalised) absorptive polynomials.
     See \cite{GraedelTan24,Naaf24} for more details. 
\end{itemize}

To define semiring semantics also for infinite universes, and to extend it to valuations of infinite collections $\Phi$ of first-order sentences, 
the semirings need to be equipped with infinitary summation and product operators. Algebraic foundations, mathematical properties
and provenance valuations of such \emph{infinitary semirings} have been studied in \cite{BrinkeGraMrkNaa24}. We will not go into details here and tacitly assume that infinitary sums and products are available and have the desired associativity and commutativity properties. In the specific semirings that we consider here, the definitions of these operations are straightforward.

To evaluate first-order formulae in semirings, classical structures are generalised to semiring interpretations that map the atomic facts and their negations to semiring values. While $0$ represents falsity, every non-zero element provides an annotation of truth.

\begin{Definition}[Semiring interpretations]
\label{def:semiring-interpretation}
	Let $A$ be a (finite or infinite) universe, let $\tau$ be a relational vocabulary, and let $\Semi$ be a semiring. We denote by $\Lit_A (\tau)$ the set of $\tau$-literals of the form $R \bar a$ or $\lnot R \bar a$ that are instantiated with tuples from $A$. Analogously, $\operatorname{Atoms}_A (\tau)$ denotes the set of $\tau$-atoms $R \bar a$ over $A$.
	An \emph{$\Semi$-interpretation} (for $A$ and $\tau$) is a function $\pi \colon \Lit_A(\tau)\ra \Semi$.
	We say that $\pi$ is 
	\emph{model-defining} if for every literal $L\in\Lit_A(\tau)$ precisely one of the values $\pi(L)$ and $\pi(\neg L)$ is 0. 
    Every model-defining $\Semi$-interpretation defines a $\tau$-model $\AA_\pi$ where $\AA_\pi \models L$ if $\pi (L) \neq 0$ for $L \in \Lit_A (\tau)$.
\end{Definition}

We only consider model-defining interpretations in this paper.
Semiring semantics lifts each $\Semi$-interpretation for a universe $A$ and vocabulary $\tau$ to a mapping $\FO_A (\tau) \to \Semi$ where $\FO_A (\tau)$ is the set of first-order formulae instantiated with elements from $A$. 
We assume that all formulae are written in negation normal form. This has emerged as the standard approach for dealing with negation in semiring semantics, and is explained in detail in \cite{GraedelTan24}.

\begin{Definition}[Semiring Semantics]
	For every $\phi \in \FO_A (\tau)$ in negation normal form, the semiring valuation $\pi\ext\phi$ is lifted beyond the literals inductively, by
	\begin{alignat*}{3}
		\pi \llb \psi \vee \theta \rrb &\coloneqq \pi \llb \psi \rrb +  \pi \llb  \theta \rrb, \hspace*{1cm} &  \pi \llb \forall x \psi (x) \rrb &\coloneqq \prod_{a \in A} \pi \llb \psi (a) \rrb
        , \\
		\pi \llb \psi  \wedge \theta \rrb &\coloneqq \pi \llb \psi  \rrb \cdot  \pi \llb  \theta \rrb, \hspace*{1cm}  & \pi \llb \exists x \psi (x) \rrb &\coloneqq \sum_{a \in A} \pi \llb \psi (a) \rrb.
	\end{alignat*}
	Equality atoms are interpreted by their Boolean truth value, that is, $\pi \llb a = a \rrb \coloneqq 1 $ and $\pi \llb a = b \rrb \coloneqq 0$ for $a \neq b$ (and analogously for inequalities). For finite or infinite sets $\Phi \subseteq \FO_A (\tau)$ we set $ \pi \ext \Phi \coloneqq  \prod_{\phi \in \Phi} \pi\ext{\phi}$.
\end{Definition}

\begin{Definition}[Entailment] Let $\Phi, \Psi \subseteq \FO$ be sets of first-order sentences and let $\Semi$ be an (absorptive) semiring. We write
	\begin{enumerate}[(1)]
		\item $\Phi\models_\Semi \Psi$ ($\Phi$ $\Semi$-entails $\Psi$) if  $\pi\ext{\Phi}\leq \pi\ext{\psi}$ for every model-defining $\Semi$-interpretation $\pi$, and
        \item $\Phi\equiv_\Semi \Psi$ ($\Phi$ is $\Semi$-equivalent to $\Psi$) if  $\pi\ext{\Phi} = \pi\ext{\Psi}$ for every model-defining $\Semi$-interpretation~$\pi$.
	\end{enumerate}
    If $\Phi$ or $\Psi$ is a singleton, we replace them with the corresponding formula in these notations.
\end{Definition}

\paragraph*{Notions of preservation in semiring semantics}

We next discuss how semantic preservation properties (upwards, downwards, and under homomorphisms) present themselves in the context of semiring semantics, and why they are implied by 
specific syntactic forms.

\begin{Definition}[Homomorphisms and Subinterpretations]
\begin{itemize}
\item A \emph{homomorphism} between $\Semi$-interpretations $\pi_A$ and $\pi_B$ is a mapping $g \colon A \to B$ such that for any $R \bar a \in \operatorname{Atoms}_A(\tau)$ it holds that $\sum_{\bar a': g (\bar a') = g(\bar a)} \pi_A (R \bar a') \leq \pi_B (R g(\bar a))$ (note that we only take into account the positive literals here).
Further, $g$ is a \emph{strong homomorphism} if $\pi_A (L \bar a) = \pi_B (L g(\bar a))$ for all $L\bar a \in \Lit_A(\tau)$. If $g$, in addition, is injective, we call it an \emph{embedding} and write $g \colon \pi_A \subseteq \pi_B$. Note that for idempotent semirings (where addition is the supremum with respect to the natural order), the inequality above reduces to $\pi_A (R \bar a) \leq \pi_B (R g(\bar a))$ for each $R \bar a \in \operatorname{Atoms}_A(\tau)$.
\item  $\pi_A$ is a \emph{subinterpretation} of $\pi_B$ (and $\pi_B$ is an \emph{extension} of $\pi_A$), denoted $\pi_A \subseteq \pi_B$, if $A \subseteq B$ and $\pi_A (L\bar a) = \pi_B (L\bar a)$ holds for all literals $L\bar a \in \Lit_A (\tau)$.
\item $\pi_A$ is an \emph{elementary subinterpretation} of $\pi_B$ (and $\pi_B$ is an \emph{elementary extension} of $\pi_A$), denoted $\pi_A \preceq \pi_B$, if $A \subseteq B$ and $\pi_A \llb \phi(\bar a) \rrb = \pi_B \llb \phi(\bar a) \rrb$ holds for all $\phi(\bar a) \in \FO_A(\tau)$.
\end{itemize}
\end{Definition}

As an example, consider the $\Vit$-interpretation $\pi_A$ with universe $A=\{a\}$ and vocabulary $\tau = \{R\}$, where $R$ is unary, where we have $\pi_A (Ra) = 1/2$ and $\pi_A (\lnot Ra)=0$. Extending $A$ by some element $b$ and putting $\pi_B (Rb) = 1/4$ while $\pi_B (\lnot Rb) = 0$ yields a $\Vit$-interpretation $\pi_B$, which is an non-elementary extension of $\pi_A$, as $\pi_A \llb \forall x Rx \rrb = 1/2 \neq 1/8 = \pi_B \llb \forall x Rx \rrb$. Further, $g \colon a \mapsto a, b \mapsto a$ is a homomorphism from $\pi_B$ to $\pi_A$, but not a strong one.

Classical preservation properties have the form that 
sentences which are true in a specific structure remain true
when we move to another structure, related to the first one
as an extension, a substructure, or by a homomorphism.
In the context of semiring interpretations we evaluate
formulae not just by true or false, but in a general semiring. The property that the formula remains true
is replaced by the stronger relationship that the formula evaluates to a greater or equal value (w.r.t.\ to the natural order on the semiring) when we move from one interpretation to a new one.

\begin{Definition}[Preservation] 
\begin{itemize}
    \item A sentence $\psi \in\FO (\tau)$ is \emph{preserved under homomorphisms in \Semi} if for all model-defining $\Semi$-interpretations $\pi_A, \pi_B$ with a homomorphism $g \colon A \to B$, we have that $\pi_A \llb \psi \rrb \leq \pi_B \llb \psi \rrb$.

    \item A sentence $\phi \in \FO (\tau)$ is \emph{preserved under subinterpretations in \Semi} if for all model-defining $\Semi$-interpretations $\pi_A \subseteq \pi_B$, we have that $\pi_A \llb \phi \rrb \geq \pi_B \llb \phi \rrb$. 
    \item Analogously, $\phi \in \FO (\tau)$ is \emph{preserved under extensions in \Semi} if $\pi_A \llb \phi \rrb \leq \pi_B \llb \phi \rrb$ is true for all model-defining $\Semi$-interpretations $\pi_A \subseteq \pi_B$.
       \end{itemize}
\end{Definition}

We observe that the directions from the syntactic form to the semantic preservation properties of the classical preservation theorems
also hold in this context.
Recall that
    \begin{itemize}
        \item $\Sigma_1^+$ denotes the set of \emph{existential-positive} sentences $\exists \bar x \phi (\bar x)$ where $\phi (\bar x)$ is quantifier- and negation-free,
        \item $\Sigma_1$ denotes the set of \emph{existential} sentences $\exists \bar x \phi (\bar x)$ where $\phi (\bar x)$ is quantifier-free, and
        \item $\Pi_1$ is the set of \emph{universal} sentences $\forall \bar x \phi (\bar x)$ where $\phi (\bar x)$ is quantifier-free.
    \end{itemize}

At this point, it becomes clear why our definition of homomorphism involves summation over the preimages of tuples. This accounts for the possibility that multiple tuples in the left-hand side interpretation $\pi_A$ are collapsed onto the same tuple in $\pi_B$. Our definition ensures that existential-positive sentences are indeed preserved under homomorphisms. For example, consider the sentence $\exists x \exists y E xy$. Now take for example $\pi_A$ to be a bipartite graph with edge weights from $\Semi$, and $\pi_B$ an $\Semi$-weighted graph that is just a single edge. 
Since every bipartite graph maps homomorphically to a single edge, $\pi_A$ should map homomorphically to $\pi_B$. But we have to insist that the value of the single edge in $\pi_B$ is at least the sum of its preimages, or else, we would not have $\pi_A \llb \exists x \exists y E xy \rrb \leq \pi_B \llb \exists x \exists y E xy \rrb$. Moreover, in defining homomorphisms we use the natural order rather than insisting on equalities, so that edges are preserved while non-edges need not be.

\begin{lemma} \label{lem:closure-under-and-or}
    If $\psi$ is a positive Boolean combination of (positive) existential sentences, there is some $\phi \in \Sigma_1$ ($\phi \in \Sigma_1^+$) such that $\psi \equiv_\Semi \phi$ for all additively idempotent semirings $\Semi$. If $\psi$ is a positive Boolean combination of universal sentences, there is some $\phi \in \Pi_1$ such that $\psi \equiv_\calL \phi$ for all lattice semirings $\calL$.
\end{lemma}

\begin{proof}
    Let $\Semi$ be a semiring and $\exists \bar x \theta (\bar x), \exists \bar y \theta^* (\bar y) \in \FO_A (\tau)$ where $\bar x$ and $\bar y$ are disjoint tuples of variables.
    By (1) commutativity, associativity, and idempotence of addition, and (2) distributivity, we have the logical equivalences
	$
	\exists \bar x \theta (\bar x) \circ \exists \bar y \theta^* (\bar y) \equiv_\Semi \exists \bar x \exists \bar y (\theta (\bar x) \circ  \theta^* (\bar y))
	$
    for $\circ \in \{\vee, \wedge\}$
    in any semiring $\Semi$ since for any $\Semi$-interpretation $\pi$ of universe $A$,
    \begin{align*}
	\pi \llb \exists \bar x \theta (\bar x) \vee \exists \bar y \theta^* (\bar y) \rrb &= \sum_{\bar a \subseteq A} \pi \llb \theta (\bar a) \rrb + \sum_{\bar a \subseteq A} \pi \llb \theta^* (\bar a) \rrb \\
    &\overset{(1)}= \sum_{\bar a \subseteq A}\sum_{\bar b \subseteq A} \pi \llb \theta (\bar a) \rrb + \pi \llb \theta^* (\bar b) \rrb
    = \pi \llb \exists \bar x \exists \bar y (\theta (\bar x) \vee  \theta^* (\bar y)) \rrb \text{ and } \\
    \pi \llb \exists \bar x \theta (\bar x) \wedge \exists \bar y \theta^* (\bar y) \rrb &= \sum_{\bar a \subseteq A} \pi \llb \theta (\bar a) \rrb \cdot \sum_{\bar a \subseteq A} \pi \llb \theta^* (\bar a) \rrb \\
    &\overset{(2)}= \sum_{\bar a \subseteq A}\sum_{\bar b \subseteq A} \pi \llb \theta (\bar a) \rrb \cdot \pi \llb \theta^* (\bar b) \rrb
    = \pi \llb \exists \bar x \exists \bar y (\theta (\bar x) \wedge  \theta^* (\bar y)) \rrb.
	\end{align*}
    Hence, the first claim follows by induction. 
    In order to prove the second statement in an analogous way, we do not just need commutativity, associativity, and idempotence of multiplication, but also the dual distributivity law $\prod_{i \in I} s_i + \prod_{j \in J} t_j = \prod_{i \in I} \prod_{j \in J} (s_i +t_j)$, where $(s_i)_{i \in I}, (t_j)_{j \in J} \subseteq \Semi$, which only holds in lattice semirings.
    In any lattice semiring $\calL$, however, we have the logical equivalences
    \[
	\forall \bar x \theta (\bar x) \vee \forall \bar y \theta^* (\bar y) \equiv_\calL \exists \bar x \forall \bar y (\theta (\bar x) \vee  \theta^* (\bar y)) \;\; \text{ and } \;\; \forall \bar x \theta (\bar x) \wedge \exists \bar y \theta^* (\bar y) \equiv_\calL \forall \bar x \forall \bar y (\theta (\bar x) \wedge  \theta^* (\bar y)),
	\]
    which inductively imply the second claim.
\end{proof}

By monotonicity of the semiring operations and the fact that addition is increasing with respect to the natural order, the implication from syntax to semantics follows readily by induction for $\Sigma_1^+$ and $\Sigma_1$. We only need to take some care in the case $\Pi_1$, where we need multiplication to be decreasing. This is true exactly for absorptive semirings.

\begin{lemma} \label{lem:syntax-to-semantics}
\begin{enumerate}[(1)]
\item    $\Sigma_1^+$-sentences are preserved under homomorphisms in every semiring.
\item $\Sigma_1$-sentences are preserved under extensions in every semiring.
\item $\Pi_1$-sentences are preserved under subinterpretations in every absorptive semiring.
\end{enumerate}
\end{lemma}

\section{Failure of Extension Preservation in General Semirings} \label{sec:counterexamples}
While the easy direction of the classical preservation theorems from syntax to semantics survives in semiring semantics, there are cases where the converse direction fails. 
We present two examples where sentences or sets of sentences
are preserved under extensions but cannot be written without the use of universal quantifiers. 
The key observation underlying both counterexamples is that a formula of the form $\forall y \phi$, where $y$ does not appear in $\phi$, is logically equivalent to $\phi$ in the multiplicative idempotent case, but not in general. For example, this no longer holds when we evaluate in the natural numbers. In order to ensure well-defined semantics also on infinite universes, we expand the natural semiring $(\N,+,\,\cdot\,,0,1)$ to $\N\cup\{\infty\}$ and first consider the semiring $\Ninf$.
The universal quantifier in a formula $\forall y \phi$ has the effect of just raising the value of the inner formula $\phi$ to the $n$-th power if $n$ is the size of the universe, and otherwise to the infinite power.
However, one can see that expressing powers that grow with the size of the universe is not possible in $\Ninf$ without the use of a universal quantifier.

\begin{restatable}{lemma}{counterexampleNat}
    \label{lem:counterexampleNat}
    The sentence $ \psi = \exists x \forall y Rx$ is preserved under extensions in $\Nat^\infty$ but not $\Nat^\infty$-equivalent to a $\Sigma_1$-sentence.
\end{restatable}

\begin{proof}
    For $\Nat^\infty$-interpretations $\pi_A \subseteq \pi_B$ we have that
    \begin{align*}
    \hspace{-5mm}
        \pi_A \llb \psi \rrb = \sum_{a \in A} \pi_A (R a)^{|A|}  \leq \sum_{a \in A} \pi_A (R a)^{|B|} + \sum_{b \in B \setminus A} \pi_B (Rb)^{|B|}
        =\sum_{b \in B} \pi_B (R b)^{|B|} = \pi_B \llb \psi \rrb
    \end{align*}
    because addition and exponentiation are increasing. Hence, $\psi$ is preserved under extensions in $\bbN^\infty$. It remains to show that $\psi$ is not $\Nat^\infty$-equivalent to a $\Sigma_1$-sentence. For $n \in \omega$, let $\pi_n$ be the $\Nat[\{x\}]$-interpretation with universe $A_n \coloneqq \{a_i \mid i \in [n]\}$ and valuations $\pi_n (Ra_i) =x$ and $\pi_n (\lnot Ra_i) =0$ for all $a_i \in A_n$. For $\phi \in \Sigma_1$, the valuation $\pi_n \llb \phi \rrb$ is a polynomial whose degree is bounded by the length of $\phi$ (and does not depend on $n$). However, $\pi_n \llb \psi \rrb = n \cdot x^n$. So there must be some $n,m \in \omega$ such that for the homomorphism $h_m \colon \Nat[\{x\}] \to \Nat^\infty$ induced by $h_m (x) = m$ we have that $(h_m \circ \pi_n) \llb \psi \rrb = h_m (\pi_n \llb \psi \rrb) > h_m (\pi_n \llb \phi \rrb) = (h_m \circ \pi_n) \llb \phi \rrb$. Hence, $\psi \not\equiv_{\Nat^\infty} \phi$.
\end{proof}

Our next counterexamples are the Viterbi semiring $\Vit = ([0,1]_\R,\max,\cdot,0,1)$ and the semiring of doubt $\Doubt = ([0,1]_\R, \min, \oplus, 1, 0)$ with $s \oplus t = \min(s+t, 1)$. Of course these results 
translate to the tropical semiring $\Trop$ and to the {\L}ukasiewicz semiring $\Lukas$, since $\Trop$ is isomorphic to $\Vit$, and $\Lukas$ is isomorphic to $\Doubt$.
Note that, unlike for $\Ninf$, multiplication is now decreasing rather than increasing, causing extension preservation of the sentence $\exists x \forall y Rx$ to fail in these semirings. Consider for example the $\Vit$-interpretation $\pi_A$ of universe $\{a\}$ with $\pi_A(Ra)=1/2$ and its extension $\pi_B$ over the universe $\{a, b\}$ with $\pi_B(Rb)=1/2$, where we have $\pi_A \llb \psi \rrb = 1/2 > 1/4 = \pi_B \llb \psi \rrb$. 
On \emph{infinite} $\Vit$-interpretations, however, the universal quantifier $\forall y$ raises the value of $Rx$ already to its infinite power (i.e. the infimum of the finite powers), and the maximum over these values is preserved under extensions.
Thus, we construct an axiom system which forces its models to be infinite and establish that extension preservation fails for a \emph{set} of sentences, or for semiring interpretations over infinite universes. 
The key argument for proving that no set of existential sentences has the same semantics is that infinite suprema and powers do not commute in $\Vit$. The argument for $\Doubt$ is analogous.

\begin{restatable}{lemma}{counterexampleV}
\label{lem:counterexampleViterbiInfinite}
For $\Semi \in \{ \Vit, \Doubt \}$, 
    the set $\Phi \coloneqq \{ \exists x_1 \dots \exists x_n \bigwedge_{1 \leq i < j \leq n} x_i \neq x_j \mid n \in \omega\} \cup \{\exists x \forall y Rx \}$ is preserved under extensions in $\Semi$, but is not $\Semi$-equivalent to a set of $\Sigma_1$-sentences.
\end{restatable}

\begin{proof}
Consider first the case $\Semi = \Vit$.
    Let $\pi_A \subseteq \pi_B$ be $\Vit$-interpretations. If $A$ is finite, we have that $\pi_A \llb \Phi \rrb = 0 \leq \pi_B \llb \Phi \rrb$. So let $A$ (and thus $B$, as $\pi_A \subseteq \pi_B$) be infinite. Then 
    $$\pi_A \llb \Phi \rrb = \bigsqcup_{a \in A} \prod_{a' \in A} \pi_A (Ra) = \bigsqcup_{a \in A} \pi_A (Ra)^\infty = \bigsqcup_{a \in A} \pi_B (Ra)^\infty \leq \bigsqcup_{b \in B} \pi_B (Rb)^\infty = \pi_B \llb \Phi \rrb.$$ 
    Hence, $\Phi$ is preserved under extensions in $\Vit$. 

    Fix some chain $(r_i)_{i \in \omega} \subseteq [0,1)_{\mathbb{R}}$ such that $\bigsqcup_{i \in \omega} r_i = 1$. Let $\pi_A$ be the $\Vit$-interpretation over universe $A=\{a_i \mid i \in \omega\}$ where $\pi_A (Ra_i) = r_i$, and let $\pi_B$ be the extension of $\pi_A$ with universe $B \coloneqq A \uplus \{b\}$ with $\pi_B (Rb) =1$. We have that $\pi_A \llb \Phi \rrb = 0 \neq 1 = \pi_B \llb \Phi \rrb$. We show that $\Phi$ cannot be $\Vit$-equivalent to a set of $\Sigma_1$-sentences by proving that for every $\phi= \exists x_1 \dots \exists x_k \psi (x_1, \dots, x_k) \in \Sigma_1$ with $\pi_B \llb \phi \rrb =1$, we also have $\pi_A \llb \phi \rrb =1$. So fix such a $\phi \in \Sigma_1$. Let
    $\epsilon > 0$ be arbitrary and consider a tuple $\bar b \in B^k$.  
    
    Since $\pi_B \llb \psi (\bar b) \rrb$ is a product of the valuations of literals $R (\bar b)$ and $\lnot R (\bar b)$, we can replace each occurrence of the unique element $b \in B \setminus A$ in $\bar b$ by some sufficiently large $a_i \in A$ such that for the resulting tuple $\bar a$ we get $\pi_A \llb \psi (\bar a) \rrb \geq \pi_B \llb \psi (\bar b) \rrb - \epsilon$.
    Since $\epsilon > 0$ and $\bar{b} \in B^k$ were arbitrary, this shows that 
    \[
    \pi_A\llb \phi \rrb = \bigsqcup_{\bar a \in A^k} \pi_A\llb \psi (\bar a)  \rrb = \bigsqcup_{\bar b \in B^k} \pi_B\llb \psi (\bar b)  \rrb = 1.
    \]
    This finishes the case that $\Semi = \Vit$. If $\Semi = \Doubt$, then the reasoning is analogous, with the difference that in $\Doubt$, the natural order is reversed: To be clear, let $\leq_{\R}$ denote the standard order on the interval $[0,1]_{\R}$, and let $\leq_{\Doubt}$ denote the order on $\Doubt$ induced by the semiring addition operation, which is $\min$ in this case. Then $a\leq_{\R}b$ if and only if $b \leq_{\Doubt} a$. 
    Let $\pi_A \subseteq \pi_B$ be $\Doubt$-interpretations. We can again assume that both $A$ and $B$ are infinite. Let $t_{>0} \colon \Doubt \to \Doubt$ be the function that maps $0$ to $0$ and every other element of $\Doubt$ to $1$.
    Also in $\Doubt$, $\Phi$ is preserved under extensions:  
    $$\pi_A \llb \Phi \rrb = \Inf_{a \in A} \bigoplus_{a' \in A} \pi_A (Ra)  = \Inf_{a \in A} t_{>0}(\pi_A (Ra)) \leq_{\Doubt} \Inf_{a \in B} t_{>0}(\pi_B (Ra)).$$
    The inequality is due to the fact that $B \supseteq A$, so the right-hand side can only be less than the left-hand side with respect to $\leq_{\R}$. 
    To show that $\Phi$ is not $\Doubt$-equivalent to a set of $\Sigma_1$-sentences, we use a dual construction to the previous one for $\pi_A$ and $\pi_B$. That is, we fix a chain $(r_i)_{i \in \omega} \subseteq (0,1]_{\mathbb{R}}$ such that $\Inf_{i \in \omega} r_i = 0$. Let $\pi_A$ be the $\Vit$-interpretation over universe $A=\{a_i \mid i \in \omega\}$ where $\pi_A (Ra_i) = r_i$, and let $\pi_B$ be the extension of $\pi_A$ by an additional element $b$ with $\pi_B (Rb) = 0$. We have that $\pi_A \llb \Phi \rrb = 1 \neq 0 = \pi_B \llb \Phi \rrb$. 
    It remains to argue that for every $\phi= \exists x_1 \dots \exists x_k \psi (x_1, \dots, x_k) \in \Sigma_1$, it is the case that $\pi_B \llb \phi \rrb = 0$ also implies $\pi_A \llb \phi \rrb = 0$. This follows analogously as before by considering an arbitrary $\epsilon > 0$ and $\bar b \in B^k$, and observing that by replacing each occurrence of the special element $b$ in $\bar b$ with a sufficiently small $a_i$, we obtain a tuple $\bar a$ such that $\pi_A \llb \psi (\bar a) \rrb \leq_{\R} \pi_B \llb \psi (\bar b) \rrb+\epsilon$. Hence, taking the infimum over all $\bar a \in A^k$, we conclude that $\pi_A \llb \phi \rrb = 0$ and so, $\Phi$ cannot be $\Doubt$-equivalent to a set of $\Sigma_1$-sentences.
    \qedhere
\end{proof}

\section{Preservation Theorems in Lattice Semirings} \label{sec:lattices}

Whether we consider single sentences or sets of them, the extension preservation theorem does not generalise to all semirings. 
This raises the question in which semirings it does survive.
We show that both variants hold for the important and fairly general class of \emph{lattice semirings}, induced by a completely distributive complete lattice (i.e. a partial order closed under suprema and infima in which infima distributive over suprema), with supremum and infimum as semiring operations. Also the homomorphism and subinterpretation preservation theorems generalise to the lattice setting.

\begin{theorem}
\label{thm:latticePreservationMain}
    Let $\calL$ be a lattice semiring.
    A sentence $\psi \in \FO$ is preserved under homomorphisms/extensions/subinterpretations in $\calL$ if, and only if, there is some $\phi \in \Sigma_1^+$/ $\phi \in \Sigma_1$/ $\phi \in \Pi_1$ such that $\psi \equiv_\calL \phi$.
\end{theorem}

The classical proofs of most model-theoretic preservation theorems
rely on compactness and amalgamation. Compactness has recently been studied in the context of semiring semantics in \cite{BrinkeDawGraMrkNaa26} and positive results have been established, in particular, for the lattice case.
However, by itself, compactness alone does not suffice to
extend the proofs of preservation theorems to lattice semirings.

The extension preservation theorem under Boolean semantics, for example, is proved by showing that a set $\Phi \subseteq \FO$ which is preserved under extensions is logically equivalent to the set $\Phi_{\exists} \coloneqq \{\psi \in \Sigma_1 \mid \Phi \models \psi \}$ of $\Sigma_1$-consequences of $\Phi$. By applying compactness, this yields a single equivalent $\Sigma_1$-sentence in the case where $\Phi$ is a singleton. While it is an easy observation that $\Phi \models \Phi_\exists$, a more involved model-theoretic argument is needed to show the converse entailment $\Phi_\exists \models \Phi$. One proves that for every model $\AA \models \Phi_\exists$ there is a suitable $\BB \models \Phi$ that can be \emph{amalgamated} into a structure $\mathfrak C \succeq \mathfrak A$ that is also an extension of $\BB$. This yields $\mathfrak C \models \Phi$ by extension preservation of $\Phi$. Since $\mathfrak C$ is an elementary extension of $\mathfrak A$, this implies $\AA \models \Phi$. To ensure that amalgamation is possible, $\BB$ has to be chosen such that every $\Sigma_1$-sentence that holds in $\BB$ also holds in $\AA$, or, in other words, such that every finite substructure of $\BB$ is a substructure of $\AA$, too.

This relationship to shared finite subinterpretations no longer holds if we move from the Boolean semiring to lattices. It even fails for min-max semirings over $4$ elements where finite interpretations cannot be axiomatised up to isomorphism \cite{GraedelMrk21}. 
The intuitive reason for this is that the semiring elements do not occur in the logic's syntax.
The $\Sigma_1$-sentence $\psi = \exists x (Px \wedge \lnot Qx)$, for example, describes a substructure of size $1$ in the Boolean case, while its evaluation in a lattice generally does not reflect the existence of a particular subinterpretation. Since existential quantifiers are evaluated to a supremum, it might not even be possible to attribute the valuation of $\psi$ to a single instantiation of $x$. 
Even if we fix some instantiation $a$, by commutativity, the evaluation of $Pa \wedge \lnot Qa$ does not give insights into the contribution of the single literals.
So, in a sense, $\Sigma_1$-sentences seem weaker in the lattice context, and it is a priori unclear whether $\Phi_\exists$ still entails $\Phi$ in the presence of extension preservation.
This mismatch between $\Sigma_1$-sentences and the notions of subinterpretations and extensions in semiring semantics is what makes it a challenge to show that the extension preservation theorem still holds. We achieve this by combining arguments used in the classical proof with new methods specifically developed for semiring semantics.

	\subsection{Proof Methods from Semiring Semantics} 
	
	The proofs of the three preservation theorems for \emph{all} lattice semirings rely on a reduction to the case where we evaluate in \emph{one fixed} such semiring. Here we make use of the three-element lattice $\Semi_3$, where an element denoted $\epsilon$ is sandwiched in between $0$ and $1$, that is, $\Semi_3 = (\{0,\epsilon,1\}, \max, \min, 0, 1)$.

	\paragraph*{Reduction to $\Semi_3$}
	As $\Semi_3$ embeds into any lattice semiring except for the Boolean, it can be thought of as the simplest lattice semiring beyond the Boolean.
	As a consequence, both preservation properties and logical equivalences in $\Semi_3$ are a special case of the same property in any lattice $\calL \not\cong \Bool$. In order to also transfer counterexamples to these properties from lattice semirings $\calL \not\cong \Bool$ back to $\Semi_3$, we apply a decomposition method from \cite{BrinkeDawGraMrkNaa26}. Such decompositions are based on semiring homomorphisms, and in the case where they are \emph{infinitary}, i.e. also commute with the infinitary semiring operations, this inductively extends to the semiring semantics of any first-order formula.
	
	\begin{lemma}[Fundamental Property] \label{fundamental-property}
		Let $\Semi, \calT$ be (infinitary) semirings, $\pi \colon \Lit_A(\tau) \to \Semi$ be an (infinite) $\Semi$-interpretation, and $h \colon \Semi \to \calT$ be an (infinitary) semiring homomorphism. Then, $(h \circ \pi)$ is a $\Tt$-interpretation and it holds that $(h \circ \pi) \llb \Phi( \bar a)\rrb = h(\pi \llb \Phi (\bar a) \rrb )$ for all $\Phi (\bar a) \subseteq \FO_A (\tau)$.
	\end{lemma}
	
	In the special case where $\calT$ is a min-max semiring, we can drop the condition that $h$ must commute with the infinitary operations and still obtain a weaker relationship between $h(\pi \llb \Phi \rrb)$ and $(h \circ \pi) \llb \Phi \rrb$.
	
	\begin{lemma}[\cite{BrinkeDawGraMrkNaa26}]
		\label{thm-weak-fun} 
		Let $h \colon \Semi \to \mathcal{T}$ be a homomorphism (which does not necessarily respect the infinitary operations) into a min-max semiring $\mathcal{T}$, and let $t \in \mathcal{T}$. 
		\begin{enumerate}[(1)]
			\item If $h (\pi \llb \phi (\bar a) \rrb ) \geq t$ and if for all $X \subseteq \Semi$ with $\pi \llb \phi (\bar a) \rrb = \Sup X$ there is some $x \in X$ such that $h (x) \geq t$, then also $(h \circ \pi) \llb \phi (\bar a) \rrb \geq t$.
			\item If $h (\pi \llb \phi (\bar a) \rrb ) \leq t$ and if for all $X \subseteq \Semi$ with $\pi \llb \phi (\bar a) \rrb = \Inf X$ there is some $x \in X$ such that $h (x) \leq t$, then also $(h \circ \pi) \llb \phi (\bar a) \rrb \leq t$.
		\end{enumerate}
	\end{lemma}
	
	\begin{lemma}[\cite{BrinkeDawGraMrkNaa26}] \label{lem:ex-weakly-sep-horm}
		Let $\calL$ be a lattice semiring without divisors of $0$.
		For each $s \not\leq t$, there is a homomorphism $h \colon \calL \to \Semi_3$ that \emph{weakly $\leq$-separates $s$ from $t$}, which means that
		\begin{enumerate}[(1)]
			\item $h^{-1} (0) =\{0\}$;
			\item $h(s) > h(t)$;
			\item if $s = \Sup X$, then $h(x)=h (s)$ for some $x \in X$;
			\item if $t = \Inf X$, then $h(x)=h(t)$ for some $x \in X$.
		\end{enumerate}
	\end{lemma}
	
	\begin{lemma} \label{lem:counterexample-pres-SThree}
		Let $\calL$ be a lattice semiring without divisors of $0$, $\pi_A, \pi_B$ be $\calL$-interpretations, and $\Phi, \Psi \subseteq \FO$. Let further $s \not\leq t \in \calL$ and $h \colon \calL \to \Semi_3$ weakly $\leq$-separate $s$ from $t$.
		\begin{enumerate}[(1)]
			\item If $\pi_A \llb \Phi \rrb=s$ and $\pi_B \llb \Phi \rrb = t$, then $(h\circ \pi_A) \llb \Phi \rrb > (h \circ \pi_B) \llb \Phi \rrb$.
			\item If $\pi_A \llb \Phi \rrb=s$ and $\pi_A \llb \Psi \rrb = t$, then $(h\circ \pi_A) \llb \Phi \rrb > (h \circ \pi_A) \llb \Psi \rrb$.
		\end{enumerate}
	\end{lemma}
	
	\begin{proof}
		We only prove (1). The proof of (2) is analogous.
		Because $h$ weakly $\leq$-separates $s$ from $t$, it holds that $h(\pi_A \llb \Phi \rrb) \not\leq h(\pi_B \llb \Phi \rrb)$. Further, $h$ is a homomorphism, so preserves the lattice order, implying $h(\pi_A \llb \phi \rrb) \geq h(\bigsqcap_{\phi \in \Phi} \pi_A \llb \phi \rrb) = h(\pi_A \llb \Phi \rrb) = h(s)$ for each $\phi \in \Phi$. Suppose that there were $\phi \in \Phi$ and $X \subseteq \calL$ such that $\pi_A \llb \phi \rrb = \bigsqcup X$ and $h(x) < h(s)$ for each $x \in X$. Then $\pi_A \llb \Phi \rrb = \bigsqcup X \sqcap \pi_A \llb \Phi \rrb =\bigsqcup \{x \sqcap \pi_A \llb \Phi \rrb \mid x \in X\}$ where $h (x \sqcap \pi_A \llb \Phi \rrb) \leq h( x) < h(s) = h (\pi_A \llb \Phi \rrb)$, a contradiction. Hence, we can apply \cref{thm-weak-fun} to each $\phi \in \Phi$ (we choose $t \coloneqq h(\pi_A \llb \Phi \rrb)$) and obtain $(h \circ \pi_A) \llb \phi \rrb \geq h(\pi_A \llb \Phi \rrb)$. Overall we obtain $(h \circ \pi_A) \llb \Phi \rrb \geq h(\pi_A \llb \Phi \rrb)$. It remains to prove that $(h \circ \pi_B) \llb \Phi \rrb \leq h(\pi_B \llb \Phi \rrb)$. It holds that $\pi_B \llb \Phi \rrb = \bigsqcap_{\phi \in \Phi} \pi_B \llb \phi \rrb$, which means that $h(\pi_B \llb \phi \rrb) = h(\pi_B \llb \Phi \rrb)$ for some $\phi \in \Phi$ as $h$ is weakly $\leq$-separating $\pi_A \llb \Phi \rrb$ from $\pi_B \llb \Phi \rrb$. Let $\Phi_0 \coloneqq \{\phi \in \Phi \mid h(\pi_B \llb \phi \rrb) = h (\pi_B \llb \Phi \rrb) \} \neq \varnothing$. Suppose that for each $\phi \in \Phi_0$ there is some $X_\phi$ such that $\pi_B \llb \phi \rrb = \bigsqcap X_\phi$ and $h(x) > h(\pi_B \llb \phi \rrb)$ for each $x \in X$. Then $\pi_B \llb \Phi \rrb = \bigsqcap \bigcup_{\phi \in \Phi_0} X_\phi \cup \{\pi_B \llb \phi \rrb \mid \phi \in \Phi \setminus \Phi_0\}$, so $h$ could not weakly $\leq$-separate $\pi_A \llb \Phi \rrb$ from $\pi_B \llb \Phi \rrb$. Hence, there is some $\phi \in \Phi$ to which we can apply \cref{thm-weak-fun} for $t \coloneqq h (\pi_B \llb \Phi \rrb)$ and we get $(h \circ \pi_B ) \llb \Phi \rrb \leq (h \circ \pi_B ) \llb \phi \rrb \leq h(\pi_B \llb \Phi \rrb)$. Overall, we have shown that $(h\circ \pi_A) \llb \Phi \rrb \geq h (\pi_A \llb \Phi \rrb) > h (\pi_B \llb \Phi \rrb) \geq (h \circ \pi_B) \llb \Phi \rrb$.
	\end{proof}
	
	\begin{proposition} \label{prop:preservation-in-all-lattices}
		Let $\calL, \calL' \not\cong \Bool$ be lattice semirings and $\Phi \subseteq \FO$ be a set of sentences. 
		If $\Phi$ is preserved under homomorphisms/extensions/subinterpretations in $\calL$, then $\Phi$ must also be preserved under homomorphisms/extensions/subinterpretations in $\calL'$.
	\end{proposition} 
	
	\begin{proof}
		We only treat extension preservation, homomorphism and subinterpretation preservation are proved analogously.
		First, we argue that extension preservation in $\calL$ implies extension preservation in $\Semi_3$.
		Since $\calL$ must contain at least three elements, there is an embedding $e \colon \Semi_3 \subseteq \calL$. For $\Semi_3$-interpretations $\pi_A \subseteq \pi_B$, we have $e \circ \pi_A \subseteq e \circ \pi_B$. By the Fundamental Property (\cref{fundamental-property}), extension preservation in $\calL$ implies
		$
		(e \circ \pi_A) \llb \Phi \rrb = e(\pi_A \llb \Phi \rrb ) \leq e(\pi_B \llb \Phi \rrb )= (e \circ \pi_B) \llb \Phi \rrb.
		$
		Hence, $\Phi$ must be preserved under extensions in $\Semi_3$, too. 
		Now suppose that $\Phi$ was not extension preserved in $\calL'$, so that there are $\calL'$-interpretations $\pi_A \subseteq\pi_B$ such that $\pi_A \llb \Phi \rrb \not\leq \pi_B \llb \Phi \rrb$. Now extend $\calL'$ by an additional element $0^*$ such that $0^* \leq s$ for all $s \in \calL'$. Let $\pi_A^*, \pi_B^*$ be $\calL^*$-interpretations where we set $\pi_A^* (L) = 0^*$ whenever $\pi_A (L)=0$ and $\pi_A^* (L)= \pi_A (L)$ otherwise. $\pi_B^*$ arises from $\pi_B$ in the same way. Now let $h^* \colon \calL^* \to \calL'$ map every element to itself except for $0^*$, which is mapped to $0$. We have $\pi_A = h^* \circ \pi_A^*$ and $\pi_B = h^* \circ \pi_B^*$, so we cannot have that $\pi_A^* \llb \Phi \rrb \leq \pi_B^* \llb \Phi \rrb$ as this would imply $\pi_A \llb \Phi \rrb  = (h^* \circ \pi_A^*) \llb \Phi \rrb = h^* (\pi_A^* \llb \Phi \rrb) \leq h^* (\pi_B^* \llb \Phi \rrb) = (h^* \circ \pi_B^*) \llb \Phi \rrb = \pi_B \llb \Phi \rrb$. Now $\calL^*$ does not have divisors of $0^*$, which means that we can apply \cref{lem:ex-weakly-sep-horm} and obtain a homomorphism $h \colon \calL^* \to \Semi_3$ weakly $\leq$-separating $\pi_A^* \llb \Phi \rrb $ from $ \pi_B^* \llb \Phi \rrb$. By \cref{lem:counterexample-pres-SThree}, $(h \circ \pi_A^*) \llb \Phi \rrb > (h \circ \pi_B^*) \llb \Phi \rrb$, so $\Phi$ could not be preserved under extensions in $\Semi_3$, a contradiction. Hence, $\Phi$ must be extension preserved in $\calL'$.
	\end{proof}
	
	Note that this reduction only works for $\Semi_3$, but not for the Boolean semiring itself. While the sentence $\psi = \forall x (Rx \vee \lnot Rx)$, for example, as a tautology, is clearly extension preserved in the classical sense, this no longer holds as soon as the lattice contains a third truth value $0 < s < 1$. Extending an $\calL$-interpretation over a single element $a$ whose valuation with respect to $R$ is $1$ by a second element $b$ such that $Rb$ is evaluated to $s$ decreases the valuation of $\psi$.
	Applying the method from \cite{BrinkeDawGraMrkNaa26} another time yields that not just the preservation properties but also all logical equivalences are invariant under the specific lattice one considers.
	
	\begin{restatable}{proposition}{reductionToSThree} \label{lem:entailment-from-S3-to-any-lattice}
		For all $\Phi, \Psi \subseteq \FO$, the following are equivalent:
		\begin{enumerate}[(1)]
			\item $\Phi \models_\calL \Psi$ for all lattice semirings $\calL \not\cong \Bool$;
			\item $\Phi \models_\calL \Psi$ for some lattice semiring $\calL \not\cong \Bool$;
			\item $\pi \llb \Psi \rrb = 1$ for all $\Semi_3$-interpretations $\pi$ with $\pi \llb \Phi \rrb = 1$. 
		\end{enumerate}
	\end{restatable}
	
	\begin{proof}
		The implication $(1) \Rightarrow (3)$ is clear. We prove $(3) \Rightarrow (2)$ and $(2) \Rightarrow (1)$, both via contraposition. First consider $(2) \Rightarrow (1)$, so let $\calL, \calL' \not\cong \Bool$ be such that $\Phi \models_\calL \Psi$. We want to prove that $\Phi \models_{\calL'} \Psi$. Suppose the contrary. Then there was some $\calL'$-interpretation $\pi$ such that $\pi \llb \Phi \rrb \not\leq \pi \llb \Psi \rrb$. Transform $\pi$ into an $\calL^*$-interpretation $\pi^*$ such that $s:=\pi^* \llb \Phi \rrb \not\leq \pi^* \llb \Psi \rrb=:t$ in the same way as in \cref{prop:preservation-in-all-lattices}. Because $\calL^*$ does not have $0$-divisors, we can fix a homomorphism weakly $\leq$-separating $s$ from $t$ according to \cref{lem:ex-weakly-sep-horm} and obtain $(h \circ \pi) \llb \Phi \rrb > (h \circ \pi) \llb \Psi \rrb$, so we would have that $\Phi \not\models_{\Semi_3} \Psi$. But $\Semi_3$ embeds into $\calL$, a contradiction (the argument is the same as in the proof of \cref{prop:preservation-in-all-lattices}). 
		Now we prove $(3) \Rightarrow (2)$ and choose $\calL = \Semi_3$ in (2). We only have to show that $\pi \llb \Phi \rrb = \epsilon$ implies $\pi \llb \Psi \rrb \geq \epsilon$ for all $\Semi_3$-interpretations $\pi$. To this end, consider the homomorphism $h_{\geq \epsilon} \colon \Semi_3 \to \Semi_3$ with $h_{\geq \epsilon} (s) = 1$ if $s \geq \epsilon$ and $h_{\geq \epsilon} (0)=0$. By the fundamental property (\cref{fundamental-property}), we have that $(h_{\geq \epsilon} \circ \pi) \llb \Phi \rrb =h_{\geq \epsilon}(\pi \llb \Phi \rrb)=1$, implying $(h_{\geq \epsilon} \circ \pi) \llb \Psi \rrb =1$ by assumption, and hence $\pi \llb \Psi \rrb \geq \epsilon$.
	\end{proof}
	
	\begin{corollary} \label{cor:equivalence-from-S3-to-any-lattice}
		For all $\Phi, \Psi \subseteq \FO$, the following are equivalent:
		\begin{enumerate}[(1)]
			\item $\Phi \equiv_\calL \Psi$ for all lattice semirings $\calL \not\cong \Bool$;
			\item $\Phi \equiv_\calL \Psi$ for some lattice semiring $\calL \not\cong \Bool$;
			\item $\pi \llb \Phi \rrb = 1$ if, and only if, $\pi \llb \Psi \rrb = 1$ for all $\Semi_3$-interpretations $\pi$. 
		\end{enumerate}
	\end{corollary}
	
	Similar to the preservation properties, observe that there are logical equivalences that hold in the Boolean, not in any other lattice. The reason for this is that some Boolean tautologies can also only evaluate to $1$ in a lattice as we interpret (in)equalities by their Boolean truth value, while others may take values different from $1$.
	As an example, we have $\exists x (x =x) \not\equiv_{\Semi_3} \exists x (Rx \vee \lnot Rx)$.

	Combining \cref{prop:preservation-in-all-lattices} and \cref{cor:equivalence-from-S3-to-any-lattice} justifies that is suffices to show that one lattice semiring $\calL \not\cong \Bool$ satisfies one of the preservation theorems in order to prove it for all lattice semirings $\calL' \not\cong \Bool$: If $\psi$ has some preservation property w.r.t $\calL'$, say, it is extension preserved in $\calL'$, then it must also be extension preserved in $\calL$, and the $\Sigma_1$-sentence $\calL$-equivalent to $\psi$ must also be $\calL'$-equivalent to $\psi$. We will prove that the extension, subinterpretation, and homomorphism preservation theorems generalise to $\Semi_3$. 
    An essential part of the proof will be the third statement of \cref{cor:equivalence-from-S3-to-any-lattice}, which intuitively says that the $\Semi_3$-semantics of a sentence is uniquely determined by its $1$-valuations. This is not to be confused with its classical Boolean semantics, because it is still evaluated on $\Semi_3$-interpretations that may also annotate literals with $\epsilon$. However, this insight will yield some relationship between $\Sigma_1$ and finite subinterpretations. Consider the sentence $\psi = \exists x(Px \wedge \lnot Qx)$ again. While a valuation of some arbitrary element $s$ does not generally ensure the existence of a particular subinterpretation, a valuation of $1$ does: because $1$ is maximal in any lattice semiring, $\psi$ is evaluated to $1$ by a $\Semi_3$-interpretation if, and only if, it contains the Boolean substructure classically described by $\psi$.

	\paragraph*{Compactness}
	
	The semiring setting admits different notions of compactness. Restricting ourselves to the case $\Semi_3$ allows us to apply two different variants, one of them being based on \emph{$\Semi$-axioms} $(\psi, T)$, where $\psi$ is a sentence and $T$ a set of designated truth values (a modified version of the notion used in \cite{Mrkonjic25}), and another one based on entailment (studied in \cite{BrinkeDawGraMrkNaa26}). 
	While we use the first variant to prove an amalgamation theorem and justify its application, the second variant is important to infer a single equivalent $\Sigma_1$-sentence starting from a single extension preserved sentence, just as in the Boolean proof.

	\begin{Definition}
		An \emph{$\Semi$-axiom} (over a universe $A$) is a pair $(\psi, T)$ where $\psi \in \FO_A(\tau)$ and $T \subseteq \Semi$. An $\Semi$-interpretation $\pi$ satisfies $(\psi, T)$ if $\pi \llb \psi \rrb \in T$. A set of $\Semi$-axioms $\Delta$ is called an \emph{axiomatic system}, and it is satisfied by $\pi$ if $\pi$ satisfies each $(\psi, T) \in \Delta$. An axiomatic system $\Delta$ is satisfiable if there is a model-defining interpretation that satisfies it.
	\end{Definition}
		
		\begin{theorem}[\cite{Mrkonjic25}]
			\label{thm:compactnessFiniteAbsorptive}
			Let $\Semi$ be finite and absorptive, and $\Delta$ be an unsatisfiable set of $\Semi$-axioms. Then there is already a finite unsatisfiable $\Delta_0 \subseteq \Delta$.
		\end{theorem}
		It should be noted that \cite{Mrkonjic25} proves \cref{thm:compactnessFiniteAbsorptive} only for $\Semi$-axioms $(\psi, T)$ where $T$ is a singleton, but the case where $T \subseteq \Semi$ is any finite set is completely analogous.
		
		\begin{theorem}[\cite{BrinkeDawGraMrkNaa26}]
			\label{thm:compactnessLattice}
			Let $\calL$ be a lattice semiring. If $\Phi \models_\calL \psi$, then there is some finite $\Phi_0 \subseteq \Phi$ such that $\Phi_0 \models_{\calL} \psi$.
		\end{theorem}

\subsection{Homomorphism Preservation}

The first preservation theorem we prove is for homomorphisms.
We first consider the case $\Semi_3$ and generalise it to arbitrary lattice semirings $\calL$.
More precisely, we prove that every set of sentences $\Phi$ which is preserved under homomorphisms in $\Semi_3$ is $\Semi_3$-equivalent to some $\Psi \subseteq \Sigma_1^+$. 
As in Boolean semantics, we show that $\Phi \equiv_{\Semi_3} \Phi_{\exists^+}$, where $\Phi_{\exists^+} \coloneqq \{\psi \in \Sigma_1^+ \mid \Phi \models_{\Semi_3} \psi\}$. 
It holds that $\Phi \models_{\Semi_3} \Phi_{\exists^+}$ because for any model-defining $\Semi_3$-interpretation $\pi$, $\pi \llb \Phi \rrb$ is a lower bound of $\{\pi \llb \psi \rrb \mid \psi \in \Phi_{\exists^+}\}$, which implies $\pi \llb \Phi \rrb \leq \bigsqcap \{\pi \llb \psi \rrb \mid \psi \in \Phi_{\exists^+}\} = \pi \llb \Phi_{\exists^+} \rrb$. 
It remains to show that if $\Phi$ is preserved under homomorphisms, the converse entailment $\Phi_{\exists^+} \models_{\Semi_3} \Phi$ holds as well. To this end, it suffices to consider valuations of $1$ only, as justified by the \emph{fundamental property}, which states that semiring semantics commutes with homomorphisms, and the following lemma.

\begin{lemma} \label{lem-interpretation-for-amalgamation}
    Let $\pi_A$ be an $\Semi_3$-interpretation such that $\pi_A \llb \psi \rrb = 1$ for all $\psi \in \Phi_{\exists^+}$. Then there is an $\Semi_3$-interpretation $\pi_B$ such that
    \begin{enumerate}[(1)]
        \item $\pi_B \llb \Phi \rrb =1$, and
        \item $\pi_B \llb \psi \rrb \leq \epsilon$ for all $\psi \in \Sigma_1^+$ with $\pi_A \llb \psi \rrb \leq \epsilon$.
    \end{enumerate}
\end{lemma}

\begin{proof}
    Let $\Psi \coloneqq \{\psi \in \Sigma_1^+ \mid \pi_A \llb \psi \rrb \leq \epsilon\}$ and
    suppose towards a contradiction that the set $\{(\phi, \{1\}) \mid \phi \in \Phi \} \cup \{ (\psi , \{0, \epsilon\}) \mid \psi \in \Psi\}$ was unsatisfiable. By compactness (\cref{thm:compactnessFiniteAbsorptive}), there must already be a finite unsatisfiable subset. Let $\Psi_0 \subseteq \Psi$ be such that $\{(\phi, {1}) \mid \phi \in \Phi \} \cup \{(\psi, \{0, \epsilon\}) \mid \psi \in \Psi_0\}$ is unsatisfiable. That means that every model-defining $\Semi_3$-interpretation $\pi$ with $\pi \llb \Phi \rrb =1$ also satisfies $\pi \llb \bigvee \Psi_0 \rrb =1$. By \cref{lem:entailment-from-S3-to-any-lattice}, this implies $\Phi \models_{\Semi_3} \bigvee \Psi_0$ and since $\bigvee \Psi_0 \equiv_{\Semi_3} \psi$ for some $\psi \in \Sigma_1^+$ (by \cref{lem:closure-under-and-or}), $\bigvee \Psi_0$ is logically equivalent to a formula in $\Phi_{\exists^+}$. By assumption, $\pi_A \llb \bigvee \Psi_0 \rrb = 1$, but we defined $\Psi$ such that $\pi_A \llb \psi \rrb \leq \epsilon$ for all $\psi \in \Psi$, a contradiction. 
\end{proof}

\begin{lemma}[Homomorphism Amalgamation] \label{homomorphism-amalgamation}
    Let $\pi_A, \pi_B$ be $\Semi_3$-interpretations such that $\pi_B \llb \psi \rrb \leq \epsilon$ for all $\psi \in \Sigma_1^+$ with $\pi_A \llb \psi \rrb \leq \epsilon$. Then there is some $\pi_C \succeq \pi_A$ and a mapping $f \colon B \to C$ such that $\pi_C (L f (\bar b)) = 1$ for all $L \bar b \in \operatorname{Atoms}_B (\tau)$ with $\pi_B (L \bar b) = 1$.
\end{lemma}

\begin{proof}
    Suppose that such a $\pi_C$ does not exist, i.e., that $\{(\phi, \{\pi_A \llb \phi \rrb\}) \mid \phi \in \FO_A(\tau)\} \cup \{(\psi, \{1\}) \mid \psi \in \Psi\}$, where $\Psi \coloneqq \{L \bar b \in \operatorname{Atoms}_B (\tau) \mid \pi_B (L \bar b) = 1\}$, is unsatisfiable. Let $\Psi_0 \subseteq \Psi$ be finite such that $\{(\phi, \{\pi_A \llb \phi \rrb\}) \mid \phi \in \FO\} \cup \{(\psi, \{1\}) \mid \psi \in \Psi_0\}$ is unsatisfiable; this exists by \cref{thm:compactnessFiniteAbsorptive}. Consider the set $\Psi_0^*$ that is obtained from $\Psi_0$ by replacing each constant $b_i \in B$ with a variable $x_i$. Since $\Psi_0$ is finite, only a fixed number $k$ of variables is needed. 
    Suppose that $\pi_A \llb \exists x_1 \dots \exists x_k \bigwedge \Psi_0^* \rrb =1$. Then there would be $a_1, \dots, a_k \in A$ such that $\pi_A (L \bar a) =1$ for all $L(x_1, \dots, x_k) \in \Psi_0^*$, and the expansion of $\pi_A$ where we interpret each constant $b_i \in B$ occurring in $\Psi_0$ with $a_i$ would be a model of $\{(\phi, \{\pi_A \llb \phi \rrb\}) \mid \phi \in \FO\} \cup \{(\psi, \{1\}) \mid \psi \in \Psi_0\}$.
    This is a contradiction, so we must have that $\pi_A \llb \exists x_1 \dots \exists x_k \bigwedge \Psi_0^* \rrb \leq \epsilon$.
    By assumption, this implies $\pi_B \llb \exists x_1 \dots \exists x_k \bigwedge \Psi_0^* \rrb \leq \epsilon$, but this contradicts the definition of~$\Psi$.
\end{proof}

\begin{proposition}
\label{thm:homomorphismPreservationMainDirection}
    If $\Phi$ is preserved under homomorphisms in $\Semi_3$, then $\Phi_{\exists^+} \models_{\Semi_3} \Phi$.
\end{proposition}

\begin{proof}
    Let $\pi_A \llb \Phi_{\exists^+} \rrb = 1$, i.e., $\pi_A \llb \psi \rrb = 1$ for all $\psi \in \Phi_{\exists^+}$. By \cref{lem:entailment-from-S3-to-any-lattice}, it suffices to prove that $\pi_A \llb \Phi \rrb =1$. Let $\pi_B$ be as in \cref{lem-interpretation-for-amalgamation} and $\pi_C$ be the corresponding amalgamated $\Semi_3$-interpretation according to \cref{homomorphism-amalgamation}. We know that $\pi_B \llb \Phi \rrb =1$ and want to argue that $\pi_C \llb \Phi \rrb =1$ holds as well. However, there need not be a homomorphism from $\pi_B$ to $\pi_C$ (only a mapping that preserves valuations of $1$), which is why we take a detour via another $\Semi_3$-interpretation $\pi_B^*$ (over the same universe $B$) that still evaluates $\Phi$ to $1$ but admits a homomorphism into $\pi_C$.

    Consider the homomorphism $h_{\geq 1} \colon \Semi_3 \to \Semi_3$ where $h_{\geq 1} (1) = 1$ and $h_{\geq 1}(0)=h_{\geq 1}(\epsilon) = 0$. By construction, $\pi_B \llb \Phi \rrb = 1$, which implies $(h_{\geq 1} \circ \pi_B) \llb \Phi \rrb = h_{\geq 1} (\pi_B \llb \Phi\rrb) = 1$. Since valuations of $\epsilon$ do not occur in $(h_{\geq 1} \circ \pi_B)$, the mapping $f$ we get from \cref{homomorphism-amalgamation} must be a homomorphism from $(h_{\geq 1} \circ \pi_B)$ to $\pi_C$. However, $(h_{\geq 1} \circ \pi_B)$ does not have to be model-defining and we might have that both a literal and its complement are evaluated to $0$. Hence, we construct yet another $\Semi_3$-interpretation $\pi_B^*$, also on the universe $B$, by inserting the missing valuations (we simply copy them from $\pi_C$), which ensures that $f$ is still a homomorphism into $\pi_C$. 
    
    More precisely, $\pi_B^*$ is defined as follows: For $L \bar b \in \Lit_B (\tau)$ we set $\pi_B^* (L \bar b) \coloneqq \pi_B (L \bar b)$ if $\pi_B (L \bar b) = 1$ or $\pi_B (\lnot L \bar b)=1$.\footnote{$\lnot L \bar b$ refers to the dual literal, i.e. we identify $\lnot \lnot L \bar b$ with $L \bar b$} Otherwise, we set $\pi_B^* (L \bar b) \coloneqq \pi_C (L f (\bar b))$. Note that $\pi_B^*$ is model-defining; in the first case this is inherited from $\pi_B$, and in the second case this is inherited from $\pi_C$. Moreover, $f$ is a homomorphism from $\pi_B^*$ to $\pi_C$: Consider $L \bar b \in \operatorname{Atoms}_B(\tau)$. We only need to consider the case where $\pi_B (L \bar b) = 1$ or $\pi_B (\neg L \bar b) = 1$; the rest is by definition of $\pi_B^*$. So suppose that $\pi_B^* (L \bar b)=\pi_B (L \bar b) = 1$. Then $\pi_C (L f (\bar b)) = 1$ follows because $f$ preserves $1$-valuations of atoms in $\pi_B$. If $\pi_B (\lnot L \bar b) = 1$, then $\pi_B^* (L \bar b) = \pi_B (L \bar b) = 0$, and $\pi_C (L f(\bar b))$ cannot be smaller than this. So $f$ is indeed a homomorphism from $\pi_B^*$ to $\pi_C$.
 
    Now note that $\pi_B^*(L \bar b) \geq (h_{\geq 1} \circ \pi_B) (L \bar b)$ for all $L \bar b \in \Lit_B(\tau)$. By monotonicity of $\max$ and $\min$, this inductively yields $\pi_B^*\llb \Phi \rrb \geq (h_{\geq 1} \circ \pi_B) \llb \Phi \rrb = h_{\geq 1} (\pi_B \llb \Phi \rrb) =1$. Since $\Phi$ is preserved under homomorphisms by assumption, we have that $\pi_C \llb \Phi \rrb \geq \pi_B^* \llb \Phi \rrb = 1$. With $\pi_C \succeq \pi_A$, it follows that $\pi_A \llb \Phi \rrb = \pi_C \llb \Phi \rrb =1$. Hence, we can overall conclude that $\Phi_{\exists^+} \models_{\Semi_3} \Phi$.
\end{proof}

Using compactness via entailment, we can extract a single $\Sigma_1$-sentence from $\Phi_\exists$ in the case where $\Phi$ is a singleton, and generalise the homomorphism preservation theorem from $\Semi_3$ to arbitrary lattice semirings $\calL$.

\begin{theorem} \label{cor:hom-preservation}
    Let $\calL$ be a lattice semiring.
    If $\Phi \in \FO$ is preserved under homomorphisms in $\calL$, then $\Phi \equiv_{\calL} \Psi$ for some $\Psi \subseteq \Sigma_1^+$. In the special case where $\Phi = \{\phi \}$ is a singleton, there is a single $\Sigma_1^+$-sentence $\calL$-equivalent to $\calL$.
\end{theorem}

\begin{proof}
    Let $\Phi \subseteq \FO$ be preserved under homomorphisms in $\calL$. By \cref{lem:entailment-from-S3-to-any-lattice}, $\Phi$ must be preserved under homomorphisms in $\Semi_3$.
    By \cref{thm:homomorphismPreservationMainDirection}, we know that $\Phi_{\exists^+} \models_{\Semi_3} \Phi$, and thus $\Phi_{\exists^+} \equiv_{\Semi_3} \Phi$. With \cref{lem:entailment-from-S3-to-any-lattice} this yields $\Phi_{\exists^+} \equiv_{\calL} \Phi$.
    In the case where $\Phi = \{\phi\}$ only contains a single sentence, we can apply compactness via entailment (\cref{thm:compactnessLattice}), and obtain a finite $\Phi_0 \subseteq \Phi_{\exists^+}$ such that $\Phi_0 \models_{\Semi_3} \phi$. Since $\bigwedge \Phi_0$ is a conjunction of $\Sigma_1^+$-formulae, it is $\Semi_3$-equivalent to a $\Sigma_1^+$-formula by \cref{lem:closure-under-and-or}.
\end{proof}

\subsection{The \L os-Tarski Theorem} 

We now move on from homomorphism preservation to preservation under extensions and substructures. That is, we prove a generalisation of the classical \L os-Tarski Theorem for Boolean semirings to lattice semirings. While in the Boolean setting, extension preservation and preservation under substructures are exactly dual to each other (and are thus covered by a single proof), this is no longer the case in semiring semantics. For a sentence preserved under subinterpretations in $\Semi_3$, it is a priori not clear if its negation is preserved under extensions, as in the Boolean case. Therefore, we have to prove both preservation theorems individually.  

As before, we first consider the case $\Semi_3$, and later generalise the theorems to arbitrary lattice semirings. For a set of sentences $\Phi$, let $\Phi_\exists \coloneqq \{\psi \in \Sigma_1 \mid \Phi \models_{\Semi_3} \psi\}$ and $\Phi_\forall \coloneqq \{\psi \in \Pi_1 \mid \Phi \models_{\Semi_3} \psi\}$.

\paragraph*{Extension Preservation} The extension preservation theorem can be proved analogously to the homomorphism preservation theorem by replacing atoms with literals and making sure that the mapping we get from amalgamation is injective.

\begin{lemma} \label{lem:int-for-existential-amalgamation}
    Let $\pi_A$ be an $\Semi_3$-interpretation such that $\pi_A \llb \psi \rrb = 1$ for all $\psi \in \Phi_{\exists}$. Then there is an $\Semi_3$-interpretation $\pi_B$ such that
    \begin{enumerate}[(1)]
        \item $\pi_B \llb \Phi \rrb =1$, and
        \item $\pi_B \llb \psi \rrb \leq \epsilon$ for all $\psi \in \Sigma_1$ with $\pi_A \llb \psi \rrb \leq \epsilon$.
    \end{enumerate}
\end{lemma}

\begin{proof}
    Analogous to \cref{lem-interpretation-for-amalgamation}. We only have to replace $\Psi$ with $\{\psi \in \Sigma_1 \mid \pi_A \llb \psi \rrb \leq \epsilon\}$.
\end{proof}

\begin{lemma}[Existential Amalgamation] \label{lem:existential-amalgamation}
    Let $\pi_A, \pi_B$ be $\Semi_3$-interpretations such that $\pi_B \llb \psi \rrb \leq \epsilon$ for all $\psi \in \Sigma_1$ with $\pi_A \llb \psi \rrb \leq \epsilon$. Then there is some $\pi_C \succeq \pi_A$ and an injective mapping $f \colon B \to C$ such that $\pi_C (L f (\bar b)) = 1$ for all $L \bar b \in \operatorname{Lit}_B (\tau)$ with $\pi_B (L \bar b) = 1$.
\end{lemma}

\begin{proof}
    Analogously to \cref{homomorphism-amalgamation}, one proves the satisfiability of $(\phi, \{\pi_A \llb \phi \rrb\}) \mid \phi \in \FO_A(\tau)\} \cup \{(\psi, \{1\}) \mid \psi \in \Psi\}$ where $\Psi \coloneqq \{L \bar b \in \operatorname{Lit}_B (\tau) \mid \pi_B (L \bar b) = 1\}\cup \{(b \neq b', \{1\}) \mid b \neq b' \in B'\} $.
\end{proof}

\begin{theorem}
\label{thm:extensionPreservationMainDirection}
    If $\Phi$ is preserved under extensions in $\Semi_3$, then $\Phi_{\exists} \models_{\Semi_3} \Phi$.
\end{theorem}

\begin{proof}
    Let $\pi_A \llb \Phi_{\exists} \rrb = 1$, i.e., $\pi_A \llb \psi \rrb = 1$ for all $\psi \in \Phi_{\exists}$. We want to prove that $\pi_A \llb \Phi \rrb =1$ to infer $\Phi_{\exists} \models_{\Semi_3} \Phi$.
    We apply \cref{lem:int-for-existential-amalgamation} and \cref{lem:existential-amalgamation} to obtain $\pi_B, \pi_C$, and $f$, and construct $\pi_B^*$ as in \cref{thm:homomorphismPreservationMainDirection}: For $L \bar b \in \Lit_B (\tau)$ we set $\pi_B^* (L \bar b) \coloneqq \pi_B (L \bar b)$ if $\pi_B (L \bar b) = 1$ or $\pi_B (\lnot L \bar b)=1$. Otherwise, we set $\pi_B^* (L \bar b) \coloneqq \pi_C (L f (\bar b))$.
    
    It remains to show that $f$ embeds $\pi_B^*$ into $\pi_C$. By construction, $f$ is injective.
    So let $L \bar b \in \Lit_B (\tau)$. By definition of $\pi_B^*$, we can restrict ourselves to the case $\pi_B (L \bar b) = 1$ or $\pi_B (\lnot L \bar b) = 1$. First let $\pi_B^* (L \bar b)=\pi_B (L \bar b) = 1$. We have $\pi_C (L f (\bar b)) = 1$ follows because $f$ preserves valuations of $1$ inside $\pi_B$. If $\pi_B (\lnot L \bar b) = 1$, we must have that $\pi_C (\lnot L f(\bar b)) = 1$ by construction of $f$, i.e. $\pi_C (L f(\bar b)) = 0$, and $\pi_B^* (L \bar b) = \pi_B (L \bar b) = 0 = \pi_C (L f (\bar b))$.
 
    We have $\pi_B^*(L \bar b) \geq (h_{\geq 1} \circ \pi_B) (L \bar b)$ for all $L \in \Lit_B(\tau)$ and, by monotonicity, $\pi_B^*\llb \Phi \rrb \geq (h_{\geq 1} \circ \pi_B) \llb \Phi \rrb = h_{\geq 1} (\pi_B \llb \Phi \rrb) =1$. Since $\Phi$ is preserved under extensions by assumption, we have that $\pi_C \llb \Phi \rrb \geq \pi_B^* \llb \Phi \rrb = 1$. With $\pi_C \succeq \pi_A$, it follows that $\pi_A \llb \Phi \rrb = \pi_C \llb \Phi \rrb =1$.
\end{proof}

\begin{theorem} \label{thm:extension-preservation-lattice}
    Let $\calL$ be a lattice semiring.
    If $\Phi \in \FO$ is preserved under extensions in $\calL$, then $\Phi \equiv_{\calL} \Psi$ for some $\Psi \subseteq \Sigma_1$. In the special case where $\Phi = \{\phi \}$ is a singleton, there is a single $\Sigma_1$-sentence $\calL$-equivalent to $\calL$.
\end{theorem}

\begin{proof}
 Analogous to \cref{cor:hom-preservation}.
\end{proof}

\paragraph*{Subinterpretation Preservation}

\begin{lemma} \label{lem-interpretation-for-uni-amalgamation}
    Let $\pi_A \llb \psi \rrb = 1$ for all $\psi \in \Phi_{\forall}$. Then there is an $\Semi_3$-interpretation $\pi_B$ such that
    \begin{enumerate}[(1)]
        \item $\pi_B \llb \Phi \rrb =1$, and
        \item $\pi_B \llb \psi \rrb \leq \epsilon$ for all $\psi \in \Pi_1$ with $\pi_A \llb \psi \rrb \leq \epsilon$.
    \end{enumerate}
\end{lemma}

\begin{proof}
Analogous to \cref{lem-interpretation-for-amalgamation} if one considers $\Psi \coloneqq \{\psi \in \Pi_1 \mid \pi_A \llb \psi \rrb \leq \epsilon\}$.
\end{proof}

\begin{lemma}[Universal Amalgamation] \label{universal-amalgamation}
    Let $\pi_A, \pi_B$ be $\Semi_3$-interpretations such that $\pi_B \llb \psi \rrb \leq \epsilon$ for all $\psi \in \Pi_1$ with $\pi_A \llb \psi \rrb \leq \epsilon$. Then there is some $\pi_C \succeq \pi_B$ and an injective mapping $f \colon A \to C$ such that $\pi_C (L f (\bar a)) \leq \epsilon$ for all $L \bar a \in \Lit_A (\tau)$ with $\pi_A (L \bar a) \leq \epsilon$.
\end{lemma}

\begin{proof}
    Suppose that such a $\pi_C$ does not exist, i.e. that 
    $$\{(\phi, \{\pi_B \llb \phi \rrb\}) \mid \phi \in \FO_B(\tau)\} \cup \{(\psi, \{0, \epsilon\}) \mid \psi \in \Psi\},$$
    where $\Psi \coloneqq \{L \bar a \in \Lit_A (\tau) \mid \pi_A (L \bar a) \leq \epsilon\} \cup \{(a \neq a', \{1\}) \mid a \neq a' \in A\}$, is unsatisfiable. Invoking \cref{thm:compactnessFiniteAbsorptive}, let $\Psi_0 \subseteq \Psi$ be finite such that 
    $\Delta \coloneqq \{(\phi, \{\pi_B \llb \phi \rrb\}) \mid \phi \in \FO_B(\tau)\} \cup \{(\psi, \{0, \epsilon\}) \mid \psi \in \Psi_0\}$
    is unsatisfiable. 
    Consider the set $\Psi_0^*$ that is obtained by replacing in $\Psi_0$ each constant $a_i \in A$ with a variable $x_i$. We can enumerate the variables $x_1, \dots x_k$, for some $k \in \bbN$, because $\Psi_0$ is finite.
    Then we must have that $\pi_B \llb \theta \rrb = 1$, where $\theta \coloneqq \forall x_1 \dots \forall x_k (\bigvee_{1 \leq i < j \leq k} x_i \neq x_j \vee \bigvee \Psi_0^*)$, since otherwise there would be pairwise distinct $b_1, \dots, b_k \in B$ such that $\pi_B (L \bar b) \leq \epsilon$ for all $L(x_1, \dots, x_k) \in \Psi_0^*$ and the expansion of $\pi_B$ where we interpret each constant $a_i \in A$ occurring in $\Psi_0$ with $b_i$ would be a model of $\Delta$. By the assumption of the lemma, $\pi_B \llb \theta \rrb = 1$ implies $\pi_A \llb \theta \rrb = 1$. But this contradicts the definition of $\Psi$.
\end{proof}

\begin{theorem} \label{thm:SubinterpretationPreservationMainDirection}
    If $\Phi$ is preserved under subinterpretations in $\Semi_3$, then $\Phi_{\forall} \models_{\Semi_3} \Phi$.
\end{theorem}

\begin{proof}
     Let $\pi_A \llb \Phi_{\forall} \rrb = 1$, i.e. $\pi_A \llb \psi \rrb = 1$ for all $\psi \in \Phi_{\forall}$. By \cref{lem:entailment-from-S3-to-any-lattice}, it suffices to prove that $\pi_A \llb \Phi \rrb =1$. Let $\pi_B$ be as in \cref{lem-interpretation-for-uni-amalgamation} and $\pi_C$ be the corresponding amalgamated $\Semi_3$-interpretation according to \cref{universal-amalgamation}. We know that $\pi_B \llb \Phi \rrb =1$ and due to $\pi_B \preceq \pi_C$, this implies $\pi_C \llb \Phi \rrb =1$. We want to use the fact that $\Phi$ is preserved under subinterpretations to infer $\pi_A \llb \Phi \rrb =1$. However, the mapping $f$ we obtain from \cref{universal-amalgamation} need not be an embedding. Therefore, we manipulate $\pi_C$ to obtain an $\Semi_3$-interpretation $\pi_C^*$ with the same universe, which still has the property that $\pi_C^* \llb \Phi \rrb =1$, but into which $f$ embeds $\pi_A$. We construct $\pi_C^*$ as follows:
     \[
     \pi_C^* (L \bar c)=\left\{\begin{array}{ll} \pi_C (L \bar c) & \text{ if } \pi_C (L \bar c)=1 \text{ or } \pi_C (\lnot L \bar c)=1 \\
         \pi_A (L f^{-1}(\bar c)) & \text{ if } 1 \not\in \{\pi_C (L \bar c), \pi_C (\lnot L \bar c)\} \text{ and } \bar c \in f(A) \\
         1 & \text{ otherwise, if $L$ is positive} \\
         0 & \text{ otherwise, if $L$ is negative}
        \end{array}\right. .
  \]

  \begin{claim}
      $\pi_C^* \llb \Phi \rrb = 1$.
  \end{claim}
  \begin{claimproof}
        Let $h_{\geq 1} \colon \Semi_3 \to \Semi_3$ with $h_{\geq 1} (1)=1$ and $h_{\geq 1} (\epsilon)=h_{\geq 1} (0)=0$. By the fundamental property (\cref{fundamental-property}), $(h_{\geq 1} \circ \pi_C) \llb \Phi \rrb = h_{\geq 1} (\pi_C \llb \Phi \rrb) = h_{\geq 1} (1) = 1$. By definition of $\pi_C^*$, we have $\pi_C^* (L \bar c) \geq (h_{\geq 1} \circ \pi_C) (L \bar c)$ for all $L \in \Lit_C(\tau)$. Using monotonicity of $\max$ and $\min$, this inductively extends to $\pi_C^* \llb \Phi \rrb \geq (h_{\geq 1} \circ \pi_C) \llb \Phi \rrb =1$.
  \end{claimproof}

    Let $\pi_{f(A)}$ denote the $\Semi$-interpretation on universe $f(A)$ with $\pi_{f(A)}(L\bar{c}) = \pi_A(Lf^{-1}(\bar c))$ for each literal $L \bar c$.

    \begin{claim}
    $\pi_{f(A)} \subseteq \pi_C^*$.
  \end{claim}
  \begin{claimproof}
        We only have to consider literals $L \bar a$ in $\pi_A$ with $\pi_C (L f(\bar a)) =1$ or $\pi_C (\lnot L f(\bar a)) =1$ (all remaining cases are by definition of $\pi_C^*$). Suppose that $\pi_C^* (L f(\bar a))=\pi_C (L f(\bar a)) =1$. By construction of $f$, this implies $\pi_A (L \bar a) =1$. If otherwise $\pi_C^* (\lnot L  f(\bar a))=\pi_C (\lnot L f(\bar a)) =1$, then $\pi_A (\lnot L \bar a) =1$, i.e. $\pi_A (L \bar a)= 0 = \pi_C^* (L f(\bar a))$.
  \end{claimproof}

  Because $\Phi$ is preserved under subinterpretations, this yields $\pi_A \llb \Phi \rrb =  \pi_{f(A)} \llb \Phi \rrb \geq \pi_C^* \llb \Phi \rrb = 1$, and overall $\Phi_\forall \models_{\Semi_3} \Phi$.
\end{proof}

\begin{theorem} \label{thm:subinterpretation-preservation-lattice}
    Let $\calL$ be a lattice semiring.
    If $\Phi \in \FO$ is preserved under subinterpretations in $\calL$, then $\Phi \equiv_{\calL} \Psi$ for some $\Psi \subseteq \Pi_1$. In the special case where $\Phi = \{\phi \}$ is a singleton, there must be a single $\Pi_1$-sentence $\calL$-equivalent to $\calL$.
\end{theorem}

\begin{proof}
 Analogous to \cref{cor:hom-preservation} and \cref{thm:extension-preservation-lattice}.
\end{proof}

\section{Extension Preservation in the Finite}\label{section:finiteExtPresIdea}

The classical \L os--Tarski theorem is known to fail in finite model theory, when both preservation and equivalence are restricted to finite structures. 
However, unlike for other model-theoretic properties, this does not immediately imply the failure of the theorem for other semirings, as sentences that are preserved under extensions in the Boolean sense may not be preserved in other semirings as well. 
Therefore, we study the question of whether there are semirings in which the extension preservation theorem remains true when restricted to finite universes.
It turns out that in semiring semantics, the situation is much more diverse than in the Boolean world. In some cases, e.g. if we consider the natural semiring, we can easily adapt the counterexample used to disprove the extension preservation theorem in the general setting.

\begin{theorem}
	There is a sentence $\psi$ which is extension preserved on finite $\Nat$-interpretations, but, on finite $\Nat$-interpretations, it is not equivalent to any $\phi \in \Sigma_1$.
\end{theorem}

For other semirings, this is different. Despite the failure in the Boolean, the extension preservation theorem holds in the finite for the Viterbi and the \L ukasiewicz semiring (and thus the tropical semiring and the semiring of doubt), and every lattice semiring.
In this section, we discuss the key concepts and main proof ideas, and refer to \cref{sec:extension-preservation-finite-proof} and \cref{sec:FinExtPresLattices} for a formal definition of all notions and proofs.

\subsection{The Viterbi and \L ukasiewicz Semiring}

Recall the counterexample from \cref{sec:counterexamples}. It relied on the fact that on infinite $\Vit$-interpretations, the sentence $\psi =\exists x  \forall y Rx$ is preserved under extensions.
If we restrict ourselves to finite $\Vit$-interpretations, this is no longer true as the universal quantifier generates a power to the size of the universe. 
Consider for example the $\Vit$-interpretation $\pi_1$ of universe $\{a\}$ with $\pi_1(Ra)=1/2$ and its extension $\pi_2$ over the universe $\{a, b\}$ with $\pi_2(Rb)=1/2$, where we have $\pi_1 \llb \psi \rrb = 1/2 > 1/4 = \pi_2 \llb \psi \rrb$. The reason why $\psi$ is not extension-preserved is that multiplication is strictly decreasing in $\Vit \setminus \{0,1\}$. 

We use this property to show that universal quantifiers always make it impossible for a sentence to be extension-preserved in the finite, except when: Either the universal quantifier is \emph{trivial} and always evaluates to $1$ (such as in $\forall x (x=x)$) and we can replace it with $\top$, or it does not actually contribute anything to the semantics of the formulae. As an example for the latter case, consider the sentence $\exists x Rx \vee \forall x Rx$. While addition is increasing, multiplication is decreasing in $\Vit$, so we have $\pi \llb \exists x Rx \rrb \geq \pi \llb \forall x Rx \rrb$ for all $\Vit$-interpretations $\pi$, and we can replace the universal subformula with $\bot$. 
Based on this idea, we can show that the extension preservation theorem holds not just for $\Vit$ in the finite, but also for $\Trop \cong \Vit$ and $\Lukas \cong \Doubt$.

\begin{restatable}{theorem}{extensionPresFinite}
\label{thm:viterbiFinExtPres}
	Let $\Semi \in \{\Vit, \Trop, \Lukas, \Doubt\}$. A sentence $\psi \in \FO$ is preserved under extensions on finite $\Semi$-interpretations if, and only if, there is some $\phi \in \Sigma_1$ which, on finite $\Semi$-interpretations, is equivalent to $\psi$.
\end{restatable}

The implication from syntax to semantics is covered by \cref{lem:syntax-to-semantics}. It remains to prove the converse. Consider the case $\Semi=\Vit$.
The main steps of the proof are as follows.

\begin{enumerate}[(1)]
    \item Reduce the problem of equivalence to a $\Sigma_1$-sentence to sufficiently large $\Vit$-interpretations.
    \item Track the contribution of subformulae to the semantics of a formula by means of so-called evaluation strategies -- these are equivalent to strategies of the Verifier in the classical model checking game for FO.
    \item Show that for extension preserved sentences and large $\Vit$-interpretations, evaluation strategies that use a subformula $\forall y \phi (\bar a, y)$ are always dominated by ones that avoid universal quantifiers. 
    \item Remove all strategies using a subformula $\forall y \phi (\bar a, y)$ and rewrite the formula without universal quantifiers.
\end{enumerate}

\noindent\textbf{Equivalence on large $\Vit$-interpretations.} Given an extension-preserved sentence $\psi$, our goal is to construct an equivalent existential sentence. The first step is to argue that it suffices to find an existential sentence $\phi$ which is equivalent on $\Vit$-interpretations that have at least a certain cardinality. By replacing each universal quantifier of $\psi$ with an $n$-ary conjunction, we can always find an existential sentence that is equivalent to $\phi$ on $\Vit$-interpretations of fixed size $n$. Using extension preservation of $\psi$ and the fact that addition is the maximum with respect to the natural order,\footnote{Note that this is true for $\Trop$ and $\Doubt$ too, where the natural order is the reverse order of $\R$.} we can combine such formulae with $\phi$ to obtain an existential sentence that is equivalent to $\psi$ on \emph{all} finite $\Vit$-interpretations.\\

\noindent\textbf{Strategies and their valuation.} The semiring semantics of a sentence can be defined via the notion of \emph{evaluation strategies}. We make this formal later on. Intuitively, an evaluation strategy $\Tt$ for a sentence $\psi$ and an $\Semi$-interpretation $\pi$ is given by specifying for each subformula $\phi_1 \lor \phi_2$ of $\psi$, whether we would like to evaluate $\phi_1$ or $\phi_2$, and for each instantiated subformula $\exists x \phi(x, \bar a)$, an instantiation of $x$ with an element of the universe. 
Then the value of the strategy $\pi \llb \mathcal{T}\rrb$ is the product over all $\pi \llb L \bar a \rrb$ such that $L \bar a$ is an instantiated literal in $\psi$ that can be reached in the classical FO \emph{model-checking game} in a play consistent with the (Verifier) strategy $\Tt$. Essentially, these are all literals in subformulae of the disjuncts chosen by $\Tt$ with all variable instantiations consistent with the existential ones chosen by $\Tt$. 
\begin{theorem}[Sum-of-Strategies \cite{GraedelTan24}] \label{lem:sum-of-proof-trees-1}
	For every semiring $\Semi$ and every $\Semi$-interpretation $\pi$ and $\psi \in \FO$ we have
	$\pi \llb \psi \rrb = \sum \{ \pi \llb \mathcal{T}\rrb  \mid \Tt \text{ is an evaluation strategy for $\psi$ and $\pi$}\}$.
\end{theorem}

Because addition is the maximum operation in $\Vit$, by the above theorem, the valuation $\pi \llb \psi \rrb$ must coincide with the maximum valuation of any evaluation strategy. We refer to such a strategy as \emph{optimal} for $\pi$ and $\psi$. We call a strategy \emph{existential} if it does not use any subformula $\forall y \phi (\bar a, y)$. The central claim we prove is that extension preservation implies that for any sufficiently large $\Vit$-interpretation $\pi$ there must always be an existential strategy which is optimal for $\pi$ and $\psi$.
\\

\noindent\textbf{Existence of existential optimal strategies.} To prove this claim, we consider the contraposition. Starting from a strategy $\Tt$ that is optimal for a $\Vit$-interpretation $\pi$ which uses a subformula $\forall y \phi (\bar a, y)$, we aim to disprove extension preservation of $\psi$.
Essentially, we make use of the fact that multiplication is strictly decreasing on $\Vit \setminus \{0,1\}$ to construct a strategy $\Tt^*$ for a subinterpretation $\pi^*$ of $\pi$ such that $\pi^*\llb \Tt^* \rrb > \pi \llb \Tt \rrb$. By \cref{lem:sum-of-proof-trees-1} and the optimality of $\Tt$, this shows that $\psi$ is not extension-preserved because $\pi^* \subseteq \pi$, but $\pi^*\llb \psi \rrb > \pi \llb \psi \rrb$. Let $A$ denote the universe of $\pi$ and $A^* \subset A$ the universe of $\pi^*$.
The strategy $\Tt^*$ is constructed from $\Tt$ by restricting it to instantiations of $y$ with elements of $A^*$, but there are some technical intricacies to this: The strategy $\Tt$ may choose to instantiate existentially quantified variables with elements in $A \setminus A^*$, so we have to show that we can find alternative choices in $A^*$ without changing the valuation of the strategy too much. We refer to this construction of $\Tt^*$ from $\Tt$ as \emph{strategy translation}, and this is the key new technique that we develop for this proof.\\

 \noindent\textbf{Omitting universal quantifiers.} The final argument is to show that, given that we now know an optimal existential strategy for $\psi$ exists for every $\pi$, we can syntactically remove universal subformulae from $\psi$. More precisely, we replace every subformula $\forall y \phi (\bar a, y)$ with $\top$ if it always evaluates to $1$, and otherwise with $\bot$.

\subsection{Lattice Semirings}

 The central ingredient to the proof of \cref{thm:viterbiFinExtPres} is the existence of existential optimal strategies in the presence of finite extension preservation. Note that this argument really relied on the specific algebraic properties of the Viterbi and \L ukasiewicz semiring and does not apply to the Boolean. In fact, it does not hold for any lattice semiring $\calL$, in which multiplication is idempotent. It is not just that the existence of any optimal strategy is not guaranteed in the case where $\calL$ is not linearly ordered, even for min-max semirings it may be that all optimal strategies are non-existential. The sentence $\psi = \forall y \exists z Rz \equiv_\calL \exists z Rz$, for example, is finite extension preserved in $\calL$. However, it clearly does not admit \emph{any} existential strategy, so, in particular, not an optimal one. 

 Surprisingly, by combining the techniques from \cref{sec:lattices} and \cref{section:finiteExtPresIdea}, we can prove that the finite extension preservation theorem also holds for every lattice semiring apart from the Boolean.

 \begin{restatable}{theorem}{FiniteExtPresLattices} \label{thm:extension-preservation-lattice}
    Let $\calL \not\cong \Bool$ be a lattice semiring.
	A sentence $\psi \in \FO$ is finite extension preserved in $\calL$ if, and only if, on finite $\calL$-interpretations, it is equivalent to an existential sentence.
\end{restatable}

To prove \cref{thm:extension-preservation-lattice}, we reuse the reduction method from \cref{sec:lattices} and restrict ourselves to the case where we evaluate in the \emph{fuzzy semiring} $\Fuzzy=([0,1], \max, \min, 0, 1)$. 
The structure of the proof is similar to the one from \cref{section:finiteExtPresIdea}: The first two steps are just the same. Due to its linear order, the notion of optimal strategies also applies to $\Fuzzy$. The existence of existential optimal strategies in step 3, however, has to be replaced by a weaker statement. Intuitively, the reason why finite extension preservation in $\Fuzzy$ might still hold in the presence of universal quantifiers is that strategies might use subformulae $\forall y \phi (\bar x, y)$, however, without actually relying on a literal that contains the universally quantified variable~$y$. In the example above, it is easy to rewrite $\psi$ as an existential sentence: Because no strategy uses a literal containing $y$, we can simply delete the universal quantifier. A more complicated case occurs when strategies use a literal that contains $y$ for some instantiations, but not for all of them, which might happen if we consider $\psi = \forall y (\exists z Rz \vee \exists z (Rz \wedge Qy))$, for example.

Therefore, we consider \emph{almost existential} strategies, where whenever a subformula $\forall y \phi (\bar a,y)$ is used, at least one substrategy for some $\phi (\bar a, b)$ \emph{does not} use any relational literal that contains $b$, and
the main claim we prove is that finite extension preservation in $\Fuzzy$ guarantees the existence of almost existential strategy optimal for sufficiently large $\Fuzzy$-interpretations.
The proof requires a strategy translation lemma similar to the one for $\Vit$ and $\Lukas$. 
Knowing just the existence of almost, but not necessarily fully, existential optimal strategies, we cannot replace all universal subformulae by $\bot$ right away as done in the Viterbi and \L ukasiewicz case. Instead, in step 4 of the proof, we will first manipulate a given extension preserved sentence by basic logical equivalences to make sure that universal subformulae only occur at places where they do not admit almost existential strategies. This will allow us in the end to argue that such universal subformulae must be redundant and can be removed from the formula without changing its semantics. In the example above, we first transform $\psi = {\color{blue}\forall y (}\exists z Rz \vee \exists z (Rz \wedge Qy){\color{blue})}$ into $\exists z Rz \vee {\color{blue}\forall y} \exists z (Rz \wedge Qy)$, which is $\Fuzzy$-equivalent by continuity, and then argue that $\psi \equiv_\Fuzzy^\omega \exists z Rz \vee \bot \equiv_\Fuzzy  \exists z Rz$ as $\forall y \exists z (Rz \wedge Qy))$ can never be used in an almost existential strategy.

\section{Proof of the Extension Preservation Theorem in the Finite for the Viterbi and \L ukasiewicz Semiring} \label{sec:extension-preservation-finite-proof}

In this section, we provide the proof details of \cref{thm:viterbiFinExtPres}. Recall that we aim to show the following: Given a sentence $\psi \in \FO$ that is preserved under extensions on \emph{finite} $\Semi$-interpretations, for $\Semi \in \{ \Vit, \Trop, \Lukas, \Doubt \}$, we can eliminate all universal quantifiers in it. That is, $\psi$ is equivalent to a $\Sigma_1$-sentence on finite $\Semi$-interpretations.

\subsection{Reducing the Problem to Large Universes}

First, we want to justify that it suffices to construct a $\Sigma_1$-formula that is equivalent to $\psi$ on sufficiently large $\Semi$-interpretations, as we can always hardcode the semantics of universal quantifiers on interpretations of a fixed size.

\begin{Definition}
	For a semiring $\Semi$ and sentences $\phi, \theta$, we write $\phi \equiv^n_{\Semi} \theta$ if $\pi \llb \phi \rrb = \pi \llb \theta \rrb$ for all $\Semi$-interpretations with universe of size $n$. Similarly, we write $\phi \equiv^{\leq n}_{\Semi} \theta$ (and $\phi \equiv^{\geq n}_{\Semi} \theta$) is the above is satisfied by each \emph{finite} $\Semi$-interpretation with universe of size at most (at least) $n$. If $\pi \llb \phi \rrb = \pi \llb \theta \rrb$ for all \emph{finite} $\Semi$-interpretations, we write $\phi \equiv^\omega_{\Semi} \theta$.
\end{Definition}

For a sentence $\psi \in \FO(\tau)$ and $n \geq 1$, let $\bar x=(x_1,\dots,x_n)$ be a tuple of new variables, not occurring in $\psi$, and let $\psi_n \coloneqq \exists x_1 \dots \exists x_n (\bigwedge_{1 \leq i < j\leq n} x_i \neq x_j \wedge \psi_n^* (\bar x))$ where $\psi_n^*(\bar x)$ is defined inductively by replacing each subformula $\E y\phi$ by $\bigvee_{i \in [n]} (\phi[y / x_i])_n^*(\bar x)$ 
and each subformula $\A y\phi$ by $\bigwedge_{i \in [n]} (\phi[y / x_i])_n^*(\bar x)$.

\begin{lemma} \label{lem-subint-size-n}
	Let $\Semi$ be a semiring, $\pi$ be an $\Semi$-interpretation and $\psi \in \FO$. Then $\pi \llb \psi_n \rrb = \sum \{\pi^* \llb \psi \rrb \mid \pi^* \subseteq \pi \text{ of size } n\}$. In particular, $\psi \equiv_\Semi^n \psi_n$.
\end{lemma}

\begin{proof}
	Let $\pi$ be a $\Semi$-interpretation and $\bar a$ be a tuple of distinct elements of its universe. By induction, it follows 
	that $\pi \llb \psi^*_n (\bar a) \rrb = \pi^* \llb \psi \rrb$ where $\pi^*$ is the subinterpretation of $\pi$ which is induced by~$\bar a$.
\end{proof}

\begin{lemma} \label{lem-subint-size-at-most-n}
	Let $\Semi$ be additively idempotent and $\psi \in \FO(\tau)$ be preserved under extensions in $\Semi$. Then $\psi \equiv^{\leq n}_\Semi \bigvee_{1 \leq i \leq n} \psi_i$.
\end{lemma}

\begin{proof}
	We know that in additively idempotent semirings, addition is the same as supremum.
	Let $\pi$ be a $\Semi$-interpretation whose universe has size at most $n$. By \cref{lem-subint-size-n}, we have
	$
	\pi \llb \bigvee_{i \leq n} \psi_i \rrb = \bigsqcup \{\pi^* \llb \psi \rrb \mid \pi^* \subseteq \pi\} \overset{*}= \pi \llb \psi \rrb,
	$
	where $(*)$ is due to preservation under extensions.
\end{proof}

\begin{lemma} \label{lem:large-universes-suffice}
	Let $\Semi$ be additively idempotent.
	If $\psi$ is preserved under extensions in $\Semi$ and there is some $\phi \in \Sigma_1$ and $n \in \omega$ such that $\phi \equiv^{\geq n}_\Semi \psi$, then there must be some $\phi^* \in \Sigma_1$ such that $\phi^* \equiv^\omega_\Semi \psi$.
\end{lemma}

\begin{proof}
	We claim that $\psi \equiv \exists x_1 \dots \exists x_n (\bigwedge_{1 \leq i < j \leq n} x_i \neq x_j \wedge \phi) \vee \bigvee_{1 \leq i \leq n} \psi_i =: \theta$. Let $\pi$ be an $\Semi$-interpretation of cardinality less than $n$. Then $\pi \llb \theta \rrb = \pi \llb \bigvee_{1 \leq i \leq n} \psi_i \rrb \overset{*}= \pi \llb \psi \rrb$, where $(*)$ is by \cref{lem-subint-size-at-most-n}. If $\pi$ is of size at least $n$, then 
	\[
	\pi \llb \theta \rrb = \pi \llb \phi \rrb \sqcup \pi \llb \bigvee_{1 \leq i \leq n} \psi_i \rrb \overset{(1)}{=} \bigsqcup \{\pi \llb \psi \rrb \} \cup \{\pi^* \llb \psi \rrb \mid \pi^* \subseteq \pi \text{ of size at most } n \} \overset{(2)}{=} \pi \llb \psi \rrb,
	\]
	where $(1)$ is by \cref{lem-subint-size-n} and $(2)$ follows from preservation under extensions.
	Since $\theta$ does not contain a universal quantifier, we can apply \cref{lem:closure-under-and-or} and obtain $\theta \equiv_\Semi^\omega \phi^*$ for some $\phi^* \in \Sigma_1$.
\end{proof}

Hence, our goal is to construct for a given sentence $\psi$ which is preserved under extensions on finite $\Semi$-interpretations some $\phi \in \Sigma_1$ such that $\psi \equiv^{\geq n}_\Semi \phi$ for some $n \in \omega$.

\subsection{Tracking Redundancies via Evaluation Strategies}

In order to avoid having to consider different equality types of instantiations, we use the variant $\FO^{\neq}$ of $\FO$ where quantifiers only range over instantiations that do not occur in the instantiations of the free variables. More precisely, formulae from $\FO^{\neq}$ only contain relational atoms, but no equality atoms, and rather than the usual quantifiers, they may contain quantifiers $\exists^{\neq}, \forall^{\neq}$ with the semantics
\[
\pi \llb \exneq y \phi (\bar a, y)\rrb \coloneqq \sum_{b \in A\setminus \bar a} \pi \llb \phi (\bar a, b) \rrb \;\; \text{ and } \;\; \pi \llb \foraneq y \phi (\bar a, y)\rrb \coloneqq \prod_{b \in A\setminus \bar a} \pi \llb \phi (\bar a, b) \rrb.
\]
This way, we only have to consider tuples of instantiations that are pairwise distinct. Throughout this section we implicitly assume this when we write $\bar a \subseteq A$. By $\Sigma_1^{\neq}$ we denote the existential fragment of $\FO^{\neq}$.  According to the following lemma, $\FO$ and $\FO^{\neq}$ (as well as $\Sigma_1$ and $\Sigma_1^{\neq}$) have the same expressive power, which is why we can assume that all formulae considered in the following are in $\FO^{\neq}$.

\begin{lemma} \label{lem:FOneq}
	For every $\phi \in \FO$ ($\phi \in \Sigma_1$) there is some $\phi^* \in \FO^{\neq}$ ($\phi^* \in \Sigma_1^{\neq}$)  such that $\phi \equiv \phi^*$ and vice versa. 
\end{lemma}

\begin{proof}
	Inductively translate $\phi \in \FO$ into $\phi^* \in  \FO^{\neq}$ as follows.
	\begin{itemize}
		\item $(x_i = x_i)^* \coloneqq \top$ and $(x_i = x_j) \coloneqq \bot$ for $i \neq j$. Analogously for inequalities.
		\item $(L \bar x)^* \coloneqq L\bar x$ for relational atoms $L \bar x$.
		\item $(\psi \circ \theta)^* \coloneqq \psi^* \circ \theta^*$ for $\circ \in \{\vee, \wedge\}$.
		\item $(\exists y \psi (x_1, \dots, x_n, y))^* \coloneqq \bigvee_{i \in [n]} \psi^* (x_1, \dots, x_n, x_i) \vee \exneq y \psi^* (x_1, \dots, x_n, y)$. Analogously for universal quantifiers.
	\end{itemize}
	Since neither of the translation steps introduces new universal quantifiers, we have $\phi \in \Sigma_1$ if, and only if, $\phi^* \in \Sigma_1^{\neq}$.
	Conversely, we can translate $\phi \in \FO^{\neq}$ into  $\phi^* \in \FO$ according to:
	\begin{itemize}
		\item $(L \bar x)^* \coloneqq L\bar x$ for relational atoms $L \bar x$.
		\item $(\psi \circ \theta)^* \coloneqq \psi^* \circ \theta^*$ for $\circ \in \{\vee, \wedge\}$.
		\item $(\exists^{\neq} y \psi (x_1, \dots, x_n, y))^* \coloneqq \exists y (\bigwedge_{i \in [n]} y \neq x_i \wedge \psi^* (x_1, \dots, x_n, y))$. Analogously for universal quantifiers.
	\end{itemize}
	As before we have $\phi \in \Sigma_1^{\neq}$ if, and only if, $\phi^* \in \Sigma_1$.
\end{proof}

To formally define what it means for a subformula $\phi (\bar x)$ not to contribute to the semantics of a sentence $\psi$, we make use of evaluation strategies in the model checking game and track whether or not they use $\phi (\bar x)$.
The \emph{model checking game graph for $\psi$ and $n \geq \qr(\psi)$} is a labelled tree $C_n(\psi) = (V, E, \lambda)$ where each node $v$ is labelled by a subformula $\lambda (v) = \phi (\bar a)$ of $\psi$ which is instantiated with elements from $[n]$. We inductively define $C_n(\phi (\bar a))$ for $\phi (\bar a) \in \FO^{\neq}_{[n]} (\tau)$ and $n \geq \qr(\psi)$ in negation normal form as follows. In any case, the root of $C_n(\phi (\bar a))$ is labelled $\phi (\bar a)$. If $\phi (\bar a)$ is a literal or a Boolean constant, $C_n(\phi (\bar a))$ consists of a single node. For $\phi (\bar a) = \phi_0 (\bar a) \star \phi_1 (\bar a)$ where $\star \in \{\vee, \wedge\}$, in $C_n (\varphi (\bar a))$ we append $C_n (\phi_i (\bar a))$ for $i \in \{0,1\}$ as subtrees to the root. Finally, if $\phi (\bar a)$ has the form $Q x \vartheta (\bar a, x)$ where $Q \in \{\exists, \forall\}$ we now append the trees $C_n (\vartheta (\bar a, b))$ for each $b \in [n] \setminus \bar a$ to the root.

\begin{Definition}
	An \emph{evaluation strategy} (or just \emph{strategy})
    $\mathcal{T}=(V, E, \lambda)$ for $\psi$ and $n \geq \qr(\psi)$ is a subtree of $C_n (\psi)$ where each node
    labelled with $\phi_0 (\bar a) \vee \phi_1 (\bar a)$ or $\exists x \vartheta (\bar a, x)$ has exactly one successor and each node labelled with $\phi_0 (\bar a) \wedge \phi_1 (\bar a)$ or $\forall x \vartheta (\bar a, x)$ contains all successors in $C_n (\psi)$. We write $\Tt \in C_n (\psi)$  and define the valuation $\pi \llb \mathcal{T}\rrb$ in an $\Semi$-interpretation $\pi$ of universe $[n]$ as the products of all valuations $\pi (L(\bar a))$ of the leaves in $\mathcal{T}$. 
\end{Definition}

In every linearly ordered semiring $\Semi$ such as $\Vit$ and $\Lukas$, a $\Semi$-interpretation induces a preorder $\preceq_{\pi}$ on $C_n (\psi)$, where we put $\Tt \preceq_\pi \Tt'$ whenever $\pi \llb \Tt \rrb \leq \pi \llb \Tt' \rrb$. Such a preorder corresponds to a linear order of equivalence classes. 
According to the following theorem, the evaluation of a formula in a semiring where addition corresponds to the maximum always coincides with the valuation of any maximal strategy with respect to this order. We refer to such strategies as \emph{optimal}.

\begin{theorem}[Sum-of-Strategies \cite{GraedelTan24}] \label{lem:sum-of-proof-trees}
	For every $\Semi$-interpretation $\pi$ with universe $[n]$ and every instantiated formula $\psi (\bar a) \in \FO^{\neq}_{[n]}$ it holds that
	$\pi \llb \psi \rrb = \sum \{ \pi \llb \mathcal{T}\rrb  \mid \Tt \in C_n (\psi)\}$.
\end{theorem}

\begin{Definition}
	Let $\Tt \in C_n (\psi)$ be a strategy. We say that $\Tt$ is \emph{optimal} for an $\Semi$-interpretation~$\pi$ if $\pi \llb \psi \rrb = \pi \llb \Tt \rrb$.
	A subformula $\phi(\bar x)$ of $\psi$ is \emph{used} in $\Tt$ if $\Tt$ contains a node $v$ where $\lambda (v) = \phi(\bar a)$ for some $\bar a \subseteq [n]$ (recall that the entries of such a tuple must be pairwise distinct). If $\Tt$ does not use any subformula $\foraneq y \phi (\bar x,y)$, it is referred to as \emph{existential}. Otherwise, we call it \emph{non-existential}.
\end{Definition}

Our goal is to prove that, in the presence of finite extension preservation, non-existential strategies will always be dominated by existential ones.
The only exception we will have to exclude is when a subformula $\foraneq y \phi (\bar x, y)$ is \emph{trivial} in the following sense.

\begin{Definition}
	\label{def:trivial}
    Let $\Semi$ be a semiring and $n \in \Nat$.
	A formula $\phi (\bar x)$ is \emph{trivial for $n$} (w.r.t. $\Semi$) if for all $\Semi$-interpretations $\pi$ of universe $A$ with exactly $n$ elements and all $\bar a \subseteq A$, it holds that $\pi \llb \foraneq y \phi(\bar a, y) \rrb = 1$.
    We say that $\phi(\bar x)$ is \emph{trivial} (w.r.t. $\Semi$) if there is some $n_0$ such that $\phi(\bar x)$ is trivial for each $n \geq n_0$.
\end{Definition}

For example, the sentence $\foraneq x \exneq y (\top \vee Rx) \equiv^{\geq 2}_{\Semi} \top$ is trivial, but not trivial for $n=1$, while $\exneq x (Rx \vee \lnot Rx)$ is not trivial for any $n$.
By definition, trivial subformulae $\foraneq y \phi (\bar x, y)$ of a sentence $\psi$ can be replaced by $\top$ without changing its semantics on sufficiently large $\Semi$-interpretations. 
Now the idea is that if we replace a subformula $\foraneq y \phi (\bar x, y)$ in $\psi$ with $\bot$, there is a one-to-one correspondence between the strategies for $\psi$ that do not use $\foraneq y \phi (\bar x, y)$ and the $\bot$-free strategies of the resulting sentence $\vartheta$. 
So if we are sure that for any $\Semi$-interpretation, there is at least one optimal strategy which does not use $\foraneq y \phi (\bar x, y)$, this replacement will not change the semantics of $\psi$.
Hence, the main claim we will prove in \cref{thm:only-redundant-subform} is that once occurrences of trivial subformulae $\foraneq y \phi (\bar x, y)$ are eliminated, all remaining non-trivial subformulae $\foraneq y \phi (\bar x, y)$ no longer contribute anything to the semantics of a finite extension preserved sentence.
Formally, we say such universal subformulae are \emph{redundant}:

\begin{Definition}[Redundancy]
	\label{def:redundancy}
	Let $\Semi$ be a semiring and $\psi \in \FO^{\neq}$ be a sentence. We say that \emph{$\forall$ is redundant in $\psi$} (w.r.t $\Semi$) if there is some $n_0$ such that for every $\Semi$-interpretation $\pi$ of size $n \geq n_0$, there is an existential strategy that is optimal for $\pi$ and $\psi$. An $\Semi$-interpretation $\pi$ \emph{disproves redundancy of $\forall$ in $\psi$} (or $\pi$ is a \emph{counterexample to redundancy of $\forall$ in $\psi$}) if every strategy optimal for $\pi$ and $\psi$ is non-existential.
\end{Definition}

As an example, $\forall$ is redundant in both $\foraneq x Rx \vee \exneq x Rx$ and $\foraneq x Rx \vee \exneq x \exneq y Rx$, while it is not in $\foraneq x Rx \wedge \exneq x Rx$.
Note that if $\forall$ is not redundant in $\psi$, then there are arbitrarily large counterexamples to redundancies of $\forall$ in $\psi$.  We omit the semiring $\Semi$ in these definitions whenever it is clear from the context. 

When we want to reason about trivial formulae and redundancies, we do not need to track the semantics of a formula on just one fixed $\Semi$-interpretation, but argue about all of them at the same time. This is why we evaluate strategies not just on $\Semi$-interpretations but make use of interpretations over semirings of polynomials.
The variables we use to annotate the literals serve as placeholders and due to universality, valuations in such polynomial semirings represent the valuations in all $\Semi$-interpretations.
Because we only consider absorptive semirings in this section, we make use of absorptive polynomials for this purpose, and to reflect the fact that for complementary literals, one of them has value $0$ in model-defining interpretations, we take a quotient based on a duality between variables. See \cite{GraedelTan24} for more details.

\begin{Definition}
	For each $n \in \Nat$, let $X_n \coloneqq X_n^+ \cup X_n^-$ where $X_n^+ \coloneqq \{x_\alpha \mid \alpha \in \operatorname{Atoms}_{[n]} (\tau) \}$ and $X_n^- \coloneqq \{x_{\lnot\alpha} \mid \alpha \in \operatorname{Atoms}_{[n]} (\tau) \}$. 
    The free absorptive semiring $\Sorb(X_n)$ consists of $0, 1$, and all antichains of monomials with respect to the absorption
order $\succeq$, which is also the natural order of $\Sorb(X_n)$. A monomial $m_1$ \emph{absorbs} $m_2$, denoted $m_1 \succeq m_2$, if it has smaller exponents, i.e. $m_2 = m \cdot m_1$ for some monomial $m$. 
    The semiring $\Sorb (X_n^+, X_n^-)$ is the quotient of $\Sorb (X_n^+ \cup X_n^-)$ based on the congruence generated by $x_\alpha \cdot x_{\lnot \alpha} = 0$ for $\alpha \in \operatorname{Atoms}_{[n]}$.
\end{Definition}

We say that a variable assignment $f \colon X_n \to \Semi$ is \emph{consistent} if $f (x_\alpha)\cdot f  (x_{\lnot \alpha})= 0$ for all $ \alpha \in \operatorname{Atoms}_{[n]} (\tau)$. It is \emph{model-defining} if for all $\alpha \in \operatorname{Atoms}_{[n]}$ we have $f(x)=0$ for exactly one $x \in \{x_\alpha, x_{\lnot \alpha}\}$.  Note that model-defining assignments are in particular consistent, while the converse only holds if $\Semi$ does not contain divisors of $0$, so fails for the semirings $\Lukas$ and $\Doubt$, for example.

\begin{proposition}[Universal Property \cite{GraedelTan24}]
	For any absorptive semiring $\Semi$ and consistent $f \colon X_n \to \Semi$ there exists a unique semiring homomorphism $h\colon \Sorb (X_n^+, X_n^-) \to \Semi$ such that $h(x)= f (x)$ for all $x \in X_n$.
\end{proposition}

For every $n \in \Nat$, we consider the $\Sorb (X_n^+, X_n^-)$-interpretation $\pi_n$ over universe $[n]$ where $\pi_n (\alpha) = x_\alpha$ for $\alpha \in \operatorname{Lit}_{[n]} (\tau)$. Each such $\pi_n$ evaluates a strategy $\Tt \in C_n (\psi)$ to a monomial that describes which atomic facts are used in $\Tt$ and how often.
If we now want to argue that for \emph{all} $\Semi$-interpretations of universe $[n]$ a strategy $\Tt$ yields a valuation which is at most the valuation of a strategy $\Tt^*$, we can compare the monomials $m \coloneqq \pi_n \llb \Tt \rrb$ and $m^* \coloneqq \pi_n \llb \Tt^* \rrb$ and prove that $m \preceq m^*$ (recall that this means that each variable's exponent in $m^*$ is always at most its exponent in $m$).
This is because every $\Semi$-interpretation $\pi$ over universe $[n]$ corresponds to a model-defining variable assignment $f \colon X_n \to \Semi$ and hence a homomorphism $h\colon \Sorb (X_n^+, X_n^-) \to \Semi$. By the fundamental property (\cref{fundamental-property}) we have $(h \circ \pi_n) \llb \Tt \rrb = h(m) \leq h(m^*) = (h \circ \pi_n) \llb \Tt^* \rrb$ as homomorphisms preserve the natural order. Because multiplication is strictly decreasing on $\Semi \setminus \{0,1\}$, we can argue analogously that $\pi \llb \Tt \rrb < \pi \llb \Tt^* \rrb$ for each $\Semi$-interpretation $\pi$ with $1\not\in\img(\pi)$ whenever $m \prec m^*$.

\subsection{Strategy Translation}

Recall that our goal is to show that $\forall$ must be redundant in every finite extension preserved sentence where $\forall$ only occurs non-trivially. To prove this claim, we suppose that $\forall$ was not redundant in a finite extension preserved sentence $\psi$. Starting from a non-existential strategy $\Tt$ optimal for an $\Semi$-interpretation $\pi$,
we aim to infer a contradiction by disproving finite extension preservation of $\psi$.
Let $\Tt$ be a strategy optimal for $\pi$ and $\psi$ and $\foraneq y \phi (\bar x, y)$ be a subformula used in $\Tt$.
Now the idea is to remove some element $a$ from the universe of $\pi$ and construct a strategy $\Tt^*$ for the resulting subinterpretation $\pi^*$ by removing the substrategies from $\Tt$ which instantiate the universally quantified variable $y$ with $a$. Based on the fact that multiplication is strictly decreasing on on $\Semi \setminus \{0,1\}$, we want to argue that this strictly increases the valuation of the strategy, i.e. that $\pi^* \llb \psi \rrb \geq \pi^* \llb \Tt^* \rrb > \pi \llb \Tt \rrb = \pi \llb \psi \rrb$, which would disprove extension preservation as desired.
There are two intricacies to this. Firstly, we have to make sure that we start with an $\Semi$-interpretation $\pi$ that neither evaluates $\psi$ to $0$, nor evaluate the substrategies we remove during the strategy translation to~$1$ in order to actually increase the strategy's valuation. This is what \cref{subsec:One-valuations} will be devoted to. Secondly, $a$ may also occur as a \emph{witness}, that is, instantiation of an existential quantifier in $\Tt$.
In order to obtain a valid strategy for $\pi^*$, we have to relabel such occurrences of $a$ while still making sure that we only increase the valuation of the resulting strategy. 
We argue in the \emph{strategy translation lemma} that such relabelling is possible whenever $a$ as well as sufficiently many elements $b_1, \dots, b_r$ do not occur inside the literals in the leaves of $\Tt$. This is the key new technique we develop for this proof.

\begin{lemma}[Strategy Translation] \label{lem-relationship-between-polynomials}
    Let $\Semi$ be absorptive and $\psi \in \FO^{\neq}$ with $r \coloneqq \qr(\psi)$. For all $n > 2^{|\psi|+1}+ r$ and all strategies $\Tt \in C_{n+r+1} (\psi)$ such that $\supp(\pi_{n+r+1} \llb \Tt \rrb) \subseteq X_n$, there is some $\Tt^* \in C_{n+r} (\psi)$ such that 
     $\pi_{n+r+1} \llb \Tt \rrb \leq \pi_{n+r} \llb \Tt^* \rrb$.

     Further, for every $v \in V$ such that $\lambda (v) = \foraneq y \phi (\bar a, y)$ for some $\phi (\bar a, y) \in \FO_{[n]}^{\neq}$, there is some $w \in vE$ such that 
     $\pi_{n+r+1} \llb \Tt \rrb \leq \pi_{n+r} \llb \Tt^* \rrb \cdot \pi_{n+r+1} \llb \Tt (w) \rrb$ for the substrategy $\Tt (w)$ of $\Tt$ that is rooted at $w$.
\end{lemma}

Due to its technical complexity, we delay the proof of the strategy translation lemma and first illustrate how we apply it, formalising the idea described above.

\begin{lemma} \label{lem:counterexample-if-irredundant}
    Let $\Semi$ be absorptive such that $s \cdot t < s$ for all $s,t \in \Semi\setminus \{0,1\}$.
    If there is an $\Semi$-interpretation of universe $[k]$ where $k \geq 2\cdot(2^{|\psi|}+\qr(\psi)+1)$ and a strategy $\Tt=(V,E,\lambda)$ optimal for $\pi$ and $\psi$ such that
    \begin{enumerate}[(1)]
        \item $\pi \llb \psi \rrb = \pi \llb \Tt \rrb \neq 0$,
        \item there is some $v \in V$ such that $\lambda (v)=\foraneq y \phi (\bar a, y)$ for some $\foraneq y \phi (\bar a, y) \in \FO_{[k]}^{\neq}$ such that $\pi \llb \Tt(w) \rrb \neq 1$ for all $w \in vE$, and
        \item $\qr(\psi)+1$ elements do not occur in literals in leaves of $\Tt$,
    \end{enumerate}
    then $\psi$ cannot be finite extension preserved. 
    More precisely, $\pi$ has a subinterpretation violating extension preservation that contains a single element less that $\pi$.
\end{lemma}

\begin{proof}
    Suppose there was such an $\Semi$-interpretation $\pi$ and a strategy $\Tt$. Let $r \coloneqq \qr(\psi)$ and $n>2^{|\psi|+1}+r$ be such that $k=n+r+1$. Assume without loss of generality that the elements $n+1, \dots, n+r+1$ do not occur in the literals used in $\Tt$. Hence, $\Tt \in C_{n+r+1} (\psi)$ and $\supp (\pi_{n+r+1} \llb \Tt \rrb) \subseteq X_n$, so we can apply the Strategy Translation Lemma to $\Tt$, which yields a strategy $\Tt^* \in C_{n+r} (\psi)$. Now let $\pi^* \subseteq \pi$ be the subinterpretation of $\pi$ induced by $[n+r]$. We claim that $\pi^* \llb \psi \rrb \geq \pi^* \llb \Tt^* \rrb > \pi \llb \Tt \rrb = \pi \llb \psi \rrb$.
    Let $h \colon \Sorb(X_{n+r+1}^+, X_{n+r+1}^-) \to \Semi$ be the homomorphism induced by $\pi$ and 
    $v \in V$ such that $\lambda (v) = \foraneq y \phi (\bar a, y)$ and $\pi \llb \Tt(w) \rrb \neq 1$ for all $w \in vE$, which exists by assumption (3). 
    By the Strategy Translation Lemma, there must be some $w \in vE$ such that $\pi_{n+r+1} \llb \Tt \rrb \leq \pi_{n+r} \llb \Tt^* \rrb \cdot \pi_{n+r} \llb \Tt (w) \rrb$, 
    and since homomorphisms are order-preserving, we get the desired inequality:
    \begin{align*}
        \underbrace{\pi \llb \Tt \rrb}_{\neq 0} = h (\pi_{n+r+1} \llb \Tt \rrb ) &\leq h(\pi_{n+r} \llb \Tt^* \rrb ) \cdot h( \pi_{n+r+1} \llb \Tt (w) \rrb) \\
        &= \pi^* \llb \Tt^* \rrb \cdot \underbrace{\pi \llb \Tt(w) \rrb}_{\neq 1} < \pi^* \llb \Tt^* \rrb. \qedhere
    \end{align*}
\end{proof}

\begin{proof} (of the Strategy Translation Lemma)
        Let $\Tt=(V,E, \lambda) \in C_{n+r+1} (\psi)$ be such that $\supp(\pi_{n+r+1} \llb \Tt \rrb) \subseteq X_n$.
        Although $n+r+1$ does not occur inside any literal used in $\Tt$, it might still occur as a witness, i.e. there may be $(v,w) \in E$ such that $\lambda (v) = \exneq y \phi(\bar a, y)$ and $\lambda (w) =\phi(\bar a, n+r+1)$ for some $\exneq y \phi (\bar a, y) \in \FO_{[n+r+1]}^{\neq}$. So in order to construct a strategy $\Tt^*\in C_{n+r}(\psi)$ it does not suffice to remove substrategies $\Tt(w)$ where $(v,w) \in E$ such that $\lambda (v) = \foraneq y \phi(\bar a, y)$ and $\lambda (w) =\phi(\bar a, n+r+1)$, but all occurrences of $n+r+1$ as a witness have to be relabelled. Recall that in the definition of strategies we require the instantiations that are chosen in a play (i.e. along a path) to be pairwise distinct in order to reflect the semantics of the quantifiers $\exneq, \foraneq$. Thus, we have to find a relabelling that preserves this property.
        Whenever we replace $n+r+1$ with some element $n+i$, we can be sure to not change the valuation of the strategy because, by assumption, $n+i$ does not occur inside any literal used in the strategy either.
        Due to the potential nesting of universal and existential quantifiers, it might be that every such element $n+i$ is a witness in $\Tt$. In this case we are not be able to find a global substitute $n+i$ to replace $n+r+1$ with (this is only possible if $\Tt$ is existential).
		However, if $n+r+1$ occurs as a witness in distinct substrategies of $\Tt$, we do not have to replace it by the same element in order to obtain a strategy $\Tt^* \in C_{n+r} (\psi)$;
        we only need to make sure that (in)equalities along paths are respected.
		Hence, we inductively define a relabelling function $g_v \colon [n+r+1] \to [n+r+1]$ for each vertex $v \in \Tt$ (from the root to the leaves). The idea is as follows: 
        Whenever $n+r+1$ occurs as a witness of a formula $\exneq x \phi(\bar a, n+r+1)$, we replace it with some element $n+i$ which does not occur in $\bar a$ in the remaining substrategy. Because $|\bar a|$ is bounded by $r$, such $n+i$ must always exist.
        It may happen that $n+i$ has already occurred in this substrategy, so we further relabel the nodes in this subtree by proceeding analogously for $n+i$ rather than $n+r+1$.
        
        To define the relabelling functions formally, let $g[i \mapsto j]$ be the function that agrees with $g$ up to $i$, which is mapped to $j$.
		We set $g_{r(\Tt)} = \operatorname{id}_{[n+r+1]}$, where $r(\Tt)$ is the root of $\Tt$, and suppose that $g_v$ is already defined. Now if
		\begin{itemize}
			\item $\lambda (v) = \exneq y \phi (\bar a, y)$, we let $i \coloneqq \min \{g_v(x) \mid g_v(x)\neq x\}\cup \{n+r+1\}$ ($i$ is the element we currently want to eliminate). For the unique $w \in vE$, we set $g_w \coloneqq g_v$ if $\lambda (w) \neq \phi(\bar a, i)$ and $g_w \coloneqq g_v[i \mapsto j]$ where $j \coloneqq \max \{k \in [n+r] \mid k \not\in \bar a \cup \{g_v(x) \mid g_v(x)\neq x\}\}$ otherwise.
			\item in all other cases, we set $g_w \coloneqq g_v$ for all $w \in vE$.
		\end{itemize}
		Using the relabelling functions $g_v$, we construct $\Tt^*=(V^*, E^*, \lambda^*)$ from $\Tt$ by
		\begin{itemize}
			\item removing the subtree rooted at a node labelled $\phi (\bar a, i)$, where $i \coloneqq \min \{g_v(x) \mid g_v(x)\neq x\} \cup \{n+r+1\}$, if its predecessor is labelled $\foraneq y \phi (\bar a, y)$ and
			\item relabelling each node $v$ of $\Tt$ by $g_v(\lambda (v))$.
		\end{itemize}
		To verify that $\Tt^*$ is a strategy for $\psi$ and $n+r$, we only need to take care of the instantiations (because everything else is just copied from $\Tt$). 
		It follows inductively that for each node $v$ of $\Tt$ where $\lambda (v) = \phi (a_1, \dots, a_k)$ the mapping $g_v$ has the following properties:
		\begin{enumerate}[(1)]
			\item $g_v(n+r+1) \in [n+r]$ if $n+r+1$ occurs as a witness along the path from $r(\Tt)$ to $v$, and $g_v = \operatorname{id}_{[n+r+1]}$ otherwise,
			\item $\min \{g_v(x) \mid g_v(x) \neq x\}  > n$,
			\item $g_v$ is injective on $[n+r+1] \setminus \{i\}$ where $i \coloneqq \min \{g_v(x) \mid g_v(x)\neq x\} \cup \{n+r+1\}$ and
			\item $g_v$ is injective on $\{a_1, \dots, a_k\}$.
		\end{enumerate}
		
		Due to (1), ${n+r+1}$ cannot occur in $\lambda^* (v)$ for any node $v$ of $\Tt^*$: Either it is not contained in $\lambda (v)$ and this follows from $g_v = \operatorname{id}_{[n+r+1]}$, 
        or it must occur as an instantiation of an existentially quantified variable in $\lambda(v)$ (otherwise $v$ would not be contained in $\Tt^*$), but then it is relabelled by $g_v (n+r+1) \neq n+r+1$ in $\lambda^*(v)$.
		Further, the successors of a node $v$ with $\lambda^*(v) = \foraneq y \phi (\bar a, y)$ must be labelled $\phi (\bar a, 1), \dots, \phi (\bar a, {n+r})$ because $v$ has $n+r$ successors in $\Tt^*$ and, by (3), $\lambda (w_1) \neq \lambda (w_2)$ holds for $w_1 \neq w_2 \in vE$.
		Together with (4), this implies that $\Tt^*$ is a strategy for $\psi$ in $n+r$.
		
		To prove $\pi_{n+r+1} \llb \Tt \rrb \leq \pi_{n+r} \llb \Tt^* \rrb$, we show that $\pi_{n+r} ( \lambda^* (v) ) = \pi_{n+r+1}(\lambda (v))$ for each leaf $v$ of $\Tt^*$. 
		Since $\supp(\Tt) \subseteq X_n$, $\lambda (v)$ is a literal in $[n]$. Now (2) ensures that $\pi_{n+r} ( \lambda^* (v) ) = \pi_{n+r} ( g_v(\lambda (v)) )= \pi_{n+r+1}(\lambda (v))$.

        Let $v \in V$ be such that $\lambda (v) = \foraneq y \phi (\bar a, y)$ for some $\phi(\bar a, y) \in \FO_{[n+r+1]}^{\neq}$.
		When constructing $\Tt^*$ from $\Tt$, the subtree $\Tt(w)$ is removed from $\Tt$, where $\lambda (v) = \phi (\bar a, i)$ and $i \coloneqq \min \{g_v(x) \mid g_v(x)\neq x\} \cup \{n+r+1\}$. Hence, $\pi_{n+r+1} \llb \Tt \rrb \leq \pi_{n+r} \llb \Tt^* \rrb \cdot \pi_{n+r+1} \llb \Tt (w) \rrb$.
		\qedhere
\end{proof}

So in order to prove that finite extension preservation implies redundancy of $\forall$ via contradiction, we assume that redundancy of $\forall$ is violated in some finite extension preserved $\psi$, and show that this implies the existence of an $\Semi$-interpretation that meets the conditions from \cref{lem:counterexample-if-irredundant}. Just from the failure redundancy of $\forall$ in $\psi$, we only know the existence of arbitrarily large counterexamples to redundancy of $\forall$ in $\psi$. It remains to prove that we can find a sufficiently large such counterexamples to redundancy that additionally satisfy condition (1)-(3) in \cref{lem:counterexample-if-irredundant}. The first step will be to justify in \cref{lem:r-el-dont-occur-in-counterex-redundancy} that condition (1) and (3) can be met due to finite extension preservation and the fact that there is not a unique minimal positive element in $\Semi \in \{\Vit, \Trop, \Lukas, \Doubt\}$. \cref{subsec:One-valuations} will be devoted to condition (2).

\begin{lemma} \label{lem:r-el-dont-occur-in-counterex-redundancy}
    Let $\psi \not\equiv^{\geq m} \bot$ for each $m \in \Nat$ be finite extension preserved in $\Semi$. If $\forall$ is not redundant in $\psi$, then there are arbitrarily large counterexamples $\pi$ to redundancy of $\forall$ in $\psi$ such that $\pi \llb \psi \rrb >0$ at least $\qr(\psi)+1$ elements do not occur inside the literals in the leaves of any optimal strategy.
\end{lemma}

\begin{proof}
    First note that $\psi \not\equiv^{\geq m}_\Semi \bot$ for each $m \in \Nat$, together with the fact that $\psi$ is finite extension preservation in $\Semi$, implies that there must be some $m_0$ such that $\psi \not\equiv^{m}_\Semi \bot$ for all $m \geq m_0$.
    Fix some $n_0 \in \Nat$. Because $\forall$ is not redundant in $\psi$, there must be a counterexample $\pi$ to redundancy of $\forall$ of universe $[n]$ where $n > (n_0,m_0, 2^{|\psi|+1}+\qr(\psi))$. 
    We first justify that we can assume without loss of generality that $\pi \llb \psi \rrb > 0$: Suppose that $\pi \llb \psi \rrb =0$. Then we must have $\pi \llb \Tt \rrb =0$ for all strategies $\Tt \in C_n (\psi)$, but this means that every such $\Tt$ is optimal for $\pi$ and can thus not be existential. In this case, every $\Semi$-interpretation of universe $[n]$ is a counterexample to redundancy in $\forall$, and we can replace $\pi$ with some arbitrary $\Semi$-interpretation which does not evaluate $\psi$ to $0$. Such an $\Semi$-interpretation must exists by $\psi \not\equiv^{n}_\Semi \bot$ (this is why we insisted on $n\geq m_0$).
    
    Hence, we can fix some $0< v< \pi \llb \psi \rrb$ and define an extension $\pi_{\text{ext}} \supseteq \pi$ of universe $[n+r+1]$ such that $\pi_{\text{ext}} (\alpha) = v$ and $\pi_{\text{ext}} (\lnot \alpha) = 0$ for $\alpha \in \text{Atoms}_{[n+r+1]\setminus [n]}(\tau)$. 
    We claim that $\pi_{\text{ext}}$ has the required properties. By finite extension preservation in $\Semi$ and the fact that $\pi \llb \psi \rrb > 0$, we have $\pi_{\text{ext}} \llb \psi \rrb > 0$.
    Let $\Tt_{\text{ext}}$ be a strategy optimal for $\pi_{\text{ext}}$ and $\psi$. $\Tt_{\text{ext}}$ cannot use any literal $L \in \Lit_{[n+r+1]\setminus [n]} (\tau)$, otherwise extension preservation would be violated due to $\pi \llb \psi \rrb > v \geq \pi_{\text{ext}} \llb \Tt_{\text{ext}} \rrb = \pi_{\text{ext}} \llb \psi \rrb$. It remains to argue that $\Tt_{\text{ext}}$ cannot be existential.
    So suppose it was. We translate $\Tt_{\text{ext}}$ into an existential strategy $\Tt$ with $\pi \llb \Tt \rrb = \pi_{\text{ext}} \llb \Tt_{\text{ext}} \rrb$ to infer a contradiction. Note that although the element $n+1, \dots, n+r+1$ are not used in a relational atom in $\Tt_{\text{ext}}$ they may still occur as the instantiation of an existentially quantified variable, so have to be replaced by elements from $[n]$.
    The depth of $\Tt_{\text{ext}}$ is $|\psi|$ and its branching degree is at most $2$. Hence, $\Tt$ has less than $2^{|\psi|+1} < n-r$ nodes, and there must be elements $i_1, \dots, i_r \in [n]$ which does not occur at all in $\Tt_{\text{ext}}$. Hence, swapping each occurrence of $n+j$ in $\Tt_{\text{ext}}$ with $i_j$ respects all equalities atoms that may occur in $\Tt_{\text{ext}}$. Further, neither $i_j$ nor $n+j$ occurs inside a relational atom in $\Tt_{\text{ext}}$, which ensures $\pi \llb \Tt \rrb = \pi_{\text{ext}} \llb \Tt_{\text{ext}} \rrb$. But, by extension preservation, $\pi \llb \Tt \rrb = \pi_{\text{ext}} \llb \Tt_{\text{ext}} \rrb = \pi_{\text{ext}} \llb \psi \rrb \geq \pi \llb \psi \rrb$, so $\Tt$ would be an existential strategy optimal for $\pi$ and $\psi$, a contradiction.
\end{proof}

Note that although we do not necessarily need a counterexample to redundancy of $\forall$ in $\psi$ to apply \cref{lem:counterexample-if-irredundant}, but only an $\Semi$-interpretation for which \emph{some} strategy is optimal for $\psi$, the existence of sufficiently large counterexamples to redundancy of $\forall$ was necessary to prove \cref{lem:r-el-dont-occur-in-counterex-redundancy}, so to argue that we can find some optimal strategy which does not use $\qr(\psi)+1$ many elements in its literals. Otherwise, we had not been able to exclude the case that every strategy optimal for the extension $\pi_{\text{ext}}$ we construct in \cref{lem:r-el-dont-occur-in-counterex-redundancy} is existential.

\subsection{Excluding $1$-valuations} \label{subsec:One-valuations}

The final argument that is still missing to invoke \cref{lem:counterexample-if-irredundant} and finally infer that finite extension preservation implies redundancy of $\forall$ is that we can make sure to not only remove substrategies that are evaluated to~$1$ during the strategy translation. As mentioned previously, we can assume without loss of generality that every subformula $\foraneq y \phi (\bar x, y)$ in $\psi$ is non-trivial. However, recall that this only means that $\foraneq y \phi (\bar x, y)$ \emph{can} take a valuation different from $1$ on arbitrarily large $\Semi$-interpretations. A priori, it is neither clear whether this can happen for \emph{all} universe sizes, in particular those for which there is a counterexample to redundancy, nor does this immediately ensure that \emph{no} $\Semi$-interpretation of a particular size evaluates $\foraneq y \phi (\bar a, y)$ itself or a strategy for $\foraneq y \phi (\bar a, y)$ for \emph{any} instantiation $\bar a$ to~$1$. We justify in \cref{lem:non-trivial-eventually-never-1} that for $\Semi$-interpretations $\pi$ that have at least a certain cardinality, strategies for non-trivial formulae can actually only be evaluated to~$1$ if $1$ occurs as the valuation of a literal in $\pi$. In \cref{lem:one-does-not-occur-in-pi}, we then argue that we can translate every $\Semi$-interpretation into another one over the same universe, in which $1$-valuations no longer occur, while keeping at least one optimal strategy.

In order to prove \cref{lem:non-trivial-eventually-never-1}, we reuse the strategy translation lemma and apply the following auxiliary lemma, which allows us to express triviality in terms of the polynomial interpretations $\pi_n$ and to prove that triviality never depends on a particular instantiation of the free variables.

\begin{lemma} \label{lem:onlyEqualityTypeMatters}
    Let $\Semi \not\cong \Bool$ be absorptive. For each $n \in \Nat$, and every formula $\phi (\bar x) \in \FO^{\neq}$, the following are equivalent:
    \begin{enumerate}[(1)]
        \item $\phi (\bar x)$ is trivial for $n$,
        \item $\pi_n \llb \phi (\bar a) \rrb = 1$ for all $\bar a \subseteq [n]$, and
        \item $\pi_n \llb \phi (\bar a) \rrb = 1$ for some $\bar a \subseteq [n]$.
    \end{enumerate}
\end{lemma}

\begin{proof}
    $(2) \Rightarrow (1)$: Let $\pi$ be a $\Semi$-interpretation of size (assume w.l.o.g. that its universe is $[n]$), $\bar a \subseteq [n]$, and $h \colon \Sorb (X_n^+,X_n^-) \to \Semi$ be the model-defining homomorphism it induces. Then we have $\pi \llb \phi (\bar a ) \rrb = (h \circ \pi_n) \llb \phi (\bar a ) \rrb = h(\pi_n \llb \phi (\bar a ) \rrb) = h(1) = 1$.

    $(1) \Rightarrow (2)$: Suppose that $\pi_n \llb \phi (\bar a) \rrb \neq 1$ for some $\bar a \subseteq [n]$, and fix some $s \in \Semi$ such that $0<s<1$. 
    Let $\pi$ be the $\Semi$-interpretation of universe $[n]$ where we annotate each positive literal with $s$ and each negative literal with $0$. Because multiplication is decreasing, we have $\pi \llb \phi (\bar a) \rrb \in \{0\} \cup \{s^n \mid n \geq 1\}$. In any case, $\pi \llb \phi (\bar a) \rrb < 1$.

    $(3) \Rightarrow (2)$: Suppose that $\pi_n \llb \phi (\bar a) \rrb = 1$ and let $\bar b \subseteq [n]$. Let $k$ be the arity of $\bar a$ and $\bar b$. Fix some permutation $\sigma \colon [n] \to [n]$ such that $\sigma \colon b_i \mapsto a_i$ for each $i \in [k]$. This is possible, because the entries of $\bar a$ and $\bar b$, respectively, are pairwise distinct, so $\bar b \mapsto \bar a$ is an injective mapping.
	Now consider the homomorphism $h_\sigma \colon \Sorb (X_n^+, X_n^-) \to \Sorb (X_n^+, X_n^-)$ induced by $x_L \mapsto x_{\sigma (L)}$ where for each $L \in \Lit_{[n]} (\tau)$, the literal $\sigma(L)\in \Lit_{[n]} (\tau)$ arises from $L$ by substituting each instantiation $a \in [n]$ with $\sigma(a)$. It holds that $\sigma \colon \pi_n \cong h_\sigma \circ \pi_n$ and thus
	\[
	\pi_n \llb \phi (\bar b) \rrb = (h_\sigma \circ \pi_n) \llb \phi (\bar a) \rrb = h_\sigma (\pi_n \llb \phi (\bar a) \rrb = h_\sigma (1) = 1. \qedhere
	\] 
\end{proof}

\begin{lemma} \label{lem:non-trivial-eventually-never-1}
    Let $\Semi$ be absorptive and $\phi (\bar x) \in \FO^{\neq}$.
    \begin{enumerate}[(1)]
        \item If $\phi (\bar x)$ is non-trivial, then there is some $n_0$ such that $\phi (\bar x)$ is non-trivial for each $n_0 \geq n$.
        \item If $\phi (\bar x)$ is non-trivial for $n$, then $\pi \llb \Tt \rrb \neq 1$ for all $\Semi$-interpretations $\pi$ of universe $[n]$ with $1 \not\in\img(\pi)$, all $\bar a \subseteq [n]$ and strategies $\Tt$ for $\phi (\bar a)$.
    \end{enumerate}
\end{lemma}

\begin{proof}
    \begin{enumerate}[(1)]
    \item Let $\phi (\bar x)$ be non-trivial. This means that for every $n_0 \in \Nat$ there is some $n\geq n_0$ such that $\phi (\bar x)$ is non-trivial for $n$. We claim that there must be some $n_0 \in \Nat$ such that $\phi (\bar x)$ is non-trivial for all $n \geq n_0$. Suppose that this was not true. Then there must be infinitely many $n \in \Nat$ such that $\phi (\bar x)$ is trivial for $n$, but not for $n-1$.
    By \cref{lem:onlyEqualityTypeMatters}, this implies $\pi_n \llb \phi (\bar a) \rrb = 1$ for some $\bar a \subseteq [n]$, and $\pi_{n-1} \llb \phi (\bar b) \rrb \neq 1$ for all $\bar b \subseteq [n-1]$.
	Hence, for $\bar x =(x_1, \dots, x_k)$ and $\theta \coloneqq \exneq x_1 \dots \exneq x_k \phi (\bar x)$ we have $\pi_n \llb \theta \rrb =1$ and $\pi_{n-1} \llb \theta \rrb \neq 1$. Now apply \cref{lem-relationship-between-polynomials} (note that we can choose $n$ arbitrarily large, so that it exceeds the lower bound for $n$ required by \cref{lem-relationship-between-polynomials}) to~$\theta$ and infer a contradiction. By the Sum-of-Strategies Theorem, $\pi_n \llb \theta \rrb =1$ implies that there must be a strategy $\Tt \in C_n (\theta)$ such that $\pi_n \llb \Tt \rrb = 1$. But then there must also be some $\Tt^* \in C_{n-1} (\theta)$ such that $\pi_{n-1} \llb \Tt^* \rrb = 1$, i.e. $\pi_{n-1} \llb \theta \rrb = 1$, a contradiction.
    \item Let $\phi (\bar x)$ be non-trivial for $n$. By \cref{lem:onlyEqualityTypeMatters}, we have $\pi_n \llb \phi (\bar a ) \rrb \neq 1$ for all $\bar a \subseteq [n]$. For any model-defining homomorphism $h_\alpha \colon \Sorb(X_n^+,X_n^-) \to \Semi$ induced by a variable assignment $\alpha \colon X_n \to \Semi \setminus \{1\}$, we have $h_\alpha^{-1}(1)= \{1\}$. Hence, for every $\Semi$-interpretation $\pi$ of universe $[n]$ and $\bar a \subseteq [n]$, we must have $\pi \llb \phi (\bar a) \rrb \neq 1$. By the Sum-of-Strategies Theorem, this implies $\pi \llb \Tt \rrb \neq 1$ for all strategies $\Tt$ for $\phi (\bar a)$ because we have $1+s=1$ for all $s \in \Semi$ due to absorption. \qedhere
    \end{enumerate}
\end{proof}

\begin{lemma} \label{lem:one-does-not-occur-in-pi}
    Let $\Semi \in \{\Vit, \Trop, \Lukas, \Doubt\}$. For every $\Semi$-interpretation $\pi$ 
    with $\pi \llb \psi \rrb \neq 0$ 
    there is a $\Semi$-interpretation $\pi^*$ such that 
    $\pi^* \llb \psi \rrb \neq 0$, 
    $1 \not\in\img (\pi^*)$ and $\preceq_{\pi^*} \subseteq \preceq_{\pi}$. In particular, every strategy optimal for $\pi^*$ and $\psi$ must be optimal for $\pi$ and $\psi$ too.
\end{lemma}

\begin{proof}
    Let $\pi$ be an $\Semi$-interpretation, say of universe $[n]$.
    Extend $\preceq_{\pi}$ to a preorder $\preceq_{\pi}^*$ on $C_n (\psi) \cup \{\Tt_\bot\}$, where $\Tt_\bot$ is a dummy ``strategy'', which is always evaluated to $0$.
    Because $C_n (\psi) \cup \{\Tt_\bot\}$ is finite, there exist $\delta, e$ such that $\delta \coloneqq \min \{\pi \llb \Tt \rrb - \pi \llb \Tt' \rrb \mid \Tt' \prec_\pi^* \Tt \}$ and $e \coloneqq \max \{\deg(\pi_n \llb \Tt \rrb) \mid \Tt \in C_n (\psi) \cup \{\Tt_\bot\} \}$.

    \begin{itemize}
        \item First consider the case $\Semi= \Vit$. Fix some $s$ such that $ \sqrt[e]{1-\frac{\delta}{\pi \llb \psi \rrb}} < s < 1$ and set $\pi^* (L) = \pi (L)$ if $\pi(L) \neq 1$ and otherwise $\pi^* (L) = s$. Now let $\Tt \succ_{\pi} \Tt'$, i.e. $\pi \llb \Tt \rrb - \pi \llb \Tt' \rrb \geq \delta$. We claim that $\Tt \succ_{\pi^*} \Tt'$. Using $\pi \llb \Tt \rrb \leq \pi \llb \psi \rrb$ $(*)$ by the Sum-Of-Strategy Theorem, we obtain
        \begin{align*}
	       \pi^* \llb \Tt \rrb - \pi^* \llb \Tt' \rrb &\geq \pi^* \llb \Tt \rrb - \pi \llb \Tt' \rrb \\
            &\geq s^e \cdot \pi \llb \Tt \rrb - \pi \llb \Tt' \rrb \\ 
           &> (1-\delta/\pi \llb \psi \rrb) \cdot \pi \llb \Tt \rrb - \pi \llb \Tt' \rrb \\
           &\geq \pi \llb \Tt \rrb - (\delta\cdot\pi \llb \Tt \rrb)/\pi\llb \psi \rrb - \pi \llb \Tt' \rrb \\
           &\overset{(*)}\geq \pi \llb \Tt \rrb - \delta - \pi \llb \Tt' \rrb
           \geq 0.
	   \end{align*}
       Hence, $\Tt \succ_{\pi^*} \Tt'$ and we overall obtain $\preceq_{\pi^*} \subseteq \preceq_{\pi}$.
       \item Now let $\Semi=\Doubt$. Due to $+^\Doubt=\min$, the natural order of $\Doubt$ is inverted, and $0^\Doubt=1$ and $1^\Doubt=0$. We denote by $\leq_{\mathbb{R}}$ the usual order on $[0,1]_\mathbb{R}$. Now fix some $1^\Doubt = 0<_{\mathbb{R}} s<_{\mathbb{R}}\delta/ e$ and, as before, set $\pi^* (L) = \pi (L)$ if $\pi(L) \neq 0 = 1^\Doubt$ and $\pi^* (L) = s$ otherwise. Let $\Tt \succ_{\pi} \Tt'$, i.e. $\pi \llb \Tt \rrb <_{\mathbb{R}} \pi \llb \Tt' \rrb$ and thus  $\pi \llb \Tt' \rrb - \pi \llb \Tt' \rrb \geq_{\mathbb{R}} \delta$. We have
        \begin{align*}
	       \pi^* \llb \Tt' \rrb - \pi^* \llb \Tt \rrb &\geq_{\mathbb{R}} \pi \llb \Tt' \rrb - \pi^* \llb \Tt \rrb \\
            &\overset{(*)}\geq_{\mathbb{R}} \pi \llb \Tt' \rrb - (e \cdot s + \pi \llb \Tt \rrb) \\
            &> \pi \llb \Tt' \rrb - \delta - \pi \llb \Tt \rrb \geq 0
	   \end{align*}
       Note in $(*)$ that $e \cdot s + \pi \llb \Tt \rrb$ cannot exceed $1$ because we take into account the dummy strategy $\Tt_\bot$, which always evaluates to $1=0^\Doubt$, when defining $\delta$. As before, it follows that $\Tt \succ_{\pi^*} \Tt'$, proving $\preceq_{\pi^*} \subseteq \preceq_{\pi}$.
    \end{itemize}
    Now note that we must have $\pi^* \llb \psi \rrb>0$. Otherwise we would have $\pi^* \llb \Tt \rrb = 0$, i.e. $\Tt \preceq_{\pi^*} \Tt_\bot$, for all $\Tt \in C_n (\psi)$ while $\pi \llb \Tt \rrb >0$, i.e. $\Tt \not\preceq_{\pi} \Tt_\bot$, for some $\Tt \in C_n(\psi)$ according to the Sum-Of-Strategies Theorem.
    Let $\Tt$ be a strategy optimal for $\pi^*$ and $\psi$. This means that $\Tt \succeq_{\pi^*} \Tt'$ for all $\Tt' \in C_n (\psi)$
    and thus $\Tt \succeq_{\pi} \Tt'$ for all $\Tt' \in C_n (\psi)$, implying that $\Tt$ must be optimal for $\pi$ and $\psi$ too.
\end{proof}

	\subsection{Existence of Existential Optimal Strategies and Rewriting of the Formula}

    Using the previous lemmas, we can finally put everything together and argue that non-existential strategies can be neglected in the presence of finite extension preservation.

\begin{theorem} \label{thm:only-redundant-subform}
	Let $\Semi \in \{\Vit, \Trop, \Lukas, \Doubt\}$ and $\psi \in \FO^{\neq}$ be such that (1) $\psi \equiv^{\geq m}_{\Semi} \vartheta$ for some finite extension preserved $\vartheta \in \FO^{\neq}$ and $m \in \Nat$, (2) $\psi \not\equiv^{\geq m}_\Semi \bot$ for all $m \in \bbN$ and (3) such that every subformula $\foraneq y \phi (\bar x, y)$ of $\psi$ is non-trivial. Then $\forall$ must be redundant in $\psi$.
\end{theorem}

\begin{proof}
    Suppose that $\forall$ was not redundant in $\psi$. Let $n_0 > 2\cdot(2^{|\psi|}+r+1)$ be such that $\psi \equiv^{\geq n_0-1}_{\Semi} \vartheta$ and such that every subformula $\foraneq y \phi (\bar x, y)$ of $\psi$ is non-trivial for all $n \geq n_0$. Such $n_0$ must exist by \cref{lem:non-trivial-eventually-never-1}(1).
    By \cref{lem:r-el-dont-occur-in-counterex-redundancy}, there must be a counterexample $\pi$ to redundancy of $\forall$ of universe $[n]$ where $n \geq n_0$ such that $\pi \llb \psi \rrb > 0$ and such that at least $\qr(\psi)+1$ elements do not occur inside the literals used in any optimal strategy. Transform $\pi$ into an $\Semi$-interpretation $\pi^*$ such that $1 \not\in \img(\pi^*)$ according to \cref{lem:one-does-not-occur-in-pi}. Fix some strategy $\Tt =(V,E, \lambda)$ optimal for $\pi^*$ and $\psi$. We claim that $\pi^*$ and $\Tt$ satisfy all conditions to apply \cref{lem:counterexample-if-irredundant}.
    Since $\pi \llb \psi \rrb > 0$, we also have $\pi^* \llb \psi \rrb > 0$.
    Every strategy optimal for $\pi^*$ and $\psi$ must also be optimal for $\pi$ and $\psi$ and thus non-existential such that at least $\qr(\psi)+1$ elements do not occur inside the literals of the strategy. So, in particular, this must hold for $\Tt$.
    It remains to verify condition (2) from \cref{lem:counterexample-if-irredundant}.
    Let $v \in V$ be such that $\lambda(v) = \foraneq y (\bar a, y)$ for some $\foraneq y \phi (\bar a, y) \in \FO_{[n]}^{\neq}$. We know that $\foraneq y \phi (\bar x, y)$ is non-trivial for $n$ (as we have chosen $n$ to be large enough). By the Sum-Of-Strategies Theorem, $\phi (\bar x, y)$ must be non-trivial too. According to \cref{lem:non-trivial-eventually-never-1}(2), this implies that no strategy for $\forall y (\bar a, y)$, so, in particular, no substrategy $\Tt(w)$ of $\Tt$ rooted at $w \in vE$, can be evaluated to~$1$ by $\pi^*$. Hence, we can apply \cref{lem:counterexample-if-irredundant} and obtain a contradiction to finite extension preservation of $\psi$.
\end{proof}

As we have just shown that in extension-preserved sentences, non-trivial universal subformulae are redundant for all large enough $n$, we can safely replace them with $\bot$ without changing the semantics of the formula:

\begin{theorem} \label{lem:replace-redundant-subform}
    Let $\Semi \in \{\Vit, \Trop, \Lukas, \Doubt\}$ and $\vartheta \in \FO^{\neq}$ be finite extension preserved in $\Semi$ such that $\theta \not\equiv^{\geq n}_\Semi \bot$ for all $n \in \bbN$. Then $\vartheta \equiv^{\geq n}_\Semi \chi$ for some $n \in \Nat$ where $\chi$ arises from $\vartheta$ by substituting each trivial subformula $\foraneq y \phi (\bar x, y)$ by $\top$ and each non-trivial subformula $\foraneq y \phi (\bar x, y)$ by $\bot$.
\end{theorem}

\begin{proof}
    By definition, there must be some $m_0 \in \Nat$ such that $\vartheta \equiv_{\Semi}^{\geq m_0} \psi$, where $\psi$ arises from $\vartheta$ by replacing each trivial subformula $\foraneq y \phi (\bar x, y)$ by $\top$. Hence, we can apply \cref{thm:only-redundant-subform} to $\psi$ and infer that $\forall$ is redundant in $\psi$. Let $n_0 \in \Nat$ be such that for every $\Semi$-interpretation $\pi$ of size at least $n_0$ there exists an existential strategy optimal for $\pi$ and $\psi$. Now note that for any existential strategy for $\psi$ there is a strategy for $\chi$ with the same valuation and vice versa unless it uses $\bot$ and evaluates to~$0$. For any $\Semi$-interpretation of universe $[n]$ we have
    \begin{align*}
		\pi \llb \psi \rrb &= \max \{\pi \llb \Tt \rrb \mid \Tt \in C_m(\psi) \} \\
		&= \max \{\pi \llb \Tt \rrb  \mid \Tt \in C_m(\psi) \text{ existential}\} \\
		&= \max \{\pi \llb \Tt \rrb \mid \Tt \in C_m(\chi)\} = \pi \llb \chi \rrb.
	\end{align*}
    So overall we obtain $\vartheta \equiv_{\Semi}^{\geq \max(m_0,n_0)} \chi$.
\end{proof}

Combining \cref{lem:large-universes-suffice}, \cref{lem:FOneq}, and \cref{lem:replace-redundant-subform}, we get the desired result.

\extensionPresFinite*

\section{Proof of the Extension Preservation Theorem in the Finite for Lattice Semirings} \label{sec:FinExtPresLattices}

In this section, we prove that extension preservation theorem holds in the finite not just for the Viterbi and \L ukasiewicz semiring, but also in every non-Boolean lattice semiring. In particular, it holds for the $3$-element min-max semiring $\Semi_3$, which means that adding just a single additional truth value to the Boolean semiring makes finite extension preservation hold again, although it fails in the Boolean setting.
We proceed similarly as in \cref{sec:lattices}, and first prove the theorem for one specific lattice and then generalise it to arbitrary lattice semirings using the same reduction method. 
Unlike in \cref{sec:lattices}, we do not use the lattice $\Semi_3$, but consider the case where we evaluate in the \emph{fuzzy semiring} $\Fuzzy=([0,1], \max, \min, 0, 1)$. 
So our first goal is to prove the following.

\begin{theorem} \label{thm:extension-preservation-fuzzy}
	A sentence $\psi \in \FO^{\neq}$ is preserved under extensions of finite $\Fuzzy$-interpretations if, and only if, on finite $\Fuzzy$-interpretations, it is equivalent to an existential sentence.
\end{theorem}

Its linear order allows us to apply the proof framework based optimal strategies from \cref{sec:extension-preservation-finite-proof} to $\Fuzzy$.
Recall from \cref{section:finiteExtPresIdea} that unlike for $\Vit$ and $\Lukas$ sentences such as $\psi = \foraneq y \exneq z Rz$ may be extension preserved in $\Fuzzy$ although they do not admit any existential strategy.
Hence, instead of distinguishing purely existential and non-existential strategies as in \cref{sec:extension-preservation-finite-proof}, we will now consider
\begin{enumerate}[(1)]
	\item strategies where whenever a subformula $\foraneq y \phi (\bar a,y)$ is used, at least one substrategy for some $\phi (\bar a, b)$ \emph{does not} use any literal that contains $b$ from
	\item strategies that use a subformula $\foraneq y \phi (\bar a, y)$ where all substrategies rely on a literal that contains $y$.
\end{enumerate}

The crucial insight will be that, in case (1), we can assume that the optimal strategy uses the universally quantified variable only in a constant number of substrategies of $\foraneq y \phi (\bar a,y)$.
This is why we will refer to strategies from (1) as \emph{almost existential}. In the first part of the proof, we will show that extension preservation in $\Fuzzy$ ensures the existence of, not necessarily existential, but almost existential optimal strategies. In the second part, we will see that universal quantifiers can still be rewritten based on this result.

\subsection{Existence of Almost Existential Optimal Strategies}

We fix some notation first. For a labelled tree $\Tt = (V, E, \lambda)$ with $v \in T$, we refer to the subtree of $\Tt$ rooted at $v$ as $\Tt(v)$. 
For a strategy $\Tt \in C_{n}(\psi)$, let $A_\exists (\Tt)$ denote the set of \emph{witnesses} $i\in [n]$, which occur as an instantiation of an existential quantifier inside $\Tt$, i.e. for which there is some $v \in V$ such that $\lambda (v) = \exneq y \phi (\bar a, y)$ and some $w \in vE$ such that $\lambda (v) = \phi (\bar a, i)$ for some $\phi (\bar a, y) \in \FO^{\neq}_{[n]}$. By $A_{\Lit} (\Tt) \subseteq [n]$, we denote the set of elements that occur inside a literal in the leaves of~$\Tt$. Note that this might not cover the whole set of elements instantiated inside $\Tt$, as some instantiations might not actually appear inside a literal (recall $\foraneq y \exneq x Rx$, for instance).
Finally, for $\psi \in \FO^{\neq}$ let $\qr_\forall (\psi)$ denote the nesting depth of universal quantifiers in $\psi$.

\begin{Definition}
	We say that a strategy $\Tt=(V,E, \lambda) \in C_n (\psi)$ \emph{relies on $\forall$} if there is some $v \in V$ such that
	\begin{enumerate}[(1)]
		\item $\lambda(v) = \foraneq y \phi (\bar a, y)$ for some $\phi(\bar a, y) \in \FO^{\neq}_{[n]}$ and
		\item for each $w \in vE$ with $\lambda (w) = \phi (\bar a, i)$, it holds that $i \in A_{\Lit} (\Tt(w))$.
	\end{enumerate}
	Otherwise, we call $\Tt$ \emph{almost existential}.
\end{Definition}

The goal of this section is to prove that for extension preserved sentences and (sufficiently large) finite $\Fuzzy$-interpretation there will always be an almost existential optimal strategy. The structure of the proof will roughly be as follows. Suppose there is a strategy relying on $\forall$ that \enquote{beats} (i.e. has a strictly greater valuation than) all almost existential strategies on an $\Fuzzy$-interpretation $\pi$. Extend $\pi$ by sufficiently small valuations (this is why we use the dense lattice $\Fuzzy$) to make sure that all strategies optimal for the extension $\pi^*$ are almost existential. Translate one such strategy $\Tt^*$ into an almost existential strategy $\Tt$ for $\pi$ while making sure that its valuation can only be increased. This will either contradict extension preservation or the assumption that there is no almost existential optimal strategy for $\pi$.

The most technical part is the strategy translation. The strategy $\Tt^*$ may use an element $i$ as a witness which occurs only in the universe of $\pi^*$, but not in the universe of $\pi$. To obtain a valid strategy for $\pi$, a different element needs to be chosen. 
However, it might happen (unlike in the existential case) that some substrategies already use each $j$ contained in the universe of $\pi$ (which is a problem because substituting $i$ with $j$ then yields new equalities and might change the valuation of the strategy). So the idea is to first transform $\Tt^*$ into another strategy for $\pi^*$ that contains fewer witnesses. More precisely, the number of witnesses can be bounded by a constant and will no longer depend on the size of the universe. This allows us to find a suitable substitute for $i$ as long as the universe of $\pi^*$ is chosen large enough.

\begin{lemma} \label{lem:swap-instantiation}
	Let $\Tt \in C_n (\phi (\bar a, b))$ be a strategy. For each $c \not\in \bar a$, the tree $\Tt[b \leftrightarrow c]$ (where each occurrence of $b$ is replaced by $c$ and vice versa) is a strategy for $\phi (\bar a, c)$ and $n$.
\end{lemma}

\begin{lemma} \label{lem:translation-almost-existential-1}
	For every $\psi \in \FO^{\neq}$ there are constants $(c_m)_{m \in \omega}$ such that for any $m,n \in \omega$ and almost existential $\Tt \in C_{n} (\psi)$
	there is an almost existential $\Tt_m = (V_m, E_m, \lambda_m) \in C_{n} (\psi)$ such that  $\supp(\pi_{n} \llb \Tt_m \rrb) \subseteq \supp(\pi_{n} \llb \Tt \rrb)$ and for any $v \in \Tt_m$ with $\lambda_m (v)=\foraneq y \phi (\bar a, y)$ and $\qr_\forall (\lambda_m (v)) \leq m$,
	\begin{enumerate}[(1)]
		\item $|A_\exists (\Tt_m(w))| \leq c_m$ for each $w \in vE_m$;
		\item$|A_{\Lit} (\Tt_m (w))| \leq c_m$ for each $w \in vE_m$;
		\item at most $c_m$ many substrategies $\Tt_m(w_1), \dots, \Tt_m(w_{c_{m}})$ where $w_i \in vE_m$ contain a literal involving $y$. For all other $w,u \in vE_m \setminus \{w_1, \dots, w_{c_m}\}$, the substrategies $\Tt_m(w)$ and $\Tt_m(u)$ are isomorphic up to the instantiation of $y$.
	\end{enumerate}
\end{lemma}

\begin{proof}
	Inductively define the constants by $c_0=0$ and $c_{m+1} \coloneqq 2^{|\psi|+1} \cdot ((c_{m}+1) \cdot c_m+1)$.
	We prove the claim by induction on $m$.
	For $m=0$, we can choose $\Tt_0 \coloneqq \Tt$ and nothing remains to be shown.
	So suppose the claim is true for $m$. Let $\Tt \in C_n(\psi)$ be almost existential, and fix some $\Tt_m$ that satisfies the required properties with respect to $\Tt$.
	To define $\Tt_{m+1}$ from $\Tt_m= (V_m, T_m, \lambda_m)$, we successively modify the subtrees rooted at nodes $v \in \Tt_m$ such that $\lambda_m (v) = \foraneq y \phi (\bar a, y)$ where $\qr_\forall (\phi (\bar a, y))=m+1$.
	Fix one such node $v$ (the order in which we modify the subtrees will not matter).
	Let further $vE_m = \{w_{i_1} \dots w_{i_k}\}$, for $k \leq n$, be such that $\lambda_m (w_{i_\ell}) = \phi (\bar a, i_\ell)$ for all $\ell \in [k]$.
	Because $\Tt$ is almost existential, there must be some $\ell \leq k$ such that $i_\ell \not\in A_{\Lit}(\Tt_m(w_{i_\ell}))$.
	Now replace each subtree $\Tt_m(w_{i_j})$ where $i_j \not\in A_\exists (\Tt_m (w_{i_\ell})) \cup A_{\Lit}(\Tt_m(w_{i_\ell}))$ by $\Tt_m(w_{i_\ell})[i_j \leftrightarrow i_\ell]$. 
	Denote the resulting tree by $\Tt_m'= (V_m', E_m', \lambda_m')$. By \cref{lem:swap-instantiation}, $\Tt_m'$ is a strategy for $\psi$ and $n$.
	We have $\pi_{n} \llb \Tt_m (w_{i_\ell}) \rrb = \pi_n \llb \Tt_m (w_{i_\ell})[i_j \leftrightarrow i_\ell] \rrb$ because neither $i_\ell$ nor $i_j$ occurs inside a literal in the leaves of $\Tt_m (w_{i_\ell})$.
	This ensures $\supp (\pi_n \llb \Tt_m' \rrb) \subseteq \supp (\pi_n \llb \Tt_m \rrb) \subseteq \supp(\pi_{n} \llb \Tt \rrb)$, where the latter inclusion is by induction.
	It remains to verify (1)--(3). Fix some $w \in vE_m'$.
	\begin{enumerate}[(1)]
		\item We can upper-bound $|A_\exists (\Tt_m'(w))|$ by the number of inner nodes of $\Tt_m'(w)$. If we truncate the tree $\Tt_m'(w)$ at all universal nodes, we get a tree of depth at most $|\psi|$ which has branching degree at most $2$, i.e., which contains less than $2^{|\psi|}$ inner nodes. For each leaf $t$ of this truncated tree, we may have a substrategy $\Tt_m' (t)$ of a formula $\foraneq z \theta (\bar a, z)$. By induction hypothesis, there are only $c_m+1$ distinct subtrees adjacent to $t$, each of which contains at most $c_m$ elements as instantiations of existential quantifiers. Hence, $|A_\exists (\Tt_m'(w))| < 2^{|\psi|} + 2^{|\psi|} \cdot (c_m+1) \cdot c_m \leq c_{m+1}$.
		
		\item  Similarly, the number of leaves in $ \Tt'_m(w)$ provides an upper bound on $A_{\Lit} (\Tt'_m(w))$. The tree obtained from $\Tt_m'(w)$ by truncating at universal nodes has at most $2^{|\psi|}$ leaves. Using the same argument as before, we obtain $|A_{\Lit} (\Tt)| \leq 2^{|\psi|}\cdot (c_m+1) \cdot c_m \leq c_{m+1}$.
		
		\item By construction, only the subtrees $\Tt'_m (w_{i_j})$ for $i_j \in A_\exists (\Tt_m(w_{i_\ell})) \cup A_{\Lit} (\Tt_m(w_{i_\ell}))$ may contain a literal involving $i_j$. Hence, the claim follows from the fact that $|A_\exists (\Tt_m(w_{i_\ell})) \cup A_{\Lit} (\Tt_m(w_{i_\ell}))| < 2^{|\psi|} \cdot ((c_m+1) \cdot c_m+1) + 2^{|\psi|}\cdot (c_m+1) \cdot c_m \leq c_{m+1}$.
	\end{enumerate}
	
	Proceed analogously with all other $v \in \Tt_m$ with $\lambda_m (v) = \foraneq y \phi (\bar a, y)$ and $\qr_\forall (\phi (\bar a, y))=m+1$ to finally obtain $\Tt_{m+1}$.
\end{proof}

\begin{lemma}[Translation of almost existential strategies] \label{lem:translation-almost-existential}
	Let $\psi \in \FO^{\neq}$ and $r \coloneqq \qr(\psi)$.
	There is some $n_0 \in \bbN$ such that for all $n \geq n_0$ and all almost existential strategy $\Tt \in C_{n+r} (\psi)$ with $\supp (\pi_{n+r} \llb \Tt \rrb)\subseteq X_n$, there is some almost existential $\Tt^{*} \in C_{n} (\psi)$ such that $\supp(\pi_n \llb \Tt^{*} \rrb) \subseteq \supp(\pi_{n+r} \llb \Tt \rrb)$.
\end{lemma}

\begin{proof}
    Apply \cref{lem:translation-almost-existential-1} to the sentence $\foraneq y \psi$ and choose $n_0 \coloneqq c_{r+1} + r$. Now let $n \geq n_0$, $\Tt^\psi \in C_{n+r}(\psi)$ be almost existential with $\supp(\pi_{n+r} \llb \Tt^\psi \rrb) \subseteq X_n$, and $\Tt \in C_{n+r} (\foraneq y \psi)$ coincide with $\Tt^\psi$ for every instantiation of $y$. Transform $\Tt$ into $\Tt_{r+1}$ according to \cref{lem:translation-almost-existential-1}. We have $\supp(\pi_{n+r} \llb \Tt_{r+1} \rrb) \subseteq \supp(\pi_{n+r} \llb \Tt \rrb) = \supp(\pi_{n+r} \llb \Tt^\psi \rrb)$, so it suffices to provide $\Tt^* \in C_n(\psi)$ with $\supp(\pi_n \llb \Tt^{*} \rrb) \subseteq \supp(\pi_{n+r} \llb \Tt_{r+1} \rrb)$. We have $|A_\exists (\Tt_{r+1})| \leq c_{r+1} \leq n -r$, so there must be $i_1, \dots, i_r \in [n+r]$ such that $i_j \not\in A_\exists (\Tt_{r+1})$ for $j \leq r$.
    Now fix some substrategy $\Tt_{r+1}^\psi$ for $\psi$ (and an arbitrary instantiation of $y$) and consider the strategy $\Tt^*$, which arises from $\Tt_{r+1}^\psi[n+1 \leftrightarrow i_1, \dots, n+r \leftrightarrow i_r]$ 
    by deleting all subtrees $\Tt(w_j)$ where $w_j$ is successor of a universal node and $\lambda_r (w_j) = \phi (\bar a, n+j)$ for some $\phi(\bar x, y) \in \FO^{\neq}$ and $\bar a \subseteq [n+r]$. Because $i_j \not\in A_\exists (\Tt_{r+1})$, $n+j$ cannot occur in $\Tt^*$ anymore, so indeed, $\Tt^* \in C_n(\psi)$. Further, $\supp(\pi_n \llb \Tt^{*} \rrb) \subseteq \supp(\pi_{n+r} \llb \Tt \rrb)$ because $n+j$ did not occur in any leaf of $\Tt_{r+1}$ as $\supp(\pi_{n+r} \llb \Tt_{r+1} \rrb) \subseteq X_n$, and since $i_j\not\in A_\exists (\Tt_{r+1})$, $i_j$ can only occur in a leaf of $\Tt_{r+1}$ which is on a path containing $v,w$ with $\lambda_r (v)=\foraneq y \phi (\bar a, y)$ and $\lambda_r (w)= \phi (\bar a, i_j)$, so this leaf cannot occur in $\Tt^*$. Hence, neither $n+j$ nor $i_j$ occur in a leaf of $\Tt^*$, which ensures $\supp(\pi_n \llb \Tt^{*} \rrb) \subseteq \supp(\pi_{n+r} \llb \Tt \rrb)$.
\end{proof}

\begin{theorem} \label{cor:strategies-relying-on-forall-dont-matter}
	Let $\psi \in \FO^{\neq}$ be finite extension preserved in $\Fuzzy$ such that $\psi \not\equiv_\Fuzzy^{\geq m} \bot$ for all $m \in \Nat$.
    There exists an $n_0 \in \bbN$ such that for all $n \geq n_0$ and all $\Fuzzy$-interpretations $\pi$ with universe $[n]$, there is an almost existential strategy optimal for $\pi$ and $\psi$.
\end{theorem}

\begin{proof}
    Because $\psi$ is finite extension preserved in $\Fuzzy$ and $\psi \not\equiv_\Fuzzy^{\geq m} \bot$ for all $m \in \Nat$, there must be some $m_0$ such that $\psi \not\equiv_\Fuzzy^{m} \bot$ for all $m\geq m_0$.
	Let $n_0\geq m_0$ be large enough to apply \cref{lem:translation-almost-existential}.
	
	Now suppose that there was an $\Semi$-interpretation $\pi$ with universe $[n]$ where $n \geq n_0$ such that all strategies optimal for $\pi$ and $\psi$ rely on $\forall$. 
    Without loss of generality, we can assume that $\pi \llb \psi \rrb \neq 0$: In the case $\pi \llb \psi \rrb = 0$ every $\Tt \in C_n (\psi)$ is optimal for $\pi$ and $\psi$, so there cannot be an almost existential $\Tt \in C_n(\psi)$ at all and we can replace $\pi$ with some $\Semi$-interpretation that does not evaluate $\psi$ to~$0$. This is possible as we have chosen $n$ to be large enough to ensure $\psi \not\equiv_\Fuzzy^{m} \bot$.

    Hence, we can fix some $0<\epsilon< h_n(\pi_n \llb \psi \rrb)$ and define an extension $\pi_{\text{ext}}$ of $\pi$, which has universe $[n+r]$. We set $\pi_{\text{ext}} (\alpha) = \epsilon$ and $\pi_{\text{ext}} (\alpha)=0$ for each $\alpha \in \operatorname{Atoms}_{[n+r]\setminus[n]} (\tau)$.
    Now let $\Tt \in C_{n+r} (\psi)$ be optimal for $\pi_{\text{ext}}$.
    If $\supp(\pi_{n+r} \llb \Tt \rrb)\not\subseteq X_n$, we would have a contradiction to finite extension preservation of $\psi$ as $\pi_{\text{ext}} \llb \psi \rrb = \pi_{\text{ext}} \llb \Tt \rrb \leq \epsilon < \pi \llb \psi \rrb$.
	So suppose that $\supp(\pi_{n+r} \llb \Tt \rrb)\subseteq X_n$. In particular, this excludes the case that $\Tt$ relies on $\forall$ because for every node $v$ of $\Tt$ with $\lambda(v) = \foraneq y \phi (\bar a, y)$ there must be $w_i \in vE$ such that $\lambda(w_i) = \foraneq y \phi (\bar a, n+i)$ for some $i \in [r]$ due to $|\bar a|<r$. Hence, $\supp(\pi_{n+r} \llb \Tt \rrb)\subseteq X_n$ implies that $n+i \not\in A_{\Lit} (\Tt(w_i))$, which means that $\Tt$ is almost existential. But then we can apply \cref{lem:translation-almost-existential} and obtain an almost existential $\Tt^{*} \in C_n (\psi)$ such that $\supp (\pi_n \llb \Tt^{*} \rrb) \subseteq \supp (\pi_{n+r} \llb \Tt \rrb)$. Because multiplication is idempotent in $\Fuzzy$, this implies $\pi_{\text{ext}} \llb \psi \rrb = \pi_{\text{ext}}\llb \Tt \rrb \leq \pi \llb \Tt^{*} \rrb < \pi \llb \psi \rrb$, where that last inequality is because there is no almost existential strategy optimal for $\pi$, a contradiction to finite extension preservation of $\psi$.
\end{proof}

\subsection{Rewriting the Formula}

We want to apply \cref{cor:strategies-relying-on-forall-dont-matter} to eliminate subformulae $\foraneq y \phi (\bar x, y)$ of $\psi$ and eventually obtain an $\Fuzzy$-equivalent existential sentence. Unlike in the Viterbi case, where we could prove the existence of existential optimal strategies, causing all universal subformulae to be redundant, an additional argument is needed in the case where all optimal strategies are only almost existential and still use a subformula $\foraneq y \phi (\bar x, y)$. We cope with that by eliminating innermost subformulae $\foraneq y \phi (\bar x, y)$ first, and by translating $\foraneq y \phi (\bar x, y)$, based on a prenex-DNF, into an $\Fuzzy$-equivalent formula where the universal quantifier only admits strategies relying on $\forall$.

\begin{lemma} \label{lem:lattice-finitely-cont}
	Every lattice semiring $\Semi$ is \emph{finitely continuous}, i.e., for any $n \in \omega$ and $(s_i)_{i \leq n}, t \subseteq \Semi$ it holds that $\bigsqcap_{i \leq n} (t \sqcup s_i) = t \sqcup \bigsqcap_{i \leq n} s_i$.
\end{lemma}

\begin{proof}
	We only consider the case $n=1$, the rest follows by induction. Distributivity together with multiplicative idempotence and absorption yield
	\begin{align*}
		(t \sqcup s_0) \sqcap (t \sqcup s_1) &= t  \sqcup \underbrace{(t \sqcap s_1)}_{\leq t} \sqcup \underbrace{(s_0 \sqcap t)}_{\leq t} \sqcup (s_0 \sqcap s_1) = t \sqcup (s_0 \sqcap s_1). \qedhere
	\end{align*}
\end{proof}

\begin{lemma} \label{lem:remove-subform-not-cont-y}
	Let $\phi(\bar x) \in \FO^{\neq}$ be a formula that does not contain the variable $y$. Then $\foraneq y (\phi (\bar x) \vee \vartheta(\bar x, y)) \equiv_\Fuzzy^\omega \phi(\bar x) \vee \foraneq y \, \theta(\bar x, y) $.
\end{lemma}

\begin{proof}
	Let $\pi$ be a finite $\Fuzzy$-interpretation of universe $A$ and $\bar a \subseteq A$. We have
	\begin{align*}
		\pi \llb \foraneq y (\phi (\bar a) \vee \vartheta(\bar a, y)) \rrb &= \bigsqcap_{b \not\in \bar a} \bigl( \pi \llb \phi(\bar a) \rrb \operatorname{\sqcup} \pi \llb \theta(\bar a, b) \rrb  \bigr) \\ &\overset{*}= \pi \llb \phi (\bar a) \rrb \sqcup \bigsqcap_{b \not\in \bar a} \pi \llb \theta(\bar a, b) \rrb = \pi \llb \phi(\bar a) \vee \foraneq y \, \theta(\bar a, y) \rrb,
	\end{align*}
	where $(*)$ is by  \cref{lem:lattice-finitely-cont}.
\end{proof}

\begin{lemma}\label{lem:existential-pnf-dnf}
	Let $\Semi$ be a lattice semiring.
	Every formula $\phi (\bar x) \in \FO^{\neq}$ that does not contain a universal quantifier is $\Semi$-equivalent to a formula $\exneq \bar z \bigvee_i \theta_i(\bar x, \bar z)$, where each $\theta_{i} (\bar x, \bar z)$ is a conjunction of literals.
\end{lemma}

\begin{proof}
	Using distributivity, the claim readily follows by structural induction on $\phi (\bar x)$. Nothing remains to be shown in the cases where $\phi (\bar x)$ is a literal (we permit the quantifier prefix to be empty) or of the form $\exneq y \chi (\bar x, y)$. 
	
	So let $\phi (\bar x) =( \psi \circ \theta) (\bar x)$ where $\circ \in \{\vee, \wedge\}$. By IH, there are formulae $\exneq \bar z \bigvee_i \psi_i^* (\bar x, \bar z)$ and $\exneq \bar u \bigvee_j \theta_j^* (\bar x, \bar u) $  such that $\psi (\bar x) \equiv_\Semi \bigvee_{i} \exneq \bar z \psi_{i}^* (\bar x, \bar z)$ and $\theta (\bar x) \equiv_\Semi \exneq \bar u \bigvee_{j} \theta_{j}^* (\bar x, \bar u)$. 
	Rename variables if needed to make sure that the entries of $\bar z$ and $\bar u$ are pairwise distinct.
	In the case $\circ=\vee$, we get $\phi (\bar x) \equiv_\Semi \exneq \bar z \exneq \bar u \bigl( \bigvee_i \psi_i^* (\bar x, \bar z) \vee \bigvee \{\theta_j^* (\bar x, \bar v) \mid \bar v \subseteq \bar u \cup \bar z \text { pairwise distinct, } j\}\bigr)$. Note that because we now nest the existential quantifiers, we have to take into account instantiations of $\bar z$ for $\theta_j^*$ too. 
	If $\circ=\wedge$, distributivity yields
	\begin{align*}
		\phi (\bar x) &\equiv_\Semi \exneq \bar z \exneq \bar u \bigl( \bigvee_i \psi_i^* (\bar x, \bar z) \wedge \bigvee \{\theta_j^* (\bar x, \bar v) \mid \bar v \subseteq \bar u \cup \bar z \text { pairwise distinct, } j\}\bigr) \\
		&\equiv_\Semi \exneq \bar z \exneq \bar u \bigl( \bigvee \{ (\bigvee_i \psi_i^* (\bar x, \bar z)) \wedge \theta_j^* (\bar x, \bar v) \mid \bar v \subseteq \bar u \cup \bar z \text { pairwise distinct, } j\}\bigr) \\
		&\equiv_\Semi \exneq \bar z \exneq \bar u \bigl( \bigvee \{ \psi_i^* (\bar x, \bar z) \wedge \theta_j^* (\bar x, \bar v) \mid \bar v \subseteq \bar u \cup \bar z \text { pairwise distinct, } j, i\}\bigr). \qedhere
	\end{align*}
\end{proof}

Repeated application of the following \cref{lem:remove-universal-subform} finally proves \cref{thm:extension-preservation-fuzzy}.

\begin{lemma} \label{lem:remove-universal-subform}
	For every $\psi \in \FO^{\neq}$ that is extension preserved on finite $\Fuzzy$-interpretations there is some existential $\psi^* \in \FO^{\neq}$ such that $\psi \equiv_\Fuzzy^\omega \psi^*$ which contains less universal quantifiers than $\psi$.
\end{lemma}

\begin{proof}
	Fix a subformula $\foraneq y \phi (\bar x, y)$ of $\psi$ where $\qr_\forall (\phi (x,y)) = 0$.
    Let $\phi (\bar x, y) \equiv_\Fuzzy^\omega \exneq \bar z \bigvee_i \theta_i(\bar x, y, \bar z)$ such that each $\theta_i (\bar x, y, \bar z)$ is a conjunction of literals. Let \[
    \psi (\bar x, y, \bar z) \coloneqq \bigvee \{\theta_i (\bar x,y,z) \mid i \text{ such that at least one literal in } \theta_i (\bar x,y,z) \text{ contains } y\} \text{ and}
    \]
  \[
  \chi (\bar x, \bar z) \coloneqq \bigvee \{\theta_i (\bar x,y,z) \mid i \text{ such that no literal of } \theta_i (\bar x,y,z) \text{ contains } y\}.
  \]
  We allow both disjunctions to be empty (in this case we simply have $\top$). By associativity of addition, we have $\phi (\bar x, y) \equiv_\Fuzzy^\omega \exneq \bar z \chi (\bar x, \bar z) \vee \exneq \bar z \psi (\bar x, y, \bar z)$. Now we can apply \cref{lem:remove-subform-not-cont-y} and get $\foraneq y \phi (\bar x, y) \equiv_\Fuzzy^\omega \exneq \bar z \chi (\bar x, \bar z) \vee \foraneq y \exneq \bar z \psi (\bar x, y, \bar z)$. 
    In particular, this implies $\psi \equiv_\Fuzzy^\omega \psi [\foraneq y \phi (\bar x, y) / \exneq \bar z \chi (\bar x, \bar z) \vee \foraneq y \exneq \bar z \psi (\bar x, y, \bar z)]$.
    
  	Let $n$ be large enough to apply \cref{cor:strategies-relying-on-forall-dont-matter} to $\psi [\foraneq y \phi (\bar x, y) / \exneq \bar z \chi (\bar x, \bar z) \vee \foraneq y \exneq \bar z \psi (\bar x, y, \bar z)]$. Now since every disjunct of $\psi (\bar x, y, \bar z)$ contains $y$, every strategy using $\foraneq y \exneq \bar z \psi (\bar a, y, \bar z)$ for some $\bar a$ relies on $\forall$. By \cref{cor:strategies-relying-on-forall-dont-matter}, we know that for every $\Fuzzy$-interpretation $\pi$ of cardinality at least $n$ there must be an almost existential optimal strategy, i.e. a strategy that does not use $\foraneq y \exneq \bar z \psi (\bar x, y, \bar z)$.
    This implies $\psi \equiv_\Fuzzy^{\geq n} \psi [\foraneq y \phi (\bar x, y) / \exneq \bar z \chi (\bar x, \bar z) \vee \bot] \equiv_\Fuzzy \psi [\foraneq y \phi (\bar x, y) / \exneq \bar z \chi (\bar x, \bar z)]$, which eliminates the universal subformula. Now we can apply the same construction as in \cref{lem:large-universes-suffice} to obtain a sentence $\psi^*$ that is equivalent to $\psi$ not just on $\Fuzzy$-interpretations of cardinality at least $n$, but on all finite $\Fuzzy$-interpretations. This does not introduce new universal quantifiers, so $\psi^*$ contains one universal quantifier less than~$\psi$.
\end{proof}

Based on \cref{thm:extension-preservation-fuzzy}, we can now generalise the fuzzy case to all lattice semirings based on the reduction method elaborated in \cref{sec:lattices}.

\FiniteExtPresLattices*

\begin{proof}
    The implication from right to left is covered by \cref{lem:syntax-to-semantics}(2). So let $\psi \in \FO$ be extension preserved on finite $\calL$-interpretations. By \cref{prop:preservation-in-all-lattices}, $\psi$ is also extension preserved on finite $\Fuzzy$-interpretations, and this must also hold for $\psi^* \in \FO^{\neq}$ with $\psi^* \equiv_\Fuzzy \psi$ we obtain from \cref{lem:FOneq}. By \cref{thm:extension-preservation-fuzzy}, there is some $\phi^* \in \FO^{\neq}$ with $\phi^* \equiv_\Fuzzy^\omega \phi$. Another application of \cref{lem:FOneq} yields $\phi \in \FO$ with $\phi \equiv_\Fuzzy \phi^* \equiv_\Fuzzy^\omega \psi$ and, by \cref{cor:equivalence-from-S3-to-any-lattice}, this eventually implies $\phi \equiv_\calL^\omega \psi$.
\end{proof}

\section{Conclusion}

In this paper, we have provided a framework for studying the interplay between syntactic forms and semantic properties of first-order formulae evaluated on semirings. As
semiring semantics refines the classical Boolean semantics,
it provides more fine-grained model-theoretic notions. Preservation properties, for example, naturally generalise to this context based on natural order of the semirings. 
Relating them to specific syntactic forms allow us to answer questions such as: How can we syntactically characterise those queries for which we only gain confidence, or reduce costs,
if we extend our data?

As a first result, we have seen that preservation theorems do not always generalise to the semiring setting as, for instance, the extension preservation theorem fails for the tropical semiring, the Viterbi semiring, the {\L}ukasiewicz semiring, and the natural semirings $\Nat$ and $\Nat^\infty$.
However, on the fairly broad class of lattice semirings we manage to recover the classical preservation theorems for homomorphisms, extensions and subinterpretations.
The proof combines adaptations of classical methods from model theory, based on compactness and amalgamation, with a reduction technique specifically designed for semiring semantics. Results with a somewhat similar flavour have been
studied in the context of certain many-valued logics \cite{BadiaCosDelNog19,Carr26}. However, our framework 
differs from these studies in multiple ways: Our notions of preservation and logical equivalence are much more fine-grained
as they do not just distinguish valuations of $1$ from those taking different values. Moreover, in our setting, elements of the algebraic structure in which we evaluate the formulae do not appear as constants in the logic (which has several relevant consequences as existential sentences can no longer describe the finite subinterpretations of a given interpretation), and we obtained positive results for arbitrary lattice semirings rather than just the finite linearly ordered ones.

Surprisingly, the extension preservation theorem, which fails in finite model theory under Boolean semantics, does hold for a number of semirings including the tropical semiring, the Viterbi semiring, the {\L}ukasiewicz semiring, and every lattice semiring, which is the second main contribution of this paper. 

The results presented here suggest two main directions for future research: towards other semirings, and other preservation theorems (both in the finite and general case). More specifically, the approach employed for lattices might be applicable to other preservation theorems, for example concerning chain preservation. 
Moreover we believe that the curious phenomenon of extension preservation in the finite deserves further investigation.
We suspect that a similar proof technique could be used to establish extension preservation, as well as further preservation theorems such as subinterpretation preservation, in the finite also other semirings.
For all we know, it is possible that the Boolean lattice is the \emph{only one} which does not satisfy extension and substructure preservation in the finite. Exploring this question further would shed a new light on preservation theorems in finite model theory.

\bibliography{provenance}

@incollection{AtseriasKol23,
  author       = {A.~Atserias and P.~Kolaitis},
  title        = {Acyclicity, consistency, and positive semirings},
  booktitle      = {Samson Abramsky on Logic and Structure in Computer Science and Beyond, volume 25 of Outstanding Contributions to Logic},
  editor   =     {A.~Palmigiano and M.~Sadrzadeh},
  pages        = {623-628},
  year         = {2023}
}

@article{AtseriasKol24,
  author       = {A.~Atserias and P.~Kolaitis},
  title        = {Consistency of Relations over Monoids},
  journal      = {Proc. {ACM} Manag. Data},
  volume       = {2},
  number       = {2},
  pages        = {107},
  year         = {2024},
  doi          = {10.1145/3651608}
  }

@article{BadiaKolNog25,
      title={Codd's Theorem for Databases over Semirings}, 
      author={G.~Badia and P.~Kolaitis and C.~Noguera},
      year={2025},
    journal      = {Proc. {ACM} Manag. Data},
    volume       = {3},
    number       = {5},
    pages        = {277:1--277:26},
    doi          = {10.1145/3767713}
}

@article{BadiaCosDelNog19,
  author       = {G.~Badia and
                  V.~Costa and
                  P.~Dellunde and
                  C.~ Noguera},
  title        = {Syntactic characterizations of classes of first-order structures in
                  mathematical fuzzy logic},
  journal      = {Soft Comput.},
  volume       = {23},
  number       = {7},
  pages        = {2177--2186},
  year         = {2019},
  doi          = {10.1007/S00500-019-03850-6}}

@InProceedings{BrinkeGraMrk24,
  author =	{Brinke, S. and Gr\"{a}del, E. and Mrkonji\'{c}, L.},
  title =	{{Ehrenfeucht-Fra\"{i}ss\'{e} Games in Semiring Semantics}},
  booktitle =	{32nd EACSL Annual Conference on Computer Science Logic (CSL 2024)},
  pages =	{19:1--19:22},
  series =	{Leibniz International Proceedings in Informatics (LIPIcs)},
  ISBN =	{978-3-95977-310-2},
  ISSN =	{1868-8969},
  year =	{2024},
  volume =	{288},
  URN =		{urn:nbn:de:0030-drops-196623},
  doi =		{10.4230/LIPIcs.CSL.2024.19}
}

@InProceedings{BrinkeDawGraMrkNaa26,
  author =	{Brinke, S. and Dawar, A. and Gr\"{a}del, E. and Mrkonji\'{c}, L. and Naaf, M.},
  title =	{{Compactness in Semiring Semantics}},
  booktitle =	{34th EACSL Annual Conference on Computer Science Logic (CSL 2026)},
  pages =	{13:1--13:21},
  series =	{Leibniz International Proceedings in Informatics (LIPIcs)},
  ISBN =	{978-3-95977-411-6},
  ISSN =	{1868-8969},
  year =	{2026},
  volume =	{363},
  URN =		{urn:nbn:de:0030-drops-254372},
  doi =		{10.4230/LIPIcs.CSL.2026.13}
}

@InProceedings{BiziereGraNaa23,
  author =	{Bizi\`{e}re, C. and Gr\"{a}del, E. and Naaf, M.},
  title =	{{Locality Theorems in Semiring Semantics}},
  booktitle =	{48th International Symposium on Mathematical Foundations of Computer Science (MFCS 2023)},
  pages =	{20:1--20:15},
  series =	{Leibniz International Proceedings in Informatics (LIPIcs)},
  ISBN =	{978-3-95977-292-1},
  ISSN =	{1868-8969},
  year =	{2023},
  volume =	{272},
  URN =		{urn:nbn:de:0030-drops-185546},
  doi =		{10.4230/LIPIcs.MFCS.2023.20}
}

@InProceedings{BrinkeGraMrkNaa24,
	author = {Brinke, S. and Gr{\"a}del, E. and Mrkonji\'c, L. and Naaf, M.},
	title =	{{Semiring Provenance in the Infinite}},
	booktitle =	{The Provenance of Elegance in Computation - Essays Dedicated to Val Tannen},
	pages =	{3:1--3:26},
	series =	{Open Access Series in Informatics (OASIcs)},
	year =	{2024},
	volume =	{119},
	publisher =	{Schloss Dagstuhl -- Leibniz-Zentrum f{\"u}r Informatik},
	address =	{Dagstuhl, Germany},
	doi =		{10.4230/OASIcs.Tannen.3}
}

@article{Carr26,
      title={Homomorphism Preservation Theorems for Many-Valued Structures}, 
      author={J.~Carr},
      journal = {ACM Transactions on Computational Logic},
      year={2026},
      doi = {10.1145/3793665}
}

@inproceedings{DawarSan21,
  author       = {A.~Dawar and A.~Sankaran},
  title        = {Extension Preservation in the Finite and Prefix Classes of First Order
                  Logic},
  booktitle    = {29th EACSL Annual Conference on Computer Science Logic (CSL 2021)},
 pages        = {18:1--18:13},
  year         = {2021},
  doi          = {10.4230/LIPICS.CSL.2021.18}
}

@article{DellundeVid19,
  author    = {P.~Dellunde and A.~Vidal},
  title     = {Truth-preservartion under fuzzy pp-formulas},
  journal   = {International Journal of Uncertainty, Fuzzyness, and Knowledge-Based Systems},
  volume    = {27},
  pages     = {89-105},
  year      = {2019}
}

@inproceedings{FosterGreTan08,
author = {Foster, J. Nathan and Green, Todd J. and Tannen, Val},
title = {Annotated XML: queries and provenance},
year = {2008},
doi = {10.1145/1376916.1376954},
booktitle = {Proceedings of the Twenty-Seventh ACM SIGMOD-SIGACT-SIGART Symposium on Principles of Database Systems},
pages = {271–280},
numpages = {10},
series = {PODS '08}
}

@article{Glavic21,
    author = "Glavic, B.",
    title = "Data Provenance",
    journal = "Foundations and Trends in Databases",
    volume = "9",
    number = "3-4",
    pages = "209--441",
    year = "2021",
    doi = "10.1561/1900000068"
}

@inproceedings{GraedelHelNaaWil22,
    author = {Gr{\"{a}}del, E. and Helal, H. and Naaf, M. and Wilke, R.},
    title = {{Zero-One Laws and Almost Sure Valuations of First-Order Logic in Semiring Semantics}},
    howpublished = "arXiv:2203.03425 [cs.LO]",
    year = "2022",
    url = "https://arxiv.org/abs/2203.03425",
    booktitle = "{LICS} '22: 37th Annual {ACM/IEEE} Symposium on Logic in Computer Science, Haifa, Israel, August 2 - 5, 2022",
    doi = "10.1145/3531130.3533358",
    pages = "41:1--41:12",
    publisher = "{ACM}"
}

@inproceedings{GraedelLucNaa21b,
    author = {Gr{\"a}del, E. and L{\"u}cking, N. and Naaf, M.},
    editor = "Ganty, P. and Bresolin, D.",
    title = {Semiring Provenance for {B{\"u}chi} Games: Strategy Analysis with Absorptive Polynomials},
    booktitle = "Proceedings 12th International Symposium on Games, Automata, Logics, and Formal Verification (GandALF 2021)",
    series = "{EPTCS}",
    volume = "346",
    pages = "67--82",
    year = "2021",
    doi = "10.4204/EPTCS.346.5"
}

@inproceedings{GraedelMrk21,
    author = {Gr\"{a}del, E. and Mrkonji\'{c}, L.},
    address = "Dagstuhl, Germany",
    booktitle = "48th International Colloquium on Automata, Languages, and Programming (ICALP 2021)",
    doi = "10.4230/LIPIcs.ICALP.2021.133",
    pages = "133:1--133:20",
    title = "Elementary Equivalence Versus Isomorphism in Semiring Semantics",
    volume = "198",
    year = "2021"
}

@incollection{GraedelTan24,
    author = {Gr{\"a}del, E. and Tannen, V.},
    booktitle = {Model Theory, Computer Science and Graph Polynomials. Festschrift for Johann A. Makowsky},
    title = "Provenance Analysis and Semiring Semantics for First-Order Logic",
    year = "2025",
    publisher = {Birkh\"auser},
    note = {See also arXiv:2412.07986},
    doi = {10.48550/ARXIV.2412.07986}
}

@article{GraedelTan20,
    author = {Gr{\"a}del, E. and Tannen, V.},
    title = "Provenance analysis for logic and games",
    journal = "Moscow Journal of Combinatorics and Number Theory",
    volume = "9",
    number = "3",
    pages = "203--228",
    year = "2020",
    publisher = "Mathematical Sciences Publishers",
    doi = "10.2140/moscow.2020.9.203",
    note = "Preprint available at \url{https://arxiv.org/abs/1907.08470}"
}

@inproceedings{GreenKarTan07,
author = {Green, T. and Karvounarakis, G. and Tannen, V.},
title = {Provenance semirings},
year = {2007},
publisher = {Association for Computing Machinery},
doi = {10.1145/1265530.1265535},
booktitle = {Proceedings of the Twenty-Sixth ACM SIGMOD-SIGACT-SIGART Symposium on Principles of Database Systems},
pages = {31–40},
numpages = {10},
series = {PODS '07}
}

@inproceedings{GreenTan17,
author = {Green, T. and Tannen, V.},
title = {{The Semiring Framework for Database Provenance}},
year = {2017},
publisher = {Association for Computing Machinery},
doi = {10.1145/3034786.3056125},
booktitle = {Proceedings of the 36th ACM SIGMOD-SIGACT-SIGAI Symposium on Principles of Database Systems},
pages = {93–99},
numpages = {7},
series = {PODS '17}
}

@incollection{Gurevich84,
	author = {Y.~Gurevich},
	booktitle = {Computation and Proof Theory},
	pages = {175-216},
	publisher = {Springer Lecture Notes in Mathematics},
	title = {Towards logic tailored for computational complexity},
	year = {1984}
}

@phdthesis{Mrkonjic25,
	author = "Mrkonji{\'c}, L.",
	school = "{RWTH} Aachen University",
	title = "Semiring Semantics: Algebraic Foundations, Model Theory, and Strategy Analysis",
	type = "{PhD} Thesis",
	year = "2025",
    url={https://publications.rwth-aachen.de/record/1010223}
}

@phdthesis{Naaf24,
    author = "Naaf, M.",
    school = "{RWTH} Aachen University",
    title = "Logic, Semirings, and Fixed Points",
    type = "{PhD} Thesis",
    year = "2024",
    url={https://publications.rwth-aachen.de/record/996756}
}

@article{Rosen02,
    author = "E.~Rosen",
    title = "Some aspects of model theory on finite structures",
    journal = "Bulletin of Symbolic Logic",
    volume = "8",
    pages = "380-403",
    year = "2002"
}

@article{Rossman08,
author = {Rossman, Benjamin},
title = {Homomorphism preservation theorems},
year = {2008},
issue_date = {July 2008},
publisher = {Association for Computing Machinery},
volume = {55},
number = {3},
doi = {10.1145/1379759.1379763},
journal = {J. ACM},
articleno = {15},
numpages = {53},
}

@article{Tait59,
    author = "W.~Tait",
    title = {{A counterexample to a conjecture of Scott and Suppes}},
    journal = "Journal of Symbolic Logic",
    volume = "24",
    pages = "15-16",
    year = "1959"
}

@phdthesis{Lopez-thesis,
  author       = {A. Lopez},
  title        = {First Order Preservation Theorems in Finite Model Theory : Locality,
                  Topology, and Limit Constructions},
  school       = {University of Paris-Saclay, France},
  year         = {2023},
  url          = {https://tel.archives-ouvertes.fr/tel-04534568}
}

\end{document}